%% file: LDGLDA.tex
\documentclass[10pt]{scrartcl}

\usepackage{amsmath}
\usepackage{amssymb}
\usepackage{amsthm}
\usepackage{tensor}
\usepackage{mathtools}
\usepackage{multirow}
\usepackage{enumerate}
\usepackage{tikz-cd}
\usepackage{graphicx}
\usepackage{xcolor}

\usepackage[numbers,sort&compress]{natbib}

\usepackage{hyperref}

\usepackage[a4paper,textwidth=170mm]{geometry} 

\input{symbols.tex}

\makeatletter
\patchcmd{\math@cr@@@align}{\cr}{\global\let\df@label\@empty\cr}{}{}
\makeatother

\newtheorem{theorem}{Theorem}
\newtheorem{lemma}[theorem]{Lemma}
\newtheorem{corollary}[theorem]{Corollary}

\newtheorem{remark}[theorem]{Remark}

\title{Beris-Edwards Models on Evolving Surfaces:\\A Lagrange-D'Alembert Approach}
\author{Ingo Nitschke$^1$ and Axel Voigt$^{1,2,3}$ \\ 
\hspace*{-0.8cm} \large{$^1$ Institute of Scientific Computing, Technische Universität Dresden, 01062 Dresden, Germany} \\
\hspace*{-0.8cm} \large{$^2$ Dresden Center for Computational Materials Science (DCMS), Technische Universität Dresden,} \\
\large{01062 Dresden, Germany} \\
\hspace*{-0.8cm} \large{$^3$ Center for Systems Biology Dresden (CSBD), Pfotenhauerstr. 108, 01307 Dresden, Germany}}

\begin{document}

\maketitle

\begin{abstract}
    Using the Lagrange-D'Alembert principle we develop thermodynamically consistent surface Beris-Edwards models. These models couple viscous inextensible surface flow with a Landau-de Gennes-Helfrich energy and consider the simultaneous relaxation of the surface Q-tensor field and the surface, by taking hydrodynamics of the surface into account. We consider different formulations, a general model with three-dimensional surface Q-tensor dynamics and possible constraints incorporated by Lagrange multipliers and a surface conforming model with tangential anchoring of the surface Q-tensor field and possible additional constraints. In addition to different treatments of the surface Q-tensor, which introduces different coupling mechanisms with the geometric properties of the surface, we also consider different time derivatives to account for different physical interpretations of surface nematics. We relate the derived models to established models in simplified situations, compare the different formulations with respect to numerical realizations and mention potential applications in biology. 
\end{abstract}

\section{Introduction}

Nematic liquid crystals are fluids whose anisotropic molecules have long range orientational order, but no positional order. Several theories exist which attempt to capture the complexity of these materials. We here consider the Q-tensor description, proposed by de Gennes \cite{de_Gennes_book}. A Q-tensor is a symmetric, traceless $d \times d$ tensor, with $d$ the space dimension, which, if it is uniaxial, can be expressed as $ \mathbf{Q} = s (\mathbf{d} \otimes \mathbf{d} - \frac{1}{d} \Id) $, with the scalar order parameter $ s $, the director $ \mathbf{d} $, describing the direction of orientation and $ \Id $ the identity-matrix. For a comprehensive discussion on such Landau-de Gennes (or Q-tensor) models we refer, for instance to \cite{sonnet2004continuum,majumdar2010landau,ball2017mathematics,BORTHAGARAY2021313}. We are concerned with the Beris-Edwards models \cite{Beris_1994} for nematodynamics. These models couple the Landau-de Gennes models with the incompressible Navier-Stokes equations. Due to these couplings, forces in the fluid not only influence its velocity field, as in ordinary fluids, they also cause rotations and stretching of the Q-tensor field. But also the Q-tensor field influences the velocity field, e.g. by anisotropic nematic viscosity, which results in flow aligning with the Q-tensor field. For theoretical and computational results for the Beris-Edwards models for $ d = 2,3 $ we refer for instance to \cite{paicu2011energy,paicu2012energy,Abels_SIAMJMA_2014,Abels_2016}. 

A specific topic of interest are nematic shells, which are rigid particles coated with a thin film of nematic liquid crystal whose molecular orientation is subjected to a tangential anchoring. This introduces topological constraints and depending on the topology of the surface forces orientational defects to be present. Most theoretical approaches are concerned with the equilibrium configurations emerging from the non-trivial interplay between the geometry and the topology of the fixed surface and the tangential anchoring constraint \cite{Shin_PRL_2008}. Surface Landau-de Gennes models can be derived by a thin-film limit of the above mentioned three-dimensional Landau-de Gennes models \cite{Golovaty_JNS_2015,Golovaty_JNS_2017,Nitschke_2018,Nestler_2020}. In \cite{Nestler_2020} it is shown that the definition of the Q-tensor on the surface has strong implications on the physical behaviour. In the most general setting the above definition of the Q-tensor for $ d = 3 $ is only restricted to the surface. Considering in addition the tangential anchoring leads to surface conforming models. These approaches allow for a decomposition of $\mathbf{Q}$ into a tangential Q-tensor part $ \mathbf{q} $ and an eigenvalue $\beta$ in normal direction of the surface. Considering $ \beta = 0 $ essentially leads to liquid crystal properties known for $ d = 2 $ in flat space, whereas $ \beta \neq 0 $ maintains three-dimensional properties in the surface models. Also the boundary conditions specified in the thin film limit matter. They decide about possible extrinsic curvature contributions in the derived surface models  \cite{Napoli_PRE_2012,Napoli_PRL_2012,Nitschke_2018,Novack_SIAMJMA_2018}. Such contributions lead to alignment of the director $\mathbf{d}$ with principle curvature directions of the surface. Much less is known about the dynamics. In \cite{BouckNochettoYushutin_2022} a surface Beris-Edwards model is derived from a generalized Onsager principle. It combines a tangent inextensible flow, modeled by the surface Navier-Stokes equations, with a three-dimensional Q-tensor field, which is not anchored to the tangent plane of the surface. Corresponding approaches which consider the tangential anchoring of the Q-tensor can be found in \cite{Pearce_PRL_2019,Nestler_2022}. In this setting the already mentioned coupling between the flow and the Q-tensor field also depends on the geometry and the topology of the surface. This coupling is not only due to the Q-tensor field, also the tangential surface Navier-Stokes equations contain additional geometric coupling terms \cite{Reuther_2015,Koba_2017,Reuther_MMS_2018,Koba_2018,Miura_2018,Jankuhn_2018,Reuther_PF_2018}. 

In this paper we relax the assumption of a stationary surface and derive thermodynamically consistent models of a thin film of nematic liquid crystal on a deformable surface under the influence of bending forces. This leads to surface Beris-Edwards models on evolving surfaces. In addition to the various couplings between geometry and topology with the Q-tensor field and the velocity field mentioned for stationary surfaces, new couplings emerge for the evolution of the surface. Within the literature two extreme cases are considered. The first is for isotropic fluids leading to so-called fluid deformable surfaces \cite{Arroyo_2009,Torres-Sanchez_2019,Reuther_2020,Krause_2023}. These models combine the surface (Navier-)Stokes equations in tangential and normal direction \cite{Reuther_2015,Koba_2017,Reuther_MMS_2018,Koba_2018,Miura_2018,Jankuhn_2018} with bending forces resulting from Helfrich/Willmore energies \cite{Helfrich}. As a result of the tight interplay between geometry and flow field, in the presence of curvature, any shape change is accompanied by a tangential flow and, vice-versa, the surface deforms due to tangential flow. The other extreme case is the situation without flow. In \cite{Nitschke_2020} a $L^2$-gradient flow is considered for a surface Landau-de Gennes-Helfrich energy which allows for simultaneous relaxation of the Q-tensor and the surface. Similar to the situation for nematic shells, the definition of the Q-tensor and the presence of extrinsic curvature terms have strong implications on the emerging equilibrium state. While purely intrinsic models lead to a tetrahedral defect arrangement \cite{Park_EPL_1992}, taking extrinsic curvature contributions and $\beta \neq 0$ into account leads to asymmetric shapes with a planar defect arrangement \cite{Nitschke_2020}. This difference can be attributed to the tendency of the director field to align with principle curvature directions. Both, the geometric coupling terms leading to these qualitative changes for the equilibrium shape, as well as the tight interplay between geometry and flow field are also present in the surface Beris-Edwards models. However, considering the dynamics of vector- and tensor-valued quantities on evolving surfaces offers additional degrees of freedom. In contrast with surface scalar quantities, for which essentially only one time derivative is suitable for physically relevant problems, surface vector- and tensor-valued quantities offer the possibility of a versatile selection of appropriate time derivatives \cite{NitschkeVoigt_JoGaP_2022,NitschkeVoigt_2023}. In the context of gradient flows these different time derivatives require a consistent choice of the gauge of surface independence to obtain thermodynamically consistent models \cite{NitschkeSadikVoigt_A_2022}. In our context, specifying a certain time derivative for the immobility mechanism is comparable to a consistent choice in such flows. As demonstrated in \cite{NitschkeSadikVoigt_A_2022}, even if chosen consistently the evolution differs for different choices, eventually leading to different local configurations. Also this has to be accounted for by addressing surface Beris-Edwards models.

While such models might also be of interest in materials science, the most demanding applications are in biology. Elongated cells and filaments can be understood in the framework of active nematic liquid crystals. Their appearance as thin films has already been used to described multiscale biological problems in morphogenesis. A striking example is the freshwater polyp Hydra, where topological defects in the muscle fiber orientation have been shown to localize to key features of the body plan \cite{maroudas2021topological}, which has sparked growing interest in nematic liquid crystals on curved and deforming surfaces \cite{hoffmann2022theory,vafa2022active,wang2023patterning}. However, these models take the form of a minimal model. They use a conforming formulation and typically neglect dynamics and extrinsic curvature contributions. Models which account for fluid flow in this context are typically restricted to axisymmetric settings \cite{mietke2019self} or stationary surfaces \cite{Pearce_PRL_2019,Nestler_2022}. Non of these approaches discusses the choice of the time derivative or explores the consequences of the model assumptions discussed above in a systematic manner. Also theoretical work in these directions \cite{Juelicher_2018,salbreux2022theory,al2023morphodynamics} addresses these issue only partly.

The purpose of this paper is twofold. First, we derive surface Beris-Edwards models by the Lagrange-D'Alembert principle, which naturally leads to a thermodynamically consistent formulation. While only considering a one-constant approximation all other terms are considered in a relatively general way. Second, we provide two alternative formulations for these models. Next to the conforming formulation also an approach which deals with three-dimensional Q-tensor fields $ \mathbf{Q} $, restricted to the surface, and incorporating additional constraints, e.g. on conformity, by Lagrange multipliers, is considered. Based on our experience to numerically solve fluid deformable surfaces \cite{Krause_2023,Krause_PAMM_2023,BachiniKrauseNitschkeVoigt_2023} we expect this formulation to be computationally more manageable. Any active component, required for the applications discussed above are neglected. Their incorporation will be the subject of future research.

The paper is organized as follows. We provide a brief summary of the mathematical notation and calculus in Sec. \ref{sec:notation}. 
Most of it follows \cite{NitschkeSadikVoigt_A_2022,NitschkeVoigt_2023}, or extends upon it.
In Sec. \ref{sec:models} we provide a summary of the various Surface Beris-Edwards models considered and relate them to known results in the literature. The most general models can be found in Sec. \ref{sec:models_general}.
In that section we focus on bendable, nematic viscous, inextensible, hydrodynamic, biaxial Landau-de Gennes surface models. 
These models allow for the incorporation of various constraints through Lagrange parameters and their associated generalized constraint forces. Nematics are represented by Q-tensor fields defined on the surface.
Furthermore, these models can be distinguished by their prescribed nematic immobility, leading to both material and Jaumann models. In Sec. \ref{sec:model_general_cases} we consider specific cases to establish connections with existing models.
In Sec. \ref{sec:models_conforming}, we present Surface Conforming Beris-Edwards models where one eigenvector of the Q-tensor fields is constrained to align with the normal field. 
Specific cases in this context and relations with existing models are provided in Sec. \ref{sec:surface_conforming_special_cases}.
The detailed derivation of the Surface (Conforming) Beris-Edwards models is presented in Sec. \ref{sec:derivation}. 
Our derivation is based on the Lagrange-D'Alembert principle, which we introduce for our specific context in Sec. \ref{sec:lagrangedalembert}. 
In this section, we derive local equations that depend on generalized forces obtained through $ \hil $-variations. 
Additionally, we elucidate the influence of gauges of surface independence, see \cite{NitschkeSadikVoigt_A_2022}, on these forces and clarify why the resulting models are independent of the choice of these gauges.
We categorize the derived forces into Lagrangian forces (Sec. \ref{sec:lagrangian_forces}), which represent purely state forces, and flux forces (Sec. \ref{sec:flux_forces}), which are process forces capable of inducing energy exchange out of the considered open system. 
In Secs. \ref{sec:elastic} and \ref{sec:thermotropic}, we derive the elastic and thermotropic forces based on a one-constant surface Landau-de Gennes energy. Sec. \ref{sec:bending} elucidates the origin of bending forces from the Helfrich/Willmore energy.
The imposition of constraints on state variables using the Lagrange multiplier technique results in constraint forces, which are also categorized as Lagrangian forces. 
In Sec. \ref{sec:surface_conforming_constrain}, we focus on surface conformity by prohibiting mixed tangential-normal components of the Q-tensor field.
The impact of this constraint on the other forces is addressed separately in the respective subsections.
In Sec. \ref{sec:uniaxiality_constrain}, we introduce a Q-tensorial biaxiality quantity that shares the same kernel as the biaxial measurement, ultimately leading to an uniaxial constraint.
Sec. \ref{sec:isotropic_state} addresses the trivial case of an isotropic state.
In Sec. \ref{sec:immobility}, immobility forces emerge as a consequence of selecting an immobility flux potential that globally measures a nematic rate. We restrict ourselves to determining this rate using the material and Jaumann/corotational time derivative. It's worth noting that these immobilities can be seen as an overdamped consideration of the actual angular momentum of the local nematics. Choosing the material rate introduces inertia effects concerning the Euclidean embedding space, whereas the Jaumann rate affects inertia within the surface, resulting in rigid body motion invariance.
In Sec. \ref{sec:nematic_viscosity}, we contemplate an anisotropic nematic metric. Its temporal distortion gives rise to nematic viscous forces. This extends the isotropic special case that results in the common viscous force, which comprises the divergence of the surface strain rate tensor field. Constraint forces resulting from process variable restrictions are also treated as flux forces.
In Sec. \ref{sec:incompressibility}, we derive the inextensibility force to maintain a temporally constant density under mass conservation, or equivalently, local surface area conservation. Interpreting the associated Lagrange parameter as pressure, this force corresponds to the usual pressure gradient. The inextensibility constraint is the only one in this paper that we require as mandatory in the final models. However, all other forces are derived without this constraint and can also be used in compressible models.
Sec. \ref{sec:no_normal_flow} introduces the no-normal-flow force by stipulating a geometrically stationary surface, allowing only tangential material flows. We do not address its counterpart, i.e., material flows solely in the normal direction. For related models in this context, but without the inextensibility constraint, we refer to \cite{Nitschke_2020,NitschkeSadikVoigt_A_2022}. Sec. \ref{sec:no_flow} deals with the trivial dry case, where fluid motion is prohibited. We believe that the Lagrange-d'Alembert principle, when correctly applied, always yields thermodynamically consistent models. However, due to the lack of a general proof, we address the energy rate specifically for the derived Surface Beris-Edwards models in Sec. \ref{sec:energy_rate}.
In this section, we demonstrate that this rate can be exclusively described by the prescribed energy flux potentials, which govern the energy transfer between the open system and its surroundings. 
Furthermore, it is shown that the dynamics are entirely dissipative, in accordance with the thermodynamic expectations of an open passive system. In Sec. \ref{sec:phys} we consider physical implications, in particular the Parodi-Leslie relations.
We elaborate on our results in Sec. \ref{sec:discussion}, conducting a thorough examination of critical issues, and providing insights into possible future research directions.
App. \ref{app:toolbox} comprises a collection of auxiliary lemmata and identities that are required for our derivations.

\section{Notation and Preliminaries}\label{sec:notation}

We consider a notation similar to that in \cite{NitschkeSadikVoigt_A_2022,NitschkeVoigt_2023} and only rephrase it briefly. For further details, we direct readers to the aforementioned references. 
A moving surface $ \surf $ is sufficiently described by parameterizations
\begin{align}\label{eq:para}
    \para:\quad \mathcal{T}\times\mathcal{U} \rightarrow \R^3:\quad (t,y^1,y^2)\mapsto\para(t,y^1,y^2)\in\surf\vert_{t}\formComma
\end{align}
where $ \mathcal{U}\subset\R^2 $ is the open chart codomain and $ \mathcal{T}=[t_0,t_1]\subset\R $ the time domain.
For simplicity we assume that $ \para(t,\mathcal{U})=\surf\vert_{t} $ can be achieved by a single time-depending parameterization $ \para $ for all $ t\in\mathcal{T} $.
The results can be extended to the more general case considering subsets providing an open covering of $ \surf $.
The parametrization \eqref{eq:para} does not define the moving surface $ \surf $ uniquely.
It yields an observer additionally.
Due to this we call every parameterization providing the same moving surface $ \surf $ an observer parametrization. 
See \cite{NitschkeVoigt_JoGaP_2022} for more details on this topic.
How we treat the observer choice in this paper is stated below after introducing necessary tensor field spaces and spatial differential operators. 
Usually we omit arguments $ t $, $ y^1 $ and $ y^2 $, since the choice of an observer as well as the local chart are fixed but also arbitrary for our calculus.
The integration over a scalar field $ f $ is defined by
$
    \int_{\surf} f \dS := \int_{\mathcal{U}} f \mu
$
where $ \mu $ is the differential area form uniquely given by the choice of the local chart, see \cite{AbrahamMarsdenRatiu_2012}.
Functionals that map to $ \mathcal{T} \rightarrow \R $ are always of the form $ \functional = \int_{\surf} f \dS $, \ie\ their time dependencies are solely encoded in $ f $ and $ \para $.

If we are using (index) proxy notion \wrt\ a frame or local coordinates, Einstein summation is applied, 
where capital Latin indices starting with $ A $ ($ ,B,C,\ldots $) refer to the Cartesian frame $\{\eb_x, \eb_y, \eb_z\}$ and small Latin indices starting with $ i $ ($ ,j,k,\ldots $) refer to a tangential frame $ \{\partial_1\para, \partial_2\para\} $.
We write $ \tangentR[^n] $ as a shorthand for the space of sufficiently smooth $ n $-tensor fields on $ \surf $ with values in $ (\R^3)^n $.
These values could be  given as tensor proxies \wrt\ to a chosen but arbitrary frame.
For instance, it holds $ \Wb = W^{A}\eb_A = W^i\partial_i\para + W_{\bot}\normal \in\tangentR:=\tangentR[^1] $ for vector fields, where $ [W^x,W^y,W^z], [W^1,W^2,W_{\bot}] \in \R^3 $ are given at every event in $ \mathcal{T}\times\surf $, and $ \normal\in\tangentR $ is an arbitrary orientated normals field perpendicular to the surface.
The tensor transposition operator $ T_{\sigma}:\tangentR[^n]\rightarrow\tangentR[^n] $ is determined by an index permutation $ \sigma\in S_n $.
For instance, the transposition of a $ n $-tensor field $ \Rb\in\tangentR[^n] $ is given in Cartesian coordinates by
$ [  \Rb^{T_{\sigma}}]^{A_1\ldots A_n} = R^{A_{\sigma(1)}\ldots A_{\sigma(n)}} $.
We omit the $ \sigma $-notation if $ \sigma $ is a cyclically right shift, 
\ie\ $ T := T_{(1\,2\,\ldots\,n)} $ in cycle notation.
The tensor multiplication $ \tangentR[^n]\times\tangentR[^m]\rightarrow\tangentR[^{n+m-2}] $ for $ n,m \ge 1 $ are written without a dot and \wrt\ the frame, \eg\ for tangential vector fields $ \Wb=W^i\partial_i\para, \Rb=R^i\partial_i\para  $ the multiplication $ \Wb\Rb = W^i R_i = g_{ij}W^i R^j \in\tangentR[^0] $ gives the local inner product, where $ g_{ij} = (\partial_i\para)(\partial_j\para) $ is the covariant proxy of the metric for the used tangential frame.
Especially if a $ 2 $-tensor field $ \Rb\in\tangentR[^2] $ is multiplied by vector fields $ \Wb_1,\Wb_2\in\tangentR $ from left and right we also apply the bilinear notation $ \Rb(\Wb_1,\Wb_2):= \Wb_1\Rb\Wb_2 \in\tangentR[^0]$ for a better readability.
We also use the double multiplication $ \tangentR[^n]\times\tangentR[^m]\rightarrow\tangentR[^{n+m-4}] $ for $ n,m \ge 2 $ with a double dot ($ \dbdot $), where the two rear components of the left argument are contracted with the two leading components of the right argument, \eg\ $ [\Rb \dbdot \hat{\Rb}]^{A_1\ldots A_{n-2} B_1 \ldots B_{m-2}} = \tensor{R}{^{A_1,\ldots A_{n-2} C D}} \tensor{\hat{R}}{_{C D}^{B_1 \ldots B_{m-2}}}  $ \wrt\ the Cartesian frame.
At every event, $ \tangentR[^n]  $ is a vector space.
All subspaces $ \Vcal \subset \tangentR[^n]  $, which yield linear subspaces eventwisely, are called subtensor field spaces and we write $ \Vcal < \tangentR[^n] $ to emphasize this structural relation.
Every subtensor field space has its own uniquely defined orthogonal projection $ \proj_{\Vcal}: \tangentR[^n] \rightarrow \Vcal $ yielding $ \proj_{\Vcal}(\tangentR[^n])= \Vcal $.
If we do not need an explicit calculation we do not give a representation of the projection and only use its orthogonal properties.
Perhaps the most important subtensor field space on surfaces is the space of tangential $ n $-tensor fields 
$ \tangentS[^n]:= \{ \Rb\in\tangentR[^n] \mid \forall\sigma\in S_n : \normal\Rb^{T_{\sigma}} = \nullb \} $ for $ n \ge 1 $.
We use capital letters for tensor fields and small letters for tangential tensor fields usually, \eg\ $ \Rb\in\tangentR[^n] $ but $ \rb\in\tangentS[^n] $.
The space of scalar fields is $ \tangentS[^0]:=\tangentR[^0] $ and we use non-bold letters to identify its elements, \eg\ $ f, W^i, R^{AB} \in\tangentS[^0] $.
As for vector fields, we omit the index for the space of tangential vector fields $ \tangentS := \tangentS[^1] $.
Its orthogonal complement space in $ \tangentR $ is the space of surface orthogonal vector fields $ \tangentnormal\bot\tangentS $ spanned by the normal field, see Figure \ref{cd:vector_spaces}.
\begin{figure}[t]
    \centering
    \begin{tikzcd}[sep=small]
        \text{\fbox{vector fields}}
        &\text{\underline{surface orthogonal}} \arrow[r, dash, dashed, gray]
            & \tangentnormal  \arrow[dr, hook, end anchor={[xshift=-15pt]}]
                &\\
        &\text{\underline{\phantom{[general]}}} \arrow[rr, dash, dashed, gray]
            &
                & \tangentR = \tangentS\oplus\tangentnormal \\
        &\text{\underline{tangential}} \arrow[r, dash, dashed, gray]
            & \tangentS \arrow[ur, hook, end anchor={[xshift=-15pt]}]
                &
    \end{tikzcd}
     \caption{Vector field spaces and their subvector field relations.
        The naming of the spaces is generic by \textit{[row name] vector fields}, \eg\ $\tangentnormal = (\tangentS[^0])\normal$ is the space of surface orthogonal vector fields.
        Hooked arrows represent subvector field relations by:  $ \Vcal_2 \hookrightarrow \Vcal_1\ :\Leftrightarrow \ \Vcal_2 =\proj_{\Vcal_2} \Vcal_1 < \Vcal_1 $, where $ \proj_{\Vcal_2}: \Vcal_1 \rightarrow \Vcal_2 $ is the associated uniquely given orthogonal projection.}
    \label{cd:vector_spaces}
\end{figure}
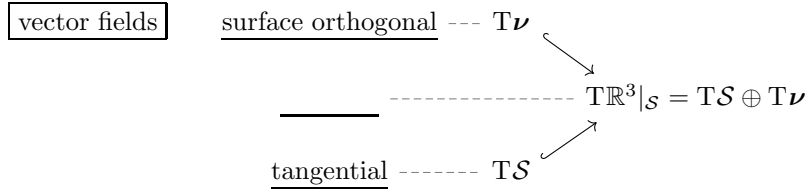
The space of $ 2 $-tensor fields $ \tangentR[^2] $ possesses more subtensor field spaces that are relevant for us.
The space of symmetric 2-tensor fields is $ \tangentSymR:=\{ \Rb\in\tangentR[^2] \mid \Rb = \Rb^T \} $ 
and its orthogonal complement is the space of skew-symmetric $ 2 $-tensor fields  $ \tangentAR := \{ \Rb\in\tangentR[^2] \mid \Rb = -\Rb^T \} \bot \tangentSymR $.
The most important subtensor field space in this paper is the space of Q-tensor fields $ \tangentQR:= \{ \Rb\in\tangentSymR \mid \Tr\Rb = 0 \} $, which comprises all symmetric and trace-free $ 2 $-tensor fields, \cf\ \cite{MottramNewton_2014}.
Its orthogonal complement in $ \tangentSymR $ is the space of isotropic $ 2 $-tensor fields $ \tangentIdR := \{ f\Id \mid f\in\tangentS[^0]  \} $,
where $ \Id $ is the identity tensor field, which can be formulated in Cartesian proxy notation by $ [\Id]^{AB}= \delta^{AB} $ or with a tangential frame by $ \Id= g^{ij}\partial_i\para\otimes\partial_j\para + \normal\otimes\normal  $ for instance.
Note that we use the identity tensor field to define the trace $ \Tr: \tangentR[^{n+2}] \rightarrow \tangentR[^n] $, $ \Rb\mapsto\Tr\Rb:=\Rb\dbdot\Id $, \ie\ it contracts the two rear tensor components.
All these subtensor field spaces have their own tangential subtensor field spaces in turn, 
\ie\ $ \tangentSymS := \projS[^2](\tangentSymR) $,   $ \tangentAS := \projS[^2](\tangentAR) $, $ \tangentQS := \projS[^2](\tangentQR) $ and  
$ \tangentIdS := \projS[^2](\tangentIdR) = \{ f\IdS \mid f\in\tangentS[^0]  \}$, where $ \IdS = \Id - \normal\otimes\normal = g^{ij}\partial_i\para\otimes\partial_j\para $ is the tangential identity tensor field.
Moreover, we use the space of surface conforming Q-tensor fields 
\begin{align}\label{eq:surface_conforming_space}
    \tangentCQR 
        &:= \left\{ \Qb \in \tangentQR \mid \exists\lambda\in\tangentS[^0] : \Qb\normal = \lambda\normal \right\} \formComma
\end{align}
where $ \tangentQS < \tangentCQR < \tangentQR $ holds.
More details on surface conforming Q-tensor fields can be found in \cite{NitschkeVoigt_2023, BouckNochettoYushutin_2022} and Sec. \ref{sec:surface_conforming_constrain}.
For an overview of all these subtensor field spaces see Figure \ref{cd:ttensor_spaces}.
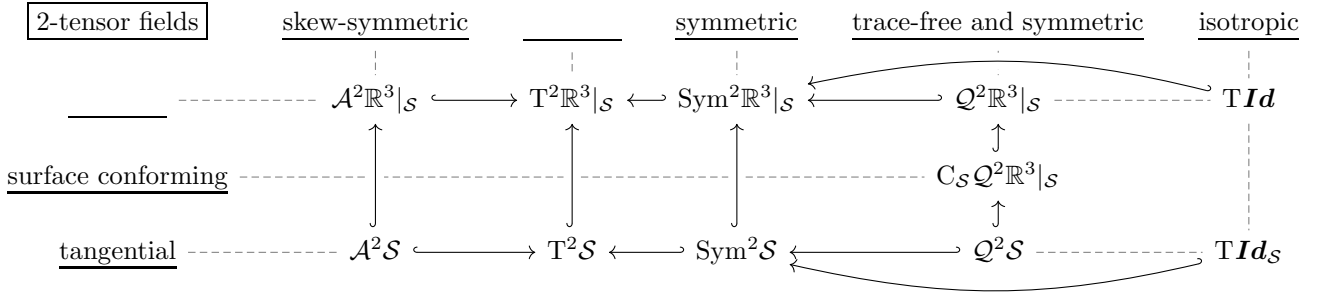
\begin{figure}[t]
    \centering
    \begin{tikzcd}[sep=small]
        \text{\fbox{2-tensor fields}}
            &\text{\underline{skew-symmetric}} \arrow[d, dash, dashed, gray]
                & \text{\underline{\phantom{[general]}}} \arrow[d, dash, dashed, gray]
                    & \text{\underline{symmetric}} \arrow[d, dash, dashed, gray]
                        & \text{\underline{trace-free and symmetric}} \arrow[d, dash, dashed, gray]
                            & \text{\underline{isotropic}} \arrow[d, dash, dashed, gray]\\
        \text{\underline{\phantom{[general]}}}\arrow[r, dash, dashed, gray] 
            &\tangentAR \arrow[r, hook]
                & \tangentR[^2] 
                    & \tangentSymR \arrow[l, hook'] 
                        & \tangentQR \arrow[l, hook'] \arrow[r, dash, dashed, gray]
                            & \tangentIdR \arrow[ll, bend right=13, hook', crossing over] \arrow[dd, dash, dashed, gray]\\
        \text{\underline{surface conforming}}\arrow[rrrr, dash, dashed, gray]
            & 
                & 
                    & 
                        & \tangentCQR \arrow[u, hook] \\
        \text{\underline{tangential}} \arrow[r, dash, dashed, gray] 
            & \tangentAS \arrow[uu, hook, crossing over] \arrow[r, hook] 
                & \tangentS[^2]\arrow[uu, hook, crossing over] 
                    & \tangentSymS \arrow[uu, hook, crossing over] \arrow[l, hook'] 
                        &  \tangentQS \arrow[l, hook'] \arrow[u, hook] \arrow[r, dash, dashed, gray]
                            & \tangentIdS \arrow[ll, bend left=13, hook']
    \end{tikzcd}
    \caption{Frequently used 2-tensor field spaces and their subtensor field relations in a commutative diagram.
        The naming of the spaces is generic by \textit{[row name] [column name] 2-tensor fields}, \eg\ $ \tangentSymS $ is the space of tangential symmetric 2-tensor fields.
        The naming component \textit{trace-free and symmetric 2-tensor} is usually combined to \textit{Q-tensor}, \eg\ $ \tangentQS $ is the space of tangential Q-tensor fields.
        Hooked arrows represent subtensor field relations by:   $ \Vcal_2 \hookrightarrow \Vcal_1\ :\Leftrightarrow \ \Vcal_2 =\proj_{\Vcal_2} \Vcal_1 < \Vcal_1 $, where $ \proj_{\Vcal_2}: \Vcal_1 \rightarrow \Vcal_2 $ is the associated uniquely given orthogonal projection.}
     \label{cd:ttensor_spaces}
\end{figure}
Tensor field spaces can entail useful orthogonal decompositions.
For instance vector fields can be decomposed in a usual way by 
\begin{align*}
\forall\Wb\in\tangentR,\exists!\wb\in\tangentS, w_{\bot}\in\tangentS[^0]
   &:\quad \Wb = \wb + w_{\bot}\normal \formPeriod
\end{align*}
Also Q-tensor fields can be approached in similar fashion by 
\begin{align}
\MoveEqLeft[10]
\forall\Qb\in\tangentQR,\exists!\qb\in\tangentQS,\etab\in\tangentS,\beta\in\tangentS[^0]:\notag\\
     &\Qb = \Qdepl(\qb,\etab,\beta) := \qb + \etab\otimes\normal + \normal\otimes\etab + \beta\left( \normal\otimes\normal -\frac{1}{2}\IdS \right) \formComma
        \label{eq:qtensor_decomposition}
\end{align}
see \cite{NitschkeVoigt_2023}.
Note that the rear orthogonal summand could be chosen differently, without any impacts on the decomposition, \eg\ 
it holds  $ \beta( \normal\otimes\normal -\frac{1}{2}\IdS ) = \tilde{\beta}( \normal\otimes\normal -\frac{1}{3}\Id ) $ for $ 3\beta = 2\tilde{\beta} $.

The local inner product $ \inner{\tangentR[^n]}{\Rb_1,\Rb_2}\in\tangentS[^0] $ on $ n $-tensor fields $ \Rb_1,\Rb_2\in\tangentR[^n] $ is defined by a fully contraction of both tensor fields locally against each other in the usual way, 
\eg\  $ \inner{\tangentR}{\Rb_1,\Rb_2} = \Rb_1 \Rb_2 $ for vector fields or $ \inner{\tangentR[^2]}{\Rb_1,\Rb_2} = \Rb_1 \dbdot \Rb_2 $ for $ 2 $-tensor fields.
Although this product can be used for subtensor field spaces $ \Vcal \le \tangentR[^n] $ without any restrictions, we often label it with $ \Vcal $, \ie\ $ \inner{\Vcal}{\Rb_1,\Rb_2} = \inner{\tangentR[^n]}{\Rb_1,\Rb_2} $ for all $ \Rb_1,\Rb_2\in\Vcal $.
Of course, using a customized base spanning $ \Vcal $ to evaluate the inner product is equivalent to use a frame for $ \tangentR[^n] $.
The rules of linear algebra remain locally valid.
We define the norm by $ \| \Rb \|_{\Vcal} := \sqrt{\inner{\Vcal}{\Rb,\Rb}} $ for all $ \Rb\in\Vcal $.
The spatial global inner product and norm is given by $ \innerH{\Vcal}{\Rb_1,\Rb_2}:=\int_{\surf}  \inner{\Vcal}{\Rb_1,\Rb_2} \dS $ and $ \| \Rb \|_{\hil(\Vcal)} := \sqrt{\innerH{\Vcal}{\Rb,\Rb}} $ for all $ \Rb_1,\Rb_2, \Rb\in\Vcal\le\tangentR[^n] $.

Spatial derivatives on tensor fields are based on two different differential operators in this paper.
One of them is covariant derivative $ \nabla:\tangentS[^n] \rightarrow \tangentS[^{n+1}]  $ for tangential tensor fields given by the Christoffel symbols
$ \Gamma_{ijk} = \frac{1}{2}(\partial_i g_{jk} + \partial_j g_{ik} - \partial_k g_{ij} ) $ in the usual way.
For instance it is $ \nabla\wb = g^{kl}\tensor{w}{^i_{|k}} \partial_i\para \otimes \partial_l \para $, where $ \tensor{w}{^i_{|k}} = \partial_k w^i + \Gamma_{kj}^{i} w^j $
and $ \Gamma_{kj}^{i} = g^{im}\Gamma_{kjm} $, for tangential vector fields $ \wb\in\tangentS $.
The other one is the (Cartesian) componentwise (surface) derivative $ \nablaC:\tangentR[^n] \rightarrow  \tangentR[^n]\otimes\tangentS $ 
defined by $ \nablaC\Rb := (\IdS\nabla_{\R^3})\Rb  $, \ie\ $ \Rb $ is differentiated by the embedding space derivative solely in tangential direction.
Note that $ \IdS\nabla_{\R^3} $ has to be read together and cannot be evaluated separately a priori. 
If one still wants to do it, then the argument has to be extended in normal direction sufficiently smooth.
However $ \nablaC $ is independent \wrt\ such extensions.
For scalar fields $ f\in\tangentS[^0]$ hold $ \nablaC f = \nabla f $.
This justifies the naming ``componentwise'', since we get $ [\nablaC\Rb]^{A_1 \ldots A_n B} \eb_B = \nabla R^{A_1 \ldots A_n}\in\tangentS $ \wrt\ a Cartesian frame for $ n $-tensor fields $ \Rb\in\tangentR[^n] $.
The coordinates independent relation between covariant and componentwise derivative for vector fields $ \Wb=\wb +\wnor\normal\in\tangentR $ is
$ \nablaC\Wb = \nabla\wb - \wnor\shop + \normal\otimes(\nabla\wnor + \shop\wb) $,
where $ \shop := -\nablaC\normal \in \tangentSymS $ is the 
shape operator\footnote{We do not distinguish between ``shape operator'', ``second fundamental form'' or ``(extended) Weingarten map''. They are all equal up to isomorphisms, which are invariant \wrt\ our calculus anyway.}.
We specify such relations for $ n>1 $ only at places we need them, since the syntactical complexity of these formulations grows exponentially with $ n $.
The covariant divergence on tangential $ n $-tensor fields $ \div:=-\nabla^*=\Tr\circ\nabla:\tangentS[^n]\rightarrow\tangentS[^{n-1}] $ is defined by the \mbox{$ \hil $-adjoint}, \resp\ equivalently as the trace, of the covariant derivative for $ n>1 $.
Unfortunately, a similar relation between the \mbox{$ \hil $-adjoint} and the trace does not hold for the componentwise derivative generally.
Therefore we define the \mbox{$ \hil $-adjoint} componentwise divergence $\divC:= -\nablaC^* $ and the trace componentwise divergence $ \DivC:=\Tr\circ\nablaC:\tangentR[^n]\rightarrow\tangentR[^{n-1}]  $ separately.  
However, on right-sided tangential tensor fields $ \tangentR[^{n-1}]\otimes\tangentS $ they are equal, see Lemma \ref{lem:tensor_divC_to_DivC}.
In this paper that is always the case for stress tensor fields $ \Sigmab= \sigmab + \normal\otimes\etab $, with $ \sigmab\in\tangentS[^2] $ and $ \etab\in\tangentS $, where the relation to the covariant divergence is given by
\begin{align}\label{eq:DivCSigma_decomposed}
    \DivC\Sigmab = \divC\Sigmab = \div\sigmab - \shop\etab + ( \sigmab\dbdot\shop + \div\etab )\normal \formPeriod
\end{align}
For  vector fields $ \Wb=\wb +\wnor\normal\in\tangentR $ holds $ \divC\Wb = \div\wb $ and $ \DivC\Wb = \div\wb - \wnor\meanc $, where $ \meanc:=\Tr\shop\in\tangentS[^0] $ is the mean curvature.
Due to the general incompatibility between the \mbox{$ \hil $-adjoint} and the trace componentwise divergence, we define the adjoint componentwise gradient $ \GradC:=-\DivC^*:\tangentR[^n] \rightarrow  \tangentR[^{n+1}] $.
We use this operator only for scalar fields $ f\in\tangentS[^0] $, where 
\begin{align} \label{eq:GradC}
    \GradC f = \DivC (f\IdS) = \nabla f + \meanc f \normal
\end{align}
holds.
The covariant Bochner-Laplace operator is given by $ \Delta:= \div\circ\nabla:\tangentS[^n] \rightarrow \tangentS[^n] $
and the componentwise Laplace operator by  $ \DeltaC:= \divC\circ\nablaC = \DivC\circ\nablaC:\tangentR[^n] \rightarrow \tangentR[^n] $.
The relation between $ \DeltaC $ and covariant differential operators depends on $ n $. 
Only for scalar fields ($ n=0 $) holds simply $ \DeltaC f = \Delta f $.
For $ n=2 $ a relation is given in \cite{NitschkeVoigt_2023}.
We define the surface Levi-Civita tensor $ \Eb\in\tangentAS[^2] $ by its covariant proxy field components $ E_{ij} := \sqrt{\vert \gb[\para] \vert}\epsilon_{ij} $, where $ \epsilon_{ij} $ are the 2-dimensional Levi-Civita symbols.
As a consequence the $ \R^3 $-cross product on tangential vector fields  $ \rb_1, \rb_2 \in\tangentS $ yields $\rb_1\times\rb_2  = \Eb(\rb_1, \rb_2)\normal\in\tangentnormal $.
The covariant curl of a tangential vector field $ \rb\in\tangentS $ is defined by  $\rot\rb := - \Eb\dbdot\nabla\rb \in\tangentS[^0] $.
In addition to the usual methods for calculating the Gaussian curvature $ \gaussc\in\tangentS[^0] $, such as $ \gaussc = \det\{\shopC^i_j\} $, 
we can also apply the Levi-Civita tensor twice to the Riemann curvature tensor $ \boldsymbol{\mathcal{R}}\in\tangentS[^4] $, \ie\ $ \gaussc = \frac{1}{4}\Eb\dbdot\boldsymbol{\mathcal{R}}\dbdot\Eb $.

We use the surface deformation gradient $ \Gbcal[\Wb] := \nablaC\Wb - \normal\nablaC\Wb\otimes\normal\in\tangentR[^2] $
for deformation directions $ \Wb\in\tangentR $ \st\ $ \para + \eps\Wb $ describes the parameterization of a linear perturbed surface with $ \eps>0 $ sufficiently small. 
In terms of covariant derivatives, this yields
\begin{align}\label{eq:Gbcal}
\Gbcal[\Wb] 
    &= \nabla\wb - \wnor\shop + \normal\otimes\left( \nabla\wnor + \shop\wb \right) - \left( \nabla\wnor + \shop\wb \right)\otimes\normal
\end{align}
for decomposition $ \Wb=\wb+\wnor\normal $.
The tangential surface deformation gradient is defined by $ \Gb[\Wb] := \projS[^2]\Gbcal[\Wb] $ and yields
\begin{align}\label{eq:Gb}
\Gb[\Wb] &= \IdS\nablaC\Wb = \nabla\wb - \wnor\shop \formPeriod
\end{align}
Note that the non-tangential part of $ \Gbcal[\Wb]  $ is always skew-symmetric, 
\ie\ it holds $ \Gbcal[\Wb] - \Gb[\Wb]\in\tangentAR $ for all $ \Wb\in\tangentR $. 
As a consequence, the symmetric surface deformation gradient is given by
\begin{align}\label{eq:Sb}
    \Sb[\Wb]
        &:= \frac{1}{2}\left( \Gbcal[\Wb] + \Gbcal^T[\Wb] \right)
        = \frac{1}{2}\left( \Gb[\Wb] + \Gb^T[\Wb] \right)
        = \frac{1}{2}\left( \nabla\wb + (\nabla\wb)^T \right) - \wnor\shop 
         \in \tangentSymS
\end{align}
and the skew-symmetric surface deformation gradient by
\begin{align}
 \Abcal[\Wb]  \label{eq:Abcal}
      &:= \frac{1}{2}\left(\Gbcal[\Wb] - \Gbcal^{T}[\Wb]\right)
      = \Ab[\Wb] + \normal\otimes\left( \nabla\wnor + \shop\wb \right) - \left( \nabla\wnor + \shop\wb \right)\otimes\normal
        \in \tangentAR\formComma \\
 \Ab[\Wb] \label{eq:Ab}
      &:= \frac{1}{2}\left(\Gb[\Wb] - \Gb^{T}[\Wb]\right)
      = \frac{1}{2}\left( \nabla\wb - (\nabla\wb)^T \right)
      = -\frac{\rot\wb}{2}\Eb
         \in\tangentAS \formComma 
\end{align}
where $ \Ab[\Wb] $ is the tangential skew-symmetric surface deformation gradient.
See \cite{NitschkeSadikVoigt_A_2022} for more details and consequences in perturbing surfaces by linear deformations.

Unless we state otherwise, we treat every derivation in Sec. \ref{sec:derivation} \wrt\ the material observer, 
\ie\ we adopt the Lagrange perspective.
Even if this simplification is not necessary, it makes the calculations a lot easier and clearer.
For instance, if we like to investigate the change of a scalar function $ f $ \wrt\ a deformation of the material by $ \para_{\mfrak} + \eps \Wb $,
we only have to evaluate $ f[\para_{\mfrak} + \eps \Wb] $ if $f$ measure a quantity from the Lagrange perspective in first place.
If $ f $ is measuring from an arbitrary perspective then we would have to assess the much more difficult expression $ f[\xb_{\ofrak}[\para_{\mfrak} + \eps \Wb]] $, where $ \xb_{\ofrak} $ has to be specified sufficiently for all possible perturbed surfaces with $ \xb_{\ofrak}[\para_{\mfrak}] = \para_{\ofrak} $.
Note that this decision is without loss of generality, since the calculus we use in this paper is observer-invariant as well as the physics we like to describe with it.
For the sake of readability we omit the subscript $ \mfrak $ and mostly the associated  prefix ``material'' in this context.
As a consequence we set $ \para := \para_{\mfrak} $ and $ \Vb:= \Vb_{\mfrak} = \partial_t\para $. 
Especially in Sec. \ref{sec:derivation} we state $ \para_{\ofrak} = \para $ and $ \Vb_{\ofrak} = \Vb $ additionally if it is necessary.
However, this also means that local equations, as a consequence of a variational principle, are in terms of the Lagrange perspective a priori.
We drop this assumption in Sec. \ref{sec:models} by representing the local equations in terms of observer-invariant formulations.
This is feasible since all local terms we get have either a derived observer-invariant representation or are even instantaneous.
In doing so we establish the convention that $ \Vb = \Vb_{\mfrak}\in\tangentR $ is the material velocity field and $ \Vb_{\ofrak}\in\tangentR $ is an arbitrary but fixed observer velocity field not necessarily equal to $ \Vb $.
Note that an Euler perspective does not exist generally. 
The normal velocity has to be equal for all observer, \ie\ $ \Vb\normal = \Vb_{\ofrak}\normal = \vnor $.
As a consequence the relative velocity
\begin{align*}
    \ub &:= \Vb - \Vb_{\!\ofrak} = \vb - \vb_{\!\ofrak} \in \tangentS
\end{align*}  
is a tangential vector field. 
Observer-invariant formulations affect only time derivatives in this context and are investigated in \cite{NitschkeVoigt_JoGaP_2022} for tangential $ n $-tensor fields and in \cite{NitschkeVoigt_2023} for $ n $-tensor fields with $ n \le 2 $. 
We only briefly summarize this. Scalar fields $ f\in\tangentS[^0] $ yield
\begin{align}\label{eq:dotf}
    \dot{f} 
        &= \partial_t f + (\nablaC f)\ub
        = \partial_t f + \nabla_{\ub} f \formComma
\end{align}
see also \cite{BachiniKrauseNitschkeVoigt_2023}.
The material acceleration $ \Dmat\Vb\in\tangentR $ is given as the material time derivative of the material velocity, \ie
\begin{align}\label{eq:DmatV}
    \Dmat\Vb
        &= \dot{V}^A \eb_A
         = \dot{\vb} -\vnor\left( \nabla\vnor + \shop\vb \right) + \left( \dot{v}_{\bot} + \nabla_{\vb}\vnor + \shop(\vb,\vb) \right)\normal \in\tangentR \formComma\\
    \dot{\vb}
        &= (\partial_t v^i)\partial_i\para_{\ofrak} + \nabla_{\ub}\vb + \nabla_{\vb}\vb_{\ofrak} - \vnor\shop\vb \in \tangentS \formComma \notag
\end{align}
see also \cite{Yavari_2016}.
With respect to decomposition $ \Qb=\Qdepl(\qb,\etab,\beta)\in\tangentQR $ \eqref{eq:qtensor_decomposition}, where $ \qb\in\tangentQS $, $ \etab\in\tangentS $ and $ \beta\in\tangentS[^0] $,
the material time derivative of a Q-tensor fields is given in a Cartesian frame as well as with the tangential differential calculus by
\begin{align}\label{eq:DmatQ}
    \Dmat\Qb
        &= \dot{Q}^{AB}\eb_A\otimes\eb_B
         = \Qdepl\left(\dot{\qb}- 2\projQS\left( \etab\otimes\bb \right),
                    \dot{\etab} + \qb\bb - \frac{3}{2}\beta\bb,
                    \dot{\beta} + 2 \etab\bb \right)
             \in \tangentQR \formComma\\
    \bb &= \normal\nablaC\Vb = \nabla\vnor + \shop\vb \in \tangentS \formComma\notag\\
    \dot{\qb} \label{eq:dotq}
        &= (\partial_t q^{ij})\partial_i\para_{\ofrak}\otimes\partial_j\para_{\ofrak}
            + \nabla_{\ub}\qb + \left( \nabla\vb_{\ofrak} - \vnor\shop \right)\qb + \qb\left( (\nabla\vb_{\ofrak})^T - \vnor\shop \right)
           \in\tangentQS \formComma\\
    \dot{\etab}
        &= (\partial_t \eta^i)\partial_i\para_{\ofrak} + \nabla_{\ub}\etab + \nabla_{\etab}\vb_{\ofrak} - \vnor\shop\etab \in \tangentS \formPeriod\notag
\end{align}
With the skew-symmetric surface deformation gradient \eqref{eq:Abcal} and its tangential part \eqref{eq:Ab}, the Jaumann time derivative yields
\begin{align}\label{eq:DjauQ}
    \Djau\Qb 
        &= \Dmat\Qb - \Abcal[\Vb]\Qb + \Qb\Abcal[\Vb]
         = \Qdepl(\timeJ\qb, \timeJ\etab, \dot{\beta}) \in\tangentQR \formComma\\
    \timeJ\qb \label{eq:jauq}
        &= \dot{\qb} - \Ab[\Vb]\qb + \qb\Ab[\Vb] \in\tangentQS \formComma\\
    \timeJ\etab
        &= \dot{\etab} - \Ab[\Vb]\etab \in \tangentS \formPeriod\notag
\end{align}

The deformation derivative for surface/parametrization dependent scalar fields $ f=f[\para]\in\tangentS[^0] $ in direction of $ \Wb\in\tangentR $ is defined by
\begin{align}\label{eq:gaugef}
    \gauge{\Wb}f &:= \ddfrac{\eps}\Big\vert_{\eps=0} f[\para + \eps\Wb] \in\tangentS[^0]\formPeriod
\end{align}
Note that we mostly omit the extra notion ``$ [\para] $'' for a better readability and contrarily to \cite{NitschkeSadikVoigt_A_2022}.
However, we suggest to keep in mind that scalar fields could have such a dependency.
For pure geometrical quantities it is always determined, \eg\ for $ \meanc = \meanc[\para] $, where the mean curvature $ \meanc[\para + \eps\Wb] $ has a fixed field value at the perturbed surface.
But for other variables this dependency could be undetermined.
A stipulation of $ f[\para + \eps\Wb] $ up to first order, in terms of a Taylor expansion at $ \eps=0 $, is called a scalar gauge of surface independence, \eg\ by $ \gauge{\Wb}f = 0 $ for all $ \Wb\in\tangentR $.
Similar to the material time derivative, we can extend the deformation derivative to $ n $-tensor fields $ \Rb \in \tangentR[^n] $ by their Cartesian proxy fields, \ie\
\begin{align}\label{eq:gaugeR}
    \gauge{\Wb}\Rb &:= (\gauge{\Wb} R^{A_1 \ldots A_n}) \bigotimes_{\alpha=1}^{n} \eb_{A_{\alpha}} \formPeriod
\end{align}
Giving a representation for an arbitrary frame is feasible, \cf\ \cite{NitschkeSadikVoigt_A_2022} for tangential tensor fields, but not necessary for our purpose.
Although undetermined deformation derivatives of degrees of freedom arise from our model derivation, they all cancel out.
In Sec. \ref{sec:lagrangedalembert} we present some more details especially for Q-tensor fields, since assuming a tensorial gauge of surface independence is necessary if  we like to determine some of the resulting forces individually.
A much easier issue to grasp are field quantities depending on the velocity $ \Vb\in\tangentR $.
For velocity depending scalar fields $ f=f[\Vb]\in\tangentS[^0] $ the velocity derivative in direction of $ \Wb\in\tangentR $ is defined by
\begin{align}\label{eq:veloderf}
    \veloder{\Wb}f 
        &:= \ddfrac{\eps}\Big\vert_{\eps=0} f[\Vb + \eps\Wb] \in\tangentS[^0]\formPeriod
\end{align}
This can also be used for tensor fields by applying the velocity derivative on their proxy component fields without any further considerations, since any considered frame in this paper is velocity independent. 
For instance it holds $ \veloder{\Wb}\Rb = (\veloder{\Wb}R^{A})\eb_A = (\veloder{\Wb}R^{i})\partial_i\para + (\veloder{\Wb}R_{\bot})\normal $  for all $ \Rb=\Rb[\Vb]\in\tangentR $.

Spatially global functionals always have the form $ \functional = \functional[\para] = \int_{\surf} f \dS $ with $ f=f[\para]\in\tangentS[^0] $ in this work. 
The (total) variation \wrt\ $ \para $ in direction of $ \Wb\in\tangentR $ is defined and also yields
\begin{align}\label{eq:varXF}
    \innerH{\tangentR}{\deltafrac{\functional}{\para}, \Wb}
        &:= \ddfrac{\eps}\Big\vert_{\eps=0} \functional[\para + \eps\Wb]
          = \int_{\surf} \gauge{\Wb}f + f\Tr\Gb[\Wb] \dS \formComma
\end{align}
as a consequence, where $ \Tr\Gb[\Wb] = \DivC\Wb $ holds.
The rear summand comes from the deformation of the differential area form, see \cite{NitschkeSadikVoigt_A_2022}.
Note that \eqref{eq:varXF} states a total variation, \ie\ it remains valid even if additional variables are partially involved in any representation, 
\eg\ for $ \functional[\para] = \functional[\para,\Rb[\para]] $, where we have to evaluate $ \functional[\para + \eps\Wb] =  \functional[\para + \eps\Wb, \Rb[\para + \eps\Wb]] $ due to this.
Analogously for velocity depending spatially global functional $ \functional = \functional[\Vb] = \int_{\surf} f \dS $ with $ f=f[\Vb]\in\tangentS[^0] $ we define the variation \wrt\ velocity $ \Vb\in\tangentR $ in direction of $ \Wb\in\tangentR $ and conclude
\begin{align}\label{eq:varVF}
    \innerH{\tangentR}{\deltafrac{\functional}{\Vb}, \Wb}
        &:= \ddfrac{\eps}\Big\vert_{\eps=0} \functional[\Vb + \eps\Wb]
          = \int_{\surf} \veloder{\Wb}f \dS \formPeriod
\end{align}
No additional summand occurs here, since the area element as an instantaneous field quantity is velocity independent.
Generally for tensor depending functionals $ \functional = \functional[\Rb] $ the variation \wrt\ $ \Rb\in\Vcal \le \tangentR[^n] $ in direction of $ \Psib\in\Vcal $ is defined by
\begin{align}\label{eq:varRF}
    \innerH{\Vcal}{\deltafrac{\functional}{\Rb}, \Psib}
        &:= \ddfrac{\eps}\Big\vert_{\eps=0} \functional[\Rb + \eps\Psib] \formPeriod
\end{align}
Note that if $ \Rb\in\tangentR[^2] $ is one of the observer-invariant tensorial rates given in \cite{NitschkeVoigt_2023}, then the variation \eqref{eq:varRF} is independent of the actual choice of this rate, see Lemma \ref{lem:ttensor_variational_independence_of_process_variable}.

\section{Surface Beris-Edwards Models}\label{sec:models}

\subsection{General Surface Q-Tensor Models}\label{sec:models_general}

    We present surface Beris-Edwards models as a modular concept based on the most general model presented in this paper.
    Individual building blocks result from an energetic consideration of potential energies and energy fluxes of an open physical system.
    The potential energies consist of an elastic energy ($ \EL $), a thermotropic energy ($ \THT $) and a bending energy ($\BE$).
    The energy fluxes are composed of an immobility potential ($ \IM $), and a nematic viscous potential ($ \NV $).
    Depending on the choice of a certain immobility mechanism identified by the symbol $ \Phi $, we consider the Jaumann model ($ \Phi=\jau $) and the material model ($\Phi=\mfrak$).
    Restrictions on these $ \Phi $-models can be chosen from the set $ \Cset \subset \{ \SC, \CB, \UN, \IS, \NN, \NF \} $ arbitrarily except for combinations that are mutually exclusive or merge into one another. These individual constraints relate to surface conformity ($ \SC $), a constant eigenvalue $ \beta $ in normal direction ($ \CB $), uniaxility ($ \UN $), an isotropic state ($ \IS $), no flow in normal direction ($ \NN $) and no flow at all ($ \NF $). Inextensibility is mandatory and already considered.
    All constraints are implemented by the Lagrange multiplier technique.
    Using the Lagrange-D'Alembert principle in Section \ref{sec:lagrangedalembert},
    the local equations \eqref{eq:LDA} yield the: 
    
    \paragraph{Surface Beris-Edwards models:} \textit{Find the material velocity field $ \Vb\in\tangentR $, Q-tensor field $ \Qb\in\tangentQR $,  surface pressure field $ p\in\tangentS $ and Lagrange parameter fields $\Lambdab_{\gamma}\in\Vcal_{\gamma}\le\tangentR[^{n_{\gamma}}]$ for all $ \gamma\in\Cset$ \st}
    \begin{subequations}\label{eq:model_lagrange_multiplier}
    \begin{gather}
        \rho\Dmat\Vb \label{eq:model_lagrange_multiplier_floweq}
            =  \GradC\left(\pTH - p\right) 
                        + \fnorBE\normal
                        + \DivC\left( \SigmabEL + \SigmabIM  + \SigmabNV^0 + \xi\SigmabNV^1 + \xi^2\SigmabNV^2 \right)
                        + \sum_{\gamma\in\Cset} \Fb_{\gamma} \formComma\\
        \widetilde{M} \Dt[\Phi]\Qb \label{eq:model_lagrange_multiplier_moleculareq}
                     = \HbEL + \HbTH + \xi\HbNV^1 + \xi^2\widetilde{\Hb}^{2,\Phi}_{\NV} + \sum_{\gamma\in\Cset} \Hb_{\gamma} \formComma\\
        0 = \DivC\Vb \formComma\\
        \nullb = \Cb_{\gamma}, \quad \forall\gamma\in\Cset \label{eq:model_lagrange_multiplier_constraineq}
    \end{gather}
    \end{subequations}
    \textit{holds for $ \dot{\rho}=0 $ and given initial conditions for $ \Vb $, $ \Qb $ and mass density $ \rho\in\tangentS[^0] $.}
    
    Mandatory quantities are given in Table \ref{tab:forces} and optional constraint quantities are given in Table \ref{tab:constraints}. The choice of $ \Phi $ for $  \SigmabNV^1  $ and $  \SigmabNV^2  $ is optional. Both given representations state equal tensor fields by \eqref{eq:DjauQ}. The relation to stress and force fields as given in Sec. \ref{sec:nematic_viscosity} is $ \SigmabNV = \SigmabNV^0 + \xi\SigmabNV^1 + \xi^2\SigmabNV^2 $ and $ \HbNV = \xi\HbNV^1 + \xi^2\left(\widetilde{\Hb}^{2,\Phi}_{\NV} -\frac{\coeffIF}{2}\Dphi\Qb\right)  $, which in total do not depend on $ \Phi $.
    Model parameters are the one-constant elastic parameter $ L $ (Sec.~\ref{sec:elastic}), thermotropic coefficients $ a $, $ b $ and $ c $ (Sec.~\ref{sec:thermotropic}), bending stiffness $ \kappa $ and spontaneous curvature $ \meanc_0 $ (Sec.~\ref{sec:bending}), immobility coefficient $ M $ (Sec.~\ref{sec:immobility}), isotropic viscosity $ \coeffIF $ and anisotropy coefficient $ \xi $ (Sec.~\ref{sec:nematic_viscosity}).
    The immobility coefficient adapted to the nematic viscosity is $ \widetilde{M}=  M + \frac{\coeffIF\xi^2}{2}  $.

    \begin{table}
        \centering
        \renewcommand{\arraystretch}{1.3}
        \begin{tabular}{|c|c|c|c|}
        \hline
            \multicolumn{2}{|c|}{Identifier}
                    & Expression
                        & Sec., Eq., Ref. \\
        \hline
            \multicolumn{2}{|c|}{$ \Dmat\Vb $}
                    & $ \partial_t\Vb + (\nablaC\Vb)(\Vb-\Vb_{\!\ofrak}) $
                        & \eqref{eq:DmatV}, \cite{NitschkeVoigt_2023}  \\
        \hline
            \multirow{2}{*}{$ \Dphi\Qb $}
                & $ \Phi=\jau $
                    & $ \partial_t\Qb + (\nablaC\Qb)(\Vb-\Vb_{\!\ofrak}) - \Abcal[\Vb]\Qb + \Qb\Abcal[\Vb] $
                        & \ref{sec:immobility}, \eqref{eq:DjauQ}, \cite{NitschkeVoigt_2023} \\
        \cline{2-4}
                & $ \Phi=\mfrak $
                    & $ \partial_t\Qb + (\nablaC\Qb)(\Vb-\Vb_{\!\ofrak}) $ 
                        & \ref{sec:immobility}, \eqref{eq:DmatQ}, \cite{NitschkeVoigt_2023} \\
        \hline
            \multicolumn{2}{|c|}{$ \SigmabEL $}
                    & $ -L\left( (\nablaC\Qb)^{T} \dbdot\nablaC\Qb - \frac{\| \nablaC\Qb \|^2}{2}\IdS \right) $
                        & \ref{sec:elastic} \\
        \hline
            \multicolumn{2}{|c|}{$ \HbEL $}
                    & $ L \DeltaC\Qb $
                        & \ref{sec:elastic}, \cite{NitschkeVoigt_2023}\\
        \hline
            \multicolumn{2}{|c|}{$ \pTH $}
                    & $  a \Tr\Qb^2 + \frac{2b}{3}\Tr\Qb^3 + c\Tr\Qb^4 $
                        & \ref{sec:thermotropic} \\
        \hline
            \multicolumn{2}{|c|}{$ \HbTH $}
                    & $ -2 \left( a \Qb + b\left( \Qb^2 - \frac{\Tr\Qb^2}{3}\Id \right) + c\Tr(\Qb^2)\Qb \right) $
                        & \ref{sec:thermotropic}, \cite{NitschkeVoigt_2023} \\
        \hline
            \multicolumn{2}{|c|}{$ \fnorBE $}
                    & $ -\kappa\left( \DeltaC\meanc + \left( \meanc - \meanc_{0} \right)\left( \frac{1}{2}\meanc(\meanc + \meanc_{0}) - 2\gaussc  \right) \right) $
                        & \ref{sec:bending}, \cite{BachiniKrauseNitschkeVoigt_2023} \\
        \hline
            \multirow{2}{*}{$ \SigmabIM $}
                & $ \Phi=\jau $
                    & $ M (\Id + \normal\otimes\normal)(\Qb\Djau\Qb - (\Djau\Qb)\Qb)\IdS $
                        & \multirow{2}{*}{\ref{sec:immobility}} \\
        \cline{2-3}
                & $ \Phi=\mfrak $
                    & $ \nullb $ & \\
        \hline
            \multicolumn{2}{|c|}{$ \SigmabNV^0 $}
                    & $ 2\coeffIF\Sb[\Vb] $
                        & \ref{sec:nematic_viscosity}, \cite{BachiniKrauseNitschkeVoigt_2023}\\
        \hline
            \multirow{2}{*}{$ \SigmabNV^1 $}
                &  $ (\Phi=\jau) $
                    & $ -\coeffIF\left( \IdS(\Djau\Qb)\IdS + (3\IdS+2\normal\otimes\normal)\Qb\Sb[\Vb] + \Sb[\Vb]\Qb\IdS \right) $
                        &\multirow{7}{*}{\ref{sec:nematic_viscosity}} \\
        \cline{2-3}
                & $ (\Phi=\mfrak) $
                    & $ -\coeffIF\left( \IdS(\Dmat\Qb)\IdS + \IdS\Qb\nablaC\Vb + 2\Qb\Sb[\Vb] + (\nablaC\Vb)^T\Qb\IdS \right)  $ & \\
        \cline{1-3}
            \multicolumn{2}{|c|}{$ \HbNV^1 $}
                    & $ \coeffIF\Sb[\Vb] $ &\\
        \cline{1-3}
            \multirow{2}{*}{$ \SigmabNV^2 $}
                &  $ (\Phi=\jau) $
                    & $ \coeffIF\left( \Qb\Djau\Qb\IdS - \normal\otimes \IdS \Qb (\Djau\Qb) \normal + \IdS\Qb\Sb[\Vb]\Qb\IdS + \Qb^2\Sb[\Vb] \right) $ &\\
        \cline{2-3}
                & $ (\Phi=\mfrak) $
                    & \splitcelltab{$ \coeffIF\big( \Qb\Dmat\Qb\IdS -\normal\otimes\IdS\Qb\left( (\Dmat\Qb)\normal - 2\Abcal[\Vb]\Qb\normal - \Qb(\nablaC\Vb)^T\normal \right)  $\\
                                                    $ +  \IdS\Qb\Gbcal^{T}[\Vb]\Qb\IdS + \Qb^2\nablaC\Vb \big)$} &\\
        \cline{1-3}
            \multirow{2}{*}{$ \widetilde{\Hb}^{2,\Phi}_{\NV} $}
                & $ \Phi=\jau $
                    & $ -\frac{\coeffIF}{2} \left( \Qb\Sb[\Vb] + \Sb[\Vb]\Qb - \frac{\Qb\dbdot\Sb[\Vb]}{3}\Id \right) $& \\
        \cline{2-3}
                & $ \Phi=\mfrak $
                    & $ -\frac{\coeffIF}{2} \left( \Qb\Gbcal[\Vb] + \Gbcal^T[\Vb]\Qb - \frac{\Qb\dbdot\Sb[\Vb]}{3}\Id \right) $& \\
        \hline
        \end{tabular}
        \caption{Necessary terms for \eqref{eq:model_lagrange_multiplier_floweq} and \eqref{eq:model_lagrange_multiplier_moleculareq} in the Surface Beris-Edwards models for a consistent choice $ \Phi\in\{\jau,\mfrak\} $.
                The deformation gradient field $ \Gbcal[\Vb] $ is given in \eqref{eq:Gbcal}. $ \Sb[\Vb] $ \eqref{eq:Sb} and $\Abcal[\Vb]$ \eqref{eq:Abcal} are its symmetric and skew-symmetric parts.
                Time derivatives are determined \wrt\ an observer velocity $ \Vb_{\!\ofrak}\in\tangentR $. }
        \label{tab:forces}
    \end{table}

    \begin{table}
        \centering
        \renewcommand{\arraystretch}{1.2}
        \begin{tabular}{|c|c|c|c|c|c|}
            \hline
            $ \gamma $
                & $\Fb_{\gamma}$
                    & $ \Hb_{\gamma} $
                        & $  \Cb_{\gamma} $ & $ \Vcal_{\gamma} $
                            & Sec.\\
            \hline
            $ \SC $ 
                & \splitcelltab{$ \DivC\big(\Qb(\normal,\normal) \normal \otimes  \lambdabSC $ \\ $-\normal \otimes \IdS\Qb\lambdabSC \big)$ } 
                    & $ -\frac{1}{2}\left( \lambdabSC\otimes\normal + \normal\otimes\lambdabSC \right) $ 
                        & $ \IdS\Qb\normal $ & $ \tangentS $
                            & \ref{sec:surface_conforming_constrain}\\
            \hline
            $ \CB $
                & $ \nullb $
                    & $ -\lambdaCB \left(\normal\otimes\normal -\frac{1}{3}\Id \right)$
                        & $\Qb(\normal,\normal) - \beta_0 $ & $ \tangentS[^0] $
                            & \ref{sec:constant_beta}\\
            \hline
            $ \UN $
                & $\nullb$
                    & \splitcelltab{$6\projQR\left( \LambdabUN\Qb^3 + \Qb\LambdabUN\Qb^2 \right)$\\
                              $-5(\Tr\Qb^2) \projQR\left( \LambdabUN\Qb \right)$\\
                              $-5 (\LambdabUN\dbdot\Qb^2)\Qb$}
                        & \splitcelltab{$  \Qb^4 - \frac{5}{6}(\Tr\Qb^2)\Qb^2$\\$\quad+ \frac{1}{9}(\Tr\Qb^2)^2\Id $} & $ \tangentQR $
                            & \ref{sec:uniaxiality_constrain}\\
            \hline
            $ \IS $
                & $\nullb$
                    & $\LambdabIS$
                        & $ \Qb $ & $ \tangentQR $ 
                            & \ref{sec:isotropic_state}\\
            \hline
            $ \NN $
                & $\lambdaNN \normal$
                    & $\nullb$
                        & $ \Vb\normal $ & $ \tangentS[^0] $ 
                            & \ref{sec:no_normal_flow}\\
            \hline
            $ \NF $
                & $ \LambdabNF $
                    & $\nullb$
                        & $ \Vb $ & $\tangentR$ 
                            & \ref{sec:no_flow}\\
            \hline
        \end{tabular}
        \caption{Generalized constraint forces $ \Fb_{\gamma}\in\tangentR $, $  \Hb_{\gamma}\in\tangentQR $ and $  \Cb_{\gamma}\in\Vcal_{\gamma} $ \wrt\ Lagrange parameter in $\Vcal_{\gamma} $. 
            These terms apply in the Surface Beris-Edwards models \eqref{eq:model_lagrange_multiplier} in some circumstances.
            }
        \label{tab:constraints}
    \end{table}
    
    All modular combinations of \eqref{eq:model_lagrange_multiplier} are thermodynamically consistent. This follows by construction from the Lagrange-D'Alembert principle. The rate of the total energy $ \energyTOT := \energyK + \potenergy $, comprising the kinetic energy $ \energyK $ and the potential energy $ \potenergy $, is given by $\ddt\energyTOT = -2 ( \energyIM + \energyNV ) \le 0$, where $ \energyIM \ge 0 $ is the immobility (Sec.~\ref{sec:immobility}) and $ \energyNV  \ge 0 $ is the nematic viscosity (Sec.~\ref{sec:nematic_viscosity}) flux potential. This means that temporal changes in the total energy are solely caused by the considered energy exchange mechanism and the system exhibits dissipation. The detailed computation is provided in Sec. \ref{sec:energy_rate}.

    \subsection{Special Cases of the General Surface Models}
    \label{sec:model_general_cases}
    
    Various related, but mostly severely simplified models of \ref{eq:model_lagrange_multiplier} have been considered in the literature. We demonstrate that these models can be obtained as special cases of \eqref{eq:model_lagrange_multiplier} by considering constraints on the nematic field and/or the flow field. Establishing the connection to existing models not only justifies \eqref{eq:model_lagrange_multiplier} it also helps to understand the various dependencies and geometric coupling mechanisms. We will here consider $ \Cset \subset \{ \IS, \NN, \NF \} $ as well as special choices of material parameters. 
    
    \subsubsection{Isotropic State Constraint}

    For the isotropic nematic state, \ie\ $ \IS\in\Cset $, we obtain
    \begin{align}\label{eq:NSE}
        \rho\Dmat\Vb 
           &=  -\GradC p 
                       + \fnorBE\normal
                       + \DivC\SigmabNV^0 
                       + \sum_{\gamma\in\Cset_{\IS}} \Fb_{\gamma} \formComma
       & 0 &= \DivC\Vb \formComma
       & \nullb &= \Cb_{\gamma}, \quad \forall \gamma\in\Cset_{\IS} \formComma
    \end{align}
    see Sec. \ref{sec:isotropic_state}. The only remaining considerable constraints are related to the flow field. For $ \Cset_{\IS}=\emptyset $, \eqref{eq:NSE} is the model for a fluid deformable surface \cite{Arroyo_2009,Reuther_2020,Krause_2023}. Such models have been established as model systems for lipid bilayers, the cellular cortex or epithelia monolayers. They exhibit a solid-fluid duality, while they store elastic energy when bent, as solid shells, they flow as viscous two-dimensional fluids under in-plane shear. In the presence of curvature, any shape change is accompanied by a tangential flow and, vice-versa, the surface deforms due to tangential flow. The Stokes limit of \eqref{eq:NSE} has been considered in \cite{Torres-Sanchez_2019}.
    
    Forcing a geometrical stationary surface, \ie\ $ \Cset_{\IS}=\{\NN\} $, we could omit the normal part of \ref{eq:NSE}, see Sec. \ref{sec:no_normal_flow}. This yields the inextensible surface Navier-Stokes equation \cite{Reuther_2015,Koba_2017,Reuther_MMS_2018,Koba_2018,Jankuhn_2018,Reuther_PF_2018,Miura_2018} 
    \begin{align}\label{eq:NSE_tangential}
        \rho \dot{\vb} = -\nabla p + \coeffIF\div(\nabla\vb + (\nabla\vb)^T) \formComma \qquad
        0 = \div\vb 
    \end{align}
    in the tangential differential calculus notation.
    Note that in this case it is common to take the Eulerian perspective, which gives $ \dot{\vb} = \partial_t\vb + \nabla_{\vb}\vb $. For flat surfaces \ref{eq:NSE_tangential} leads to the usual incompressible Navier-Stokes equation. The Stokes limit of \eqref{eq:NSE_tangential} has been intensively studied numerically, see \cite{Jankuhn_JNM_2021,Brandner_SIAMJSC_2022,Hardering_2023} and the references therein.

    \subsubsection{No-flow Constraint}

    The other severely simplified case is for the constraint $ \Vb=\nullb $, \ie\ $ \NF\in\Cset $.
    Sec. \ref{sec:no_flow} leads to
    \begin{align}\label{eq:L2flow}
        \widetilde{M}\Dt\Qb
            = \HbEL + \HbTH + \sum_{\gamma\in\Cset_{\NF}}\Hb_{\gamma} \formComma
        \qquad \nullb = \Cb_{\gamma} \formComma \quad \forall\gamma\in\Cset_{\NF} \formComma 
    \end{align}
     where the only remaining considerable constraints are related to restrictions on the nematic field, \ie\ $ \Cset_{\NF} \subset \{\SC,\CB, \UN\} $.
     This corresponds to the $ \hil $-gradient flow of the surface Landau-de Gennes energy for $ \R^3 $-valued Q-tensor fields restricted to the surface for a certain choice of constraints.
     Note that on a materially stationary surface there is no choice of a tensorial time-derivative, \ie\ we define $ \Dt\Qb:= \Dmat\Qb=\Djau\Qb  $.
     Even if an arbitrary observer is conceivable, it is probably most common to use a Lagrangian or Eulerian perspective.
     Both are trivially equal for $ \Vb=\nullb $ and it is $ \Dt\Qb=\partial_t\Qb $ valid in these cases.
     For the unrestricted biaxial case ($ \Cset_{\NF}=\emptyset $), eq. \eqref{eq:L2flow} is essentially the thin-film limit of the bulk Landau-de Gennes model \cite{Virga_2018,Chen_2015,Majumdar_2010,Kirr_2014} under homogeneous Neumann boundary conditions, \cf\ \cite{NitschkeVoigt_2023, Nitschke_2018}.
     The surface conforming case ($ \Cset_{\NF} \subset \{ \SC, \CB \} $) is discussed in Section \ref{sec:surface_conforming_no_flow} below.

    \subsubsection{Nematodynamics on Geometrically Stationary Surfaces}

    A rather general nematic model on a geometrically stationary surface ($ \NN \in \Cset $) has been considered in \cite{BouckNochettoYushutin_2022}. It is a isotropic viscous ($\xi=0$) Jaumann-Model ($ \Phi=\jau $) for an Eulerian observer ($ \Vb_{\!\ofrak} = \vb_{\ofrak} = \nullb $).  
    Considering the remaining Lagrangian Q-tensor force $ \Hb_{\potenergy}=\HbEL+\HbTH $, the Q-tensor eq. \eqref{eq:model_lagrange_multiplier_moleculareq} becomes
    \begin{align*}
        M\Djau\Qb 
            &=  M\left( \partial_t\Qb + (\nablaC\Qb)\vb - \Abcal[\vb]\Qb + \Qb \Abcal[\vb] \right)
             =\Hb_{\potenergy} \formComma
    \end{align*}
    where $ \Abcal[\vb]= \frac{1}{2}\left( \nabla\vb - (\nabla\vb)^T \right) + \nu\otimes\shop\vb - \shop\vb\otimes\nu $ \eqref{eq:Abcal} is the sum of the ``covariant spin tensor'' and the ``star spin tensor'' in compliance with the naming in \cite{BouckNochettoYushutin_2022}. 
    Subsequently, substituting this into $ \SigmabIM[\jau] $ and using the skew-symmetric tensor field $ \Sigmab_{\textup{E}}:=\Qb\Hb_{\potenergy} - \Hb_{\potenergy}\Qb $, the fluid eq. \eqref{eq:model_lagrange_multiplier_floweq} results in
    \begin{align*}
        \rho\left(\partial_t\vb +\nabla_{\vb}\vb  \right) 
            &= -\nabla p + \coeffIF\div\left( \nabla\vb + (\nabla\vb)^T \right) + \fb_{\textup{L}} + \fb_{\textup{E}}\formComma\\
        \fb_{\textup{L}}
            &:=  \IdS\left( \DivC\SigmabEL + \GradC\pTH \right)
             = - \Hb_{\potenergy}\dbdot\nablaC\Qb\formComma\\ 
        \fb_{\textup{E}}
            &:= \IdS\DivC\SigmabIM[\jau]
             = \div\projS[^2]\Sigmab_{\textup{E}} - 2\shop\Sigmab_{\textup{E}}\normal  \formComma
    \end{align*}
    where we omit the normal equation according to Section \ref{sec:no_normal_flow}.
    Following the naming in \cite{BouckNochettoYushutin_2022}, $ \fb_{\textup{L}}\in \tangentS $ is the ``Leslie force'',  $\fb_{\textup{E}}\in \tangentS$ consists of the ``Ericksen force'' and the ``star force'', and  $\Sigmab_{\textup{E}}\in\tangentAR$ is the  ``Ericksen stress''. \footnote{To proof the identity $  \IdS\left( \DivC\SigmabEL + \GradC\pTH \right) = - \Hb_{\potenergy}\dbdot\nablaC\Qb $ requires some longer calculation, which we omit here. However, the identity can easily be checked using computer algebra systems.} The surface conforming case ($ \Cset_{\NN} \subset \{ \SC, \CB \} $) is discussed in Section \ref{sec:Nematodynamics on Geometrically Stationary Surfaces} below.
    
    \subsection{Surface Conforming Q-Tensor Models}\label{sec:models_conforming}
    We consider surface conforming Q-tensor fields $ \Qb\in\tangentCQR<\tangentQR $. In the general model \eqref{eq:model_lagrange_multiplier} this requires $ \SC\in\Cset $. Physically this means that the biaxial directors of the nematic field are forced to point either in normal or in tangential direction of the surface.
    An associated approach for surface conforming Q-tensor fields is the orthogonal decomposition $ \Qb= \CQdepl(\qb,\beta) $ \eqref{eq:surface_conforming_ansatz},
    where the eigenvalue in normal direction $ \beta\in\tangentS[^0] $ and the tangential Q-tensor part $ \qb\in\tangentQS $ are uniquely and mutually independent given for all $ \Qb\in\tangentCQR $.
    In Sec. \ref{sec:surface_conforming_constrain} we derive a reduced model without the Lagrange parameter $ \lambdabSC $ for surface conformity, where the associated constraint force, which appears in the fluid equation, is given explicitly.
    Since we use a tangential related decomposition for the Q-tensor field anyway, the following surface conforming Beris-Edwards models are given in terms of tangential calculus:
    
    \paragraph{Surface Conforming Beris-Edwards Models:} \textit{Find the tangential and normal material velocity fields $ \vb\in\tangentS $ and $ \vnor\in\tangentS[^0] $, tangential Q-tensor field $ \qb\in\tangentQS $, normal eigenvalue field $ \beta\in\tangentS[^0] $, pressure field $ p\in\tangentS[^0] $ and Lagrange parameter fields $ \lambdab_{\gamma}\in\Vcal_{\gamma} $ for all $ \gamma\in\Cset_{\SC} $ \st}
    \begin{subequations}\label{eq:model_conforming}
       \begin{gather}
       \begin{align}
            \rho\ab 
                &= \nabla(\pTH - p)   + \div\left( \sigmabEL + \sigmabIM  + \sigmabNV^0 + \xi\sigmabNV^1 +  \xi^2\sigmabNV^2 \right) \notag\\
                &\quad + \left(2\shop\qb-3\beta\shop\right)\left( \zetabEL + \zetabIM \right) \label{eq:model_conforming_tangentialeq}
                       + \sum_{\gamma\in\Cset_{\SC}}\fb_{\gamma}\formComma\\
            \rho\anor 
                &= (\pTH - p)\meanc + \fnorBE  + \shop\dbdot\left( \sigmabEL + \sigmabIM  + \sigmabNV^0 + \xi\sigmabNV^1 + \xi^2\sigmabNV^2 \right) \notag\\
                &\quad - \div\left( \left(2\qb-3\beta\IdS\right) \left( \zetabEL + \zetabIM \right) \right)
                       + \sum_{\gamma\in\Cset_{\SC}} \fnor[\gamma] \formComma \label{eq:model_conforming_normalseq}
       \end{align}\\
       \begin{align}
            \widetilde{M}\timeD\qb \label{eq:model_conforming_qtensor}
                &= \hbEL + \hbTH + \xi\hbNV^1 + \xi^2\widetilde{\hb}^{2,\Phi}_{\NV} + \sum_{\gamma\in\Cset_{\SC}} \hb_{\gamma} \formComma \\
            \widetilde{M} \dot{\beta} \label{eq:model_conforming_betaeq}
                &= \omegaEL + \omegaTH  + \xi^2\widetilde{\omega}^{2}_{\NV} + \sum_{\gamma\in\Cset_{\SC}} \omega_{\gamma}\formComma
       \end{align}\\
       \div\vb = \vnor\meanc \formComma\\
       \forall\gamma\in\Cset_{\SC}:\quad  \nullb = \Cb_{\gamma} \label{eq:model_conforming_constraineqs}
       \end{gather}
    \end{subequations}
    \textit{holds for $ \dot{\rho}=0 $ and given initial conditions for $ \vb $, $ \vnor $, $ \qb $, $ \beta $ and mass density $ \rho\in\tangentS[^0] $.}
    
    Mandatory quantities are given in Table \ref{tab:forces_conforming} and optional constraint quantities are given in Table \ref{tab:constraints_conforming}. 
    The choice of $ \Phi $ for $  \sigmabNV^1  $ and $  \sigmabNV^2  $ is optional. Both given representations state equal tensor fields by \eqref{eq:jauq}. The relation to stress and force fields as given in Sec. \ref{sec:nematic_viscosity} is $ \sigmabNV = \sigmabNV^0 + \xi\sigmabNV^1 + \xi^2\sigmabNV^2 $, $ \hbNV = \xi\hbNV^1 + \xi^2\left(\widetilde{\hb}^{2,\Phi}_{\NV} -\frac{\coeffIF}{2}\timeD\qb\right)  $ and $ \omegaNV = \xi^2\left( \widetilde{\omega}^{2}_{\NV} - \frac{\coeffIF}{2}\dot{\beta} \right) $, which in total do not depend on $ \Phi $.
    We exclude the isotropic state constraint, \ie\ $ \Cset_{\SC}\subset\{ \CB, \UN, \NN, \NF\} $, since that would only yield the model for fluid deformable surfaces \eqref{eq:NSE} \wrt\ the tangential calculus notation, \resp\ the surface Navier-Stokes equation \eqref{eq:NSE_tangential} if the material flow is forced to be tangential ($ \NN $). Note that all considered stress tensors depending on the Q-tensor field become fully tangential in \eqref{eq:model_conforming}, \ie\  they are in $ \tangentS[^2] $.

     \begin{table}
        \centering
        \renewcommand{\arraystretch}{1.3}
        \begin{tabular}{|c|c|c|c|}
        \hline
            \multicolumn{2}{|c|}{Identifier}
                    & Expression
                        & Sec., Eq., Ref. \\
        \hline
            \multicolumn{2}{|c|}{$ \ab $}
                    & $ \IdS\Dmat\Vb = (\partial_t v^i)\partial_i\para_{\!\ofrak} + \nabla_{\vb-\vb_{\!\ofrak}}\vb + \nabla_{\vb}\vb_{\!\ofrak} - \vnor\left( \nabla\vnor + 2\shop\vb \right) $
                        & \multirow{2}{*}{\eqref{eq:DmatV}, \cite{NitschkeVoigt_JoGaP_2022,NitschkeVoigt_2023}}  \\
        \cline{1-3}
            \multicolumn{2}{|c|}{$ \anor $}
                    & $ \normal\Dmat\Vb = \partial_t\vnor + \nabla_{2\vb-\vb_{\!\ofrak}}\vnor + \shop(\vb,\vb) $ &\\
        \hline
            \multirow{2}{*}{$ \timeD\qb $}
                & $ \Phi=\jau $
                    & $ \timeJ\qb = \dot{\qb} - \Ab[\Vb]\qb + \qb\Ab[\Vb] $
                        & \ref{sec:immobility}, \eqref{eq:jauq}, \cite{NitschkeVoigt_JoGaP_2022, NitschkeVoigt_2023} \\
        \cline{2-4}
                & $ \Phi=\mfrak $
                    & $ \dot{\qb} = (\partial_t q^{ij})\partial_i\para_{\!\ofrak}\otimes\partial_j\para_{\!\ofrak} + \nabla_{\vb-\vb_{\!\ofrak}}\qb + \Gb[\Vb_{\!\ofrak}]\qb + \qb\Gb^T[\Vb_{\!\ofrak}] $ 
                        & \ref{sec:immobility}, \eqref{eq:dotq}, \cite{NitschkeVoigt_JoGaP_2022, NitschkeVoigt_2023}\\
        \hline
            \multicolumn{2}{|c|}{$ \sigmabEL $}
                    & \splitcelltab{$ -L\big( (\nabla\qb)^{T} \dbdot\nabla\qb + \frac{3}{2}\nabla\beta\otimes\nabla\beta 
                                              - \frac{1}{4}\left( 2\normsq{\tangentS[^3]}{\nabla\qb} + 3\normsq{\tangentS[^2]}{\nabla\beta}  \right)\IdS  $ \\
                                    $ - 6\gaussc\beta\qb + \frac{1}{2}\left( 2\meanc\Tr\qb^2 - 12\beta \qb \dbdot \shop  + 9 \meanc\beta^2 \right)\left( \shop - \frac{\meanc}{2}\IdS \right)\big) $}
                        & \multirow{2}{*}{\ref{sec:elastic}} \\
        \cline{1-3}
            \multicolumn{2}{|c|}{$ \zetabEL $}
                    & $ L\left( 2(\nabla\qb)\dbdot\shop + \qb\nabla\meanc - 3\shop\nabla\beta - \frac{3}{2}\beta\nabla\meanc \right) $ &\\
        \hline
            \multicolumn{2}{|c|}{$ \hbEL $}
                    & $ L \left(  \Delta\qb - (\meanc^2-2\gaussc)\qb + 3\beta \meanc \left( \shop - \frac{\meanc}{2}\IdS \right) \right) $
                        & \multirow{2}{*}{\ref{sec:elastic}, \cite{NitschkeVoigt_2023}}\\
        \cline{1-3}
            \multicolumn{2}{|c|}{$ \omegaEL $}
                    & $ L \left( \Delta\beta + 2 \meanc\shop \dbdot \qb  - 3\beta \left( \meanc^2 - 2\gaussc \right)  \right) $ &\\
        \hline
            \multicolumn{2}{|c|}{$ \pTH $}
                    & $ \frac{1}{2}\left( 2a - 2b\beta + c\left( \Tr\qb^2 + 3\beta^2 \right) \right)\Tr\qb^2
                         +\frac{1}{8}\left( 12 a + 4b\beta + 9c\beta^2 \right)\beta^2 $
                        & \ref{sec:thermotropic} \\
        \hline
            \multicolumn{2}{|c|}{$ \hbTH $}
                    & $ -\left(2a - 2b\beta + 3c\beta^2  + 2c\Tr\qb^2\right)\qb $
                        & \multirow{2}{*}{\ref{sec:thermotropic}, \cite{NitschkeVoigt_2023}} \\
        \cline{1-3}
            \multicolumn{2}{|c|}{$ \omegaTH $}
                    & $ -\left(2a + b\beta + 3c\beta^2  + 2c\Tr\qb^2\right)\beta + \frac{2}{3} b \Tr\qb^2 $ &\\
        \hline
            \multicolumn{2}{|c|}{$ \fnorBE $}
                    & $ -\kappa\left( \Delta\meanc + \left( \meanc - \meanc_{0} \right)\left( \frac{1}{2}\meanc(\meanc + \meanc_{0}) - 2\gaussc  \right) \right) $
                        & \ref{sec:bending}, \cite{BachiniKrauseNitschkeVoigt_2023} \\
        \hline
            \multirow{2}{*}{$ \sigmabIM $}
                & $ \Phi=\jau $
                    & $ M \left( \qb\timeJ\qb - (\timeJ\qb)\qb \right) $
                        & \multirow{4}{*}{\ref{sec:immobility}} \\
        \cline{2-3}
                & $ \Phi=\mfrak $
                    & $ \nullb $ & \\
        \cline{1-3}
            \multirow{2}{*}{$ \zetabIM $}
                & $ \Phi=\jau $
                    & $ \nullb $ &\\
        \cline{2-3}
                & $ \Phi=\mfrak $
                    & $  -M \left(\nabla\vnor + \vb\shop\right)\left( \qb - \frac{3}{2}\beta\IdS \right) $ & \\
        \hline
            \multicolumn{2}{|c|}{$ \sigmabNV^0 $}
                    & $ 2\coeffIF\Sb[\Vb] $
                        & \ref{sec:nematic_viscosity}, \cite{BachiniKrauseNitschkeVoigt_2023}\\
        \hline
            \multirow{2}{*}{$ \sigmabNV^1 $}
                &  $ (\Phi=\jau) $
                    & $ -\coeffIF\left(\timeJ\qb - \frac{\dot{\beta}}{2}\IdS + 3\qb\Sb[\Vb] + \Sb[\Vb]\qb - 2\beta\Sb[\Vb] \right) $
                        &\multirow{8}{*}{\ref{sec:nematic_viscosity}} \\
        \cline{2-3}
                & $ (\Phi=\mfrak) $
                    & $ -\coeffIF\left( \dot{\qb} - \frac{\dot{\beta}}{2}\IdS + \qb(2\Gb[\Vb]+\Gb^T[\Vb]) + \Gb^T[\Vb]\qb - 2\beta\Sb[\Vb] \right)  $ & \\
        \cline{1-3}
            \multicolumn{2}{|c|}{$ \hbNV^1 $}
                    & $ \coeffIF\Sb[\Vb] $ &\\
        \cline{1-3}
            \multirow{2}{*}{$ \sigmabNV^2 $}
                &  $ (\Phi=\jau) $
                    & \splitcelltab{$ \coeffIF\big( \qb\timeJ\qb - \frac{1}{2}\left( \beta\timeJ\qb + \dot{\beta}\qb \right)
                            + \frac{1}{4}\beta\dot{\beta}\IdS + \qb\Sb[\Vb]\qb$ \\
                      $-\frac{1}{2}\beta\left( 3\qb\Sb[\Vb] + \Sb[\Vb]\qb \right)  
                        + \frac{1}{2}\left( \Tr\qb^2 + \beta^2 \right)\Sb[\Vb] \big) $} &\\
        \cline{2-3}
                & $ (\Phi=\mfrak) $
                    & \splitcelltab{$ \coeffIF\big(  \qb\dot{\qb} - \frac{1}{2}\left( \beta\dot{\qb} + \dot{\beta}\qb \right) 
                              + \frac{1}{4}\beta\dot{\beta}\IdS + \qb\Gb^T[\Vb]\qb$ \\
                      $ -\frac{1}{2}\beta\left( \qb(2\Gb[\Vb]+\Gb^T[\Vb]) + \Gb^T[\Vb]\qb \right)  
                      + \frac{1}{2}(\Tr\qb^2)\Gb[\Vb] + \frac{1}{2}\beta^2\Sb[\Vb]\big) $} &\\
        \cline{1-3}
            \multirow{2}{*}{$ \widetilde{\hb}^{2,\Phi}_{\NV} $}
                & $ \Phi=\jau $
                    & $ -\frac{\coeffIF}{2} \left( \qb\Sb[\Vb] + \Sb[\Vb]\qb - (\qb\dbdot\Gb[\Vb])\IdS - \beta\Sb[\Vb]\right) $& \\
        \cline{2-3}
                & $ \Phi=\mfrak $
                    & $ -\frac{\coeffIF}{2} \left( \qb\Gb[\Vb] + \Gb^T[\Vb]\qb -(\qb\dbdot\Gb[\Vb])\IdS - \beta\Sb[\Vb] \right) $& \\
        \cline{1-3}
            \multicolumn{2}{|c|}{$ \widetilde{\omega}^{2}_{\NV} $}
                    & $ -\frac{\coeffIF}{6} \left( 3\dot{\beta} - 2 \qb\dbdot\Gb[\Vb] \right) $ &\\
        \hline
        \end{tabular}
        \caption{Necessary terms for the Surface Conforming Beris-Edwards models \eqref{eq:model_conforming} for a consistent choice $ \Phi\in\{\jau,\mfrak\} $.
                These representations comprise the tangential deformation gradient $ \Gb[\Vb] = \nabla\vb - \vnor\shop $ \eqref{eq:Gb} of the material velocity $ \Vb=\vb+\vnor\normal $.
                $ \Sb[\Vb] $ \eqref{eq:Sb} and $\Ab[\Vb]$ \eqref{eq:Ab} are its symmetric and skew-symmetric part.
                Time derivatives are determined \wrt\ an observer velocity $ \Vb_{\!\ofrak}= \vb_{\ofrak} + \vnor\normal $. }
        \label{tab:forces_conforming}
    \end{table}

    \begin{table}
        \centering
        \renewcommand{\arraystretch}{1.2}
        \begin{tabular}{|c|c|c|c|c|c|}
            \hline
            $ \gamma $
                & $\left(\fb_{\gamma}, \fnor[\gamma]\right)$
                    & $ \left(\hb_{\gamma}, \omega_{\gamma}\right) $
                        & $  \Cb_{\gamma} $ & $ \Vcal_{\gamma} $
                            & Sec.\\
            \hline
            $ \CB $
                & $ \left(\nullb, 0\right) $
                    & $ \left(\nullb, -\frac{2}{3}\lambdaCB\right) $
                        & $\beta - \beta_0 $ & $ \tangentS[^0] $
                            & \ref{sec:constant_beta}\\
            \hline
            $ \UN $
                & $\left(\nullb, 0\right)$
                    & \splitcelltab{$\Big(-6\beta\qb\lambdabUN\qb + \frac{4}{3}(\Tr\qb^2)\projQS(\lambdabUN\qb)$\\
                              $ + \left( 5\beta\lambdabUN\dbdot\qb + \lambdabotUN\Tr\qb^2\right)\qb$\\
                              $ - \frac{1}{4}\beta\left( 14\Tr\qb^2 - 9\beta^2 \right)\lambdabUN$,\\
                              $ \frac{1}{3}(\Tr\qb^2)\left( 2\lambdabUN\dbdot\qb - 9\lambdabotUN\beta \right)\Big) $}
                        & \splitcelltab{$  \Big( \left( 2\Tr\qb^2 - 9\beta^2 \right)\beta\qb $,\\
                                    $\left( 2\Tr\qb^2 - 9\beta^2 \right)\Tr\qb^2 \Big)$} & $ \tangentQS\times\tangentS[^0] $
                            & \ref{sec:uniaxiality_constrain}\\
            \hline
            $ \NN $
                & $\left(\nullb, \lambdaNN\right)$
                    & $\left(\nullb, 0\right)$
                        & $ \vnor $ & $ \tangentS[^0] $ 
                            & \ref{sec:no_normal_flow}\\
            \hline
            $ \NF $
                & $ \left( \lambdab_{\NF}, \lambda^{\bot}_{\NF} \right) $
                    & $\left(\nullb, 0\right)$
                        & $ \left( \vb, \vnor \right) $ & $\tangentS\times\tangentS[^0]$ 
                            & \ref{sec:no_flow}\\
            \hline
        \end{tabular}
        \caption{Generalized constraint forces $ (\fb_{\gamma}, \fnor[\gamma])\in\tangentS\times\tangentS[^0]\cong\tangentR $, 
            $  (\hb_{\gamma}, \omega_{\gamma})\in\tangentQS\times\tangentS[^0] \cong \tangentCQR $ and $  \Cb_{\gamma}\in\Vcal_{\gamma} $ 
            \wrt\ Lagrange parameter in $ \Vcal_{\gamma} $. 
            These terms apply in the Surface Conforming Beris-Edwards model \eqref{eq:model_conforming} in some circumstances. For the constraints $ \UN $ and $ \NF $ we represent the Lagrange parameter by a pair of Lagrange parameters, \ie\ $ \LambdabUN\cong\left( \lambdabUN, \lambdabotUN \right) $
            and $ \LambdabNF\cong\left( \lambdab_{\NF}, \lambda^{\bot}_{\NF} \right) $, where $ \lambdabUN\in\tangentQS $, $ \lambdab_{\NF}\in\tangentS $ and $ \lambdabotUN,\lambda^{\bot}_{\NF}\in\tangentS[^0] $.
            }
        \label{tab:constraints_conforming}
    \end{table}

    \subsection{Special Cases of the Surface Conforming Models}\label{sec:surface_conforming_special_cases}
    
    Again, we consider simplified models of in the literature and demonstrate their relation to \eqref{eq:model_conforming} by considering constraints on the nematic field  and/or the flow field $ \Cset_{\SC} \subset \{ \CB, \UN, \NN, \NF \} $ as well as special choices of material parameters. 

    \subsubsection{No-flow Constraint}\label{sec:surface_conforming_no_flow}

    For the setting $ \Cset_{\SC} = \{\NF\} $, Sec. \ref{sec:no_flow} yields the surface conforming Landau-de Gennes model
    \begin{align}
        \widetilde{M}\dot{\qb} 
            &= \hbEL + \hbTH  \formComma
        &\widetilde{M} \dot{\beta} 
            &= \omegaEL + \omegaTH  \formPeriod
    \end{align}
    A Lagrangean/Eulerian observer ($\dot{\qb}=\partial_t\qb$, $\dot{\beta}=\partial_t\beta$) and an isotropically viscous ($ \xi = 0 $) setting results in the ``free $ \beta $-model'' in \cite{Nestler_2020}, but without the penalty force to maintain tangential uniaxiality ($ 2\Tr\qb^2 = 9\beta^{2} $) approximately, \ie\ for ``$ \omega_{\beta} = 0 $'' in \cite{Nestler_2020}. 
    We get the other two models in \cite{Nestler_2020} by $ \Cset_{\SC} = \{\CB, \NF\} $ for $\beta\equiv\beta_{0}\in\{0,  -\frac{1}{3}S^*_{\THT}\}$ with 
    $ S^*_{\THT}:= \frac{1}{4 c} (\sqrt{b^2 - 24 ac}-b) $, see Sec. \ref{sec:constant_beta}.
    This yields
    \begin{align}
        {M}\dot{\qb} 
                    &= \hbEL + \hbTH  \formComma
    \end{align}
    where the solution is either flat-degenerate for $ \beta_0=0 $, as also considered in \cite{Kralj_SM_2011,Jesenek_SM_2015}, or tangential uniaxial at the thermotropic equilibrium $ \hbTH=\nullb $ for $ \beta_{0} = -\frac{1}{3}S^*_{\THT} $, see also \cite{Napoli_PRE_2012,Golovaty_JNS_2015,Golovaty_JNS_2017,Nitschke_2018,Novack_SIAMJMA_2018}. These three models have been derived in \cite{Nestler_2020} as thin film limits of a three-dimensional Landau-de Gennes model. A numerical comparison leads to fundamental differences in responds to curvature. While $\beta_0 = 0$ essentially leads to two-dimensional characteristics and only accounts for intrinsic curvature effects, the models with $\beta_0 \neq 0$ retain characteristics of the three-dimensional Landau-de Gennes model and introduce additional curvature coupling terms. In addition to isotropic geometry coupling, also alignment of the director field with the principle curvature lines is considered. Within the ``free $ \beta $-model'' the geometric coupling is even stronger. Depending on the strength of the curvature also phase transitions can be enforced and without the mentioned penalty force to maintain tangential uniaxiality, curvature can also induce biaxiality, see \cite{Nestler_2020}.

    \subsubsection{Nematodynamics on Flat Surfaces} \label{sec:model_flat_degenerated_flat_surface}
    
    There is a huge literature on nematodynamics in two dimensions. We obtain this situation by considering a geometrically stationary flat surface ($ \vnor=0 $, $ \shop=\nullb $) and compare the resulting model with classical Beris-Edwards models in two dimensions. We stipulate flat degenerate Q-tensor fields ($\beta\equiv0$) and take the Eulerian perspective  $ \Vb_{\!\ofrak}=\vb_{\!\ofrak}=\nullb $.
   With this setting the surface conforming equations \eqref{eq:model_conforming} yield
   \begin{subequations}\label{eq:model_conforming_degenerated_flatsurf}
   \begin{gather}
   \begin{align}
       \rho\left( \partial_t \vb + \nabla_{\vb}\vb \right) \label{eq:model_conforming_degenerated_flatsurf_fluid}
           &= \nabla\left(a\Tr\qb^2 + c\Tr\qb^4 - p\right) \\
           &\quad + \div\left( L\left( (\nabla\qb)^{T} \dbdot\nabla\qb - \frac{\| \nabla\qb \|^2}{2}\Id_{\R^2} \right) + \sigmabIM + \coeffIF(\nabla\vb + \nabla^T\vb) + \xi\sigmabNV^1 + \xi^2\sigmabNV^2 \right) \notag
   \end{align}\\
      \widetilde{M}\timeD\qb \label{eq:model_conforming_degenerated_flatsurf_molecular}
           = L\Delta\qb -2 \left( a + c\Tr\qb^2 \right)\qb + \frac{\coeffIF\xi}{2}(\nabla\vb + \nabla^T\vb) + \xi^2\widetilde{\hb}^{2,\Phi}_{\NV} \\
       \div\vb
           = 0
   \end{gather}
   \end{subequations}
   where $ \Phi\in\{\jau,\mfrak\} $, $\widetilde{M} =  M + \frac{\coeffIF\xi^2}{2}$, $ \nabla^T\vb := (\nabla\vb)^T $, $ [\Id_{\R^2}]_{ij}=\delta_{ij} $ and 
   \begin{gather*}
   \begin{align*}
        \timeD[\jau]\qb &= \timeJ\qb = \dot{\qb} - \frac{1}{2}(\nabla\vb - \nabla^T\vb)\qb + \frac{1}{2}\qb(\nabla\vb - \nabla^T\vb) \formComma
            & \timeD[\mfrak]\qb &= \dot{\qb} = \partial_t\qb + \nabla_{\vb}\qb \formComma
   \end{align*}\\
   \begin{align*}
        \sigmabIM[\jau] &= M \left( \qb\timeJ\qb - (\timeJ\qb)\qb \right) \formComma
            & \sigmabIM[\mfrak] &= \nullb \formComma
   \end{align*}\\
   \begin{align*}
        \sigmabNV^1 &= -\coeffIF\left( \timeJ\qb + \frac{3}{2}\qb(\nabla\vb + \nabla^T\vb) + \frac{1}{2}(\nabla\vb + \nabla^T\vb)\qb \right)
                     = -\coeffIF\left( \dot{\qb} + \qb(3\nabla\vb + \nabla^T\vb) + (\nabla^T\vb)\qb \right)\formComma \\
        \sigmabNV^2 &= \coeffIF\left( \qb\timeJ\qb + \frac{1}{2}\qb( \nabla\vb + \nabla^T\vb )\qb + \frac{\Tr\qb^2}{4}( \nabla\vb + \nabla^T\vb )  \right)
                     = \coeffIF\left( \qb\dot{\qb} + \frac{1}{2}\qb(\nabla^T\vb )\qb + \frac{\Tr\qb^2}{2}\nabla\vb   \right) \formComma
   \end{align*}\\
   \begin{align*}
        \widetilde{\hb}^{2,\jau}_{\NV}	&= -\frac{\coeffIF}{4} \left( \qb(\nabla\vb + \nabla^T\vb) + (\nabla\vb + \nabla^T\vb)\qb - 2(\qb\dbdot\nabla\vb)\Id_{\R^2} \right) \formComma\\
        \widetilde{\hb}^{2,\mfrak}_{\NV} &= -\frac{\coeffIF}{2} \left( \qb\nabla\vb + (\nabla^T\vb)\qb - (\qb\dbdot\nabla\vb)\Id_{\R^2} \right) \formPeriod          
   \end{align*}
   \end{gather*}
   
   Existing flat models, especially if they are derived using the Onsager principle, use rather terms of molecular fields than time derivatives to describe fluid forces.
   In order to get a better comparison, we substitute the Q-tensor eq.  \eqref{eq:model_conforming_degenerated_flatsurf_molecular} into the fluid eq. \eqref{eq:model_conforming_degenerated_flatsurf_fluid}.
   To do that, we have to invert $ \widetilde{M} $ to $ \widetilde{M}^{-1} = \frac{2}{2M + \coeffIF\xi^2}  $ and the equations become non-polynomial in $ \xi $ as a consequence.
   Considering only small nematic anisotropy ($|\xi| \ll 1$) it is justified to linearize the model in $ \xi $.
   If we combine the Lagrangian force fields to $ \hb_\potenergy  = L \Delta\qb - 2(a + c\Tr\qb^2)\qb $ and stress fields to $ \Sigmab_\potenergy = L\left( (\nabla\qb)^{T} \dbdot\nabla\qb - \frac{\| \nabla\qb \|^2}{2}\Id_{\R^2}\right) + (a\Tr\qb^2 + c\Tr\qb^4)\Id_{\R^2}$, the linearized Jaumann model ($ \Phi=\jau $) of \eqref{eq:model_conforming_degenerated_flatsurf} becomes
   \begin{subequations}\label{eq:model_conforming_degenerated_flatsurf_jaumann_linear}
   \begin{gather}
          \rho\left( \partial_t \vb + \nabla_{\vb}\vb \right) 
              = -\nabla p + \div\left( \Sigmab_\potenergy + \coeffIF(\nabla\vb + \nabla^T\vb) 
                + (\qb\hb_\potenergy - \hb_\potenergy \qb)  -  \frac{\coeffIF\xi}{M} \hb_\potenergy -\coeffIF\xi \widetilde{\sigmab} \right) + \landau(\xi^2) \formComma\\
          \timeJ\qb
              = \frac{1}{M}\hb_\potenergy + \frac{\coeffIF\xi}{2M}(\nabla\vb + \nabla^T\vb) + \landau(\xi^2) \formComma\\
          \div\vb
              = 0 \formComma
   \end{gather}
   \end{subequations}
   where $\widetilde{\sigmab} = \qb( \nabla\vb + \nabla^T\vb ) + ( \nabla\vb + \nabla^T\vb )\qb $. By neglecting all $ \landau(\xi^2) $-terms, this model is an anisotropic nematic viscosity model. This is reflected in the presence of the stress tensor $ \widetilde{\sigmab} \in \tangentSymS$, which states the difference between the lower- and upper-convected, \resp\ Jaumann, nematic rate,
   \ie\ it is $\widetilde{\sigmab}=\timeLll\qb - \timeLuu\qb = 2(\timeJ\qb- \timeLuu\qb)$, see \cite{NitschkeVoigt_JoGaP_2022,NitschkeVoigt_2023},
   or alternatively the difference between the Lie derivative of the covariantly and contravariantly represented Q-tensor field in direction of $ \vb $. Nematic viscosity results in a material flow aligning with the Q-tensor field.
   
   For $\xi=0$ the resulting isotropic viscous model of \eqref{eq:model_conforming_degenerated_flatsurf_jaumann_linear} reads
   \begin{subequations}\label{eq:model_conforming_degenerated_flatsurf_jaumann_linear_isotropic}
   \begin{gather}
        \rho\left( \partial_t \vb + \nabla_{\vb}\vb \right) 
              = -\nabla p + \div\left( \Sigmab_\potenergy  + \coeffIF(\nabla\vb + \nabla^T\vb) 
                + (\qb\hb_\potenergy - \hb_\potenergy \qb) \right) \formComma\\
          \timeJ\qb
              = \frac{1}{M}\hb_\potenergy
              = \frac{1}{M} \left( L\Delta\qb -2 \left( a + c\Tr\qb^2 \right)\qb \right) \formComma\\
          \div\vb
              = 0 \formComma
   \end{gather}
   \end{subequations}
   where $ \div\Sigmab_\potenergy = -\hb_\potenergy\dbdot\nabla\qb $. Eq. \eqref{eq:model_conforming_degenerated_flatsurf_jaumann_linear_isotropic} 
    corresponds with a classical Beris-Edwards model in 2D, see e.g. \cite{PhysRevE.63.056702,PhysRevE.67.051705}. This model can be shown to rigorously correspond with the Eriksen-Leslie model \cite{Wang_SIAMJMA_2015}. Analytical results on \eqref{eq:model_conforming_degenerated_flatsurf_jaumann_linear_isotropic} can be found in \cite{paicu2011energy,paicu2012energy,GUILLENGONZALEZ201584,Abels_SIAMJMA_2014,Abels_2016}.
   
   Considering a linearized material-model ($ \Phi=\mfrak $), contrarily to the Jaumann models \eqref{eq:model_conforming_degenerated_flatsurf_jaumann_linear} or \eqref{eq:model_conforming_degenerated_flatsurf_jaumann_linear_isotropic}, would not exhibit the ``broken symmetry''	term $ \qb\hb_\potenergy - \hb_\potenergy \qb $.
    
    \subsubsection{Nematodynamics on Geometrically Stationary Surfaces} \label{sec:Nematodynamics on Geometrically Stationary Surfaces}
   
    We consider a geometrically stationary but curved surface ($ \shop\neq\nullb $) without genuine uniaxiality, \ie\ $ \Cset_{\SC} = \{\CB,\NN\} $. By neglecting anisotropic viscosity $ (\xi=0) $ the associated surface conforming Beris-Edwards model corresponding to \eqref{eq:model_conforming} yields
    \begin{subequations}\label{eq:model_conforming_NN}
    \begin{gather}
        \rho\ab \label{eq:model_conforming_NN_tangentialeq}
            = \nabla(\pTH - p) + \div\left( \sigmabEL + \sigmabIM  + \sigmabNV^0 \right)
                + \left(2\shop\qb-3\beta_0\shop\right)  \left( \zetabEL + \zetabIM \right) \formComma\\
        M\timeD\qb \label{eq:model_conforming_NN_molecular}
            = \hbEL + \hbTH  \formPeriod \\
        \div\vb = 0
    \end{gather}
    \end{subequations}
    We thereby have substituted $ \vnor=0 $, instead of using it as a constraint equation in \eqref{eq:model_conforming_constraineqs}, and omitted \eqref{eq:model_conforming_normalseq} according to Sec. \ref{sec:no_normal_flow}.
    This gives an alternative to similar models, where the surface conformity as well as the constant normal eigenvalue of the Q-tensor field is implemented only approximately by penalizations of these constraints. A corresponding penalized Jaumann model ($\Phi=\jau$) to \eqref{eq:model_conforming_NN} is considered in \cite{BouckNochettoYushutin_2022}.

\section{Derivation} \label{sec:derivation}

\subsection{Lagrange-D'Alembert Principle}\label{sec:lagrangedalembert}

The derivation of the governing equations is based on the integral Lagrange-D'Alembert principle. In addition to simply transferring the classical principle to our setting of nematodynamics of fluid deformable surfaces, or superimposing Hamilton's principle of stationary action and Onsager's variational principle of least energy dissipation (\cf\ \cite{Doi_2011}), a theoretical foundation is provided by \cite{Marsden_2010}. Based on this, 
the Lagrange-D'Alembert principle has been used to derive themodynamically consistent models for two-phase surface deformable surfaces \cite{BachiniKrauseNitschkeVoigt_2023}, which provides the basis for a detailed investigation of the strong interplay of surface phase
composition, surface curvature, and surface hydrodynamics. We here follow a similar approach. 
We use the action functional $ \Acal = \int_{\Tcal}\Lfrak \dt $. 
This is the time integral of the Lagrangian $ \Lfrak=\energyK - \potenergy $, comprising the kinetic energy $ \energyK $ and a (free) potential energy $ \potenergy $.
Additionally, an essential component is the energy flux potential $ \fluxpotential $, which governs the energy exchange of an open system.
These dependencies are influenced by both state and process variables.
The primary state variables include the parameterization  $ \para $ for determining the positions of material particles in $ \R^3 $, which collectively form the surface $ \surf $. 
Additionally, we have the Q-tensor field  $ \Qb\in\tangentQR $ that describes their nematic order.
The main process variables consist of the material velocity field $ \Vb\in\tangentR $ and a Q-tensor field rate $ \Dphi\Qb\in\tangentQR $.
Furthermore, we also consider artificially generated state and process variables in the form of Lagrange multipliers to enforce constraints on both the primary state and process variables.
Formally, if $ \Cset_{\potenergy} $ and $ \Cset_{\fluxpotential} $ represent sets of state and process constraint identifiers,
then $ \lambdab_{\gamma} \in \Vcal_{\gamma} \le \tangentR[^{n_\gamma}] $ are these state or process variables for $ \gamma\in\Cset_{\potenergy} $ or $ \gamma\in\Cset_{\fluxpotential} $. 
We provide further details on this in the respective subsections below.
Eventually, the  integral Lagrange-D'Alembert principle reads
\begin{align}
\MoveEqLeft[6]
\innerH{\Tcal;\tangentR}{\deltafrac{\Acal}{\para}, \Wb}
        +\innerH{\Tcal;\tangentQR}{\deltafrac{\Acal}{\Qb}, \Psib}
        + \sum_{\gamma\in\Cset_{\potenergy}} \innerH{\Tcal;\Vcal_{\gamma}}{ \deltafrac{\Acal}{\lambdab_{\gamma}}, \thetab_{\gamma} } \notag\\
   &= \int_{\Tcal} \innerH{\tangentR}{\deltafrac{\fluxpotential}{\Vb},\Wb}
        + \innerH{\tangentQR}{\deltafrac{\fluxpotential}{\Dphi\Qb}, \Psib} 
        + \sum_{\gamma\in\Cset_{\fluxpotential}} \innerH{\Vcal_{\gamma}}{ \deltafrac{\fluxpotential}{\lambdab_{\gamma}}, \thetab_{\gamma} }\dt
    \label{eq:LDA_time_global}
\end{align}
for all virtual displacements $ \Wb\in\tangentR $, $ \Psib\in\tangentQR $ and $ \thetab_{\gamma}\in \Vcal_{\gamma} $,
where $ \innerH{\Tcal;\Vcal}{\Rb_1, \Rb_2} = \int_{\Tcal} \innerH{\Vcal}{\Rb_1, \Rb_2} \dt$ is the temporally and spatially global inner product for all $ \Rb_1, \Rb_2\in\Vcal\le\tangentR[^n] $, and $ \Tcal=[t_0,t_1] $ is a time interval.
Note that \eqref{eq:LDA_time_global} remains invariant with respect to the actual choice of the rate $ \Dphi\Qb $, as stated in Lemma \ref{lem:ttensor_variational_independence_of_process_variable}, as long as $ \tangentQR $ is closed under the associated time derivative. In specific instances, we interchangeably employ either the material rate $ \Dmat\Qb $ \eqref{eq:DmatQ} or the Jaumann rate $ \Djau\Qb $ \eqref{eq:DjauQ}, depending on which one is more convenient for the given situation.
Even though we only consider $\fluxpotential\ge0$ in this paper, it should be noted that a negative flux potential is permissible, see the discussion about nematic activity in Section \ref{sec:discussion}.
The case $\fluxpotential\ge0$ imply pure dissipation of the total energy $\energyK + \potenergy$, as we can see in Section \ref{sec:energy_rate}.
Therefore it would be reasonable to call $\fluxpotential$ the dissipation potential in this scenario.
It is also noteworthy that the so-called Rayleighian $\fluxpotential+\ddt\potenergy$ can be formulated in terms of the dissipation potential and the rate of the potential energy.

Before we start the detailed discussion of the fundamental components, let's examine the time localization of \eqref{eq:LDA_time_global}. Given that the time integral is outside the variation, and considering that the time interval $ \Tcal $ is arbitrary, we can straightforwardly omit the time integration from the right-hand side of \eqref{eq:LDA_time_global}. 
Furthermore, if we assume that the potential energy $ \potenergy $ is a genuine state potential, meaning it depends solely on the state variables, we can establish that
$ \innerH{\Tcal;\cdot}{\deltafrac{\potenergy}{\cdot}, \cdot} = \int_{\Tcal} \innerH{\cdot}{\deltafrac{\potenergy}{\cdot}, \cdot} \dt $ for all matching terms 
in the left-hand side of \eqref{eq:LDA_time_global}.
Therefore, we can also omit the time integration in this context.
It is important to note that this conclusion is not trivially derived from the definition of the inner product $ \innerH{\Tcal;\cdot}{\cdot, \cdot} $. The handling of variational terms such as "$ \deltafrac{\potenergy}{\Qb} $" is influenced not only by the choice of the inner product but could also include terms arising from time integrations by parts if $ \potenergy $ were not instantaneous.
An example for this issue is the kinetic energy
\begin{align}\label{eq:kinetic_energy}
\energyK 
    &:= \frac{1}{2}\int_{\surf}\rho \Vb\Vb \dS
      =  \frac{1}{2}\normHsq{\tangentR}{\sqrt{\rho}\Vb}
\end{align} 
for point masses, where $ \rho\in\tangentS[^0] $ is the material mass density.
It is non-instantaneous and due to this its action $\Acal_K:=\int_{\Tcal} \energyK \dt$ leads to
\begin{align*}
\innerH{\Tcal;\tangentR}{\deltafrac{\Acal_K}{\para}, \Wb} 
      &= - \int_{\Tcal} \innerH{\tangentR}{\rho\Dmat\Vb, \Wb} \dt
\end{align*}
and $ \innerH{\Tcal;\tangentQR}{\deltafrac{\Acal_K}{\Qb}, \Psib} = 0 $,
if we assume $ \Wb\vert_{\partial\Tcal} = 0 $, variational mass conservation $  \innerH{\tangentR}{\deltafrac{m}{\para}, \Wb} = 0 $ and
temporal mass conservation $ \ddt m\vert_{\mathcal{M}} = 0 $ for the material mass $ m\vert_{\mathcal{M}}=\int_{\mathcal{M}} \rho\ \dS $ for all subsurfaces $ \mathcal{M} \subseteq \surf $. 
See \cite{BachiniKrauseNitschkeVoigt_2023} for more details of this calculation.
Adding all this up yields the temporally local, but still spatially global, formulation of \eqref{eq:LDA_time_global}, which reads
\begin{subequations}\label{eq:LDA_time_local}
\begin{align}
    0 &= \innerH{\tangentR}{\rho\Dmat\Vb, \Wb}
        + \innerH{\tangentR}{\deltafrac{\potenergy}{\para}, \Wb}
        + \innerH{\tangentQR}{\deltafrac{\potenergy}{\Qb}, \Psib} 
        + \sum_{\gamma\in\Cset_{\potenergy}} \innerH{\Vcal_{\gamma}}{ \deltafrac{\potenergy}{\lambdab_{\gamma}}, \thetab_{\gamma} }\notag\\
      &\quad\quad + \innerH{\tangentR}{\deltafrac{\fluxpotential}{\Vb},\Wb}
        +  \innerH{\tangentQR}{\deltafrac{\fluxpotential}{\Dphi\Qb}, \Psib} 
        + \sum_{\gamma\in\Cset_{\fluxpotential}} \innerH{\Vcal_{\gamma}}{ \deltafrac{\fluxpotential}{\lambdab_{\gamma}}, \thetab_{\gamma} }\formComma\\
     0&= \ddt m\vert_{\mathcal{M}} = \int_{\mathcal{M}} \dot{\rho} + \rho\DivC\Vb \dS
            \quad\quad (\forall \mathcal{M} \subseteq \surf)
 \end{align}
\end{subequations}
for all virtual displacements $ \Wb\in\tangentR $, $ \Psib\in\tangentQR $ and $ \thetab_{\gamma}\in \Vcal_{\gamma} \le \tangentR[^{n_\gamma}] $.
Since $ \Wb $, $ \Psib $ and all $ \thetab_{\gamma}$ are arbitrary and mutual independent, the local formulation of \eqref{eq:LDA_time_local} results in
\begin{subequations}\label{eq:LDA}
\begin{align}
    &&\rho\Dmat\Vb &= \totalFb := \totalFb_{\potenergy} + \totalFb_{\fluxpotential}\formComma \label{eq:LDA_fluid}&&\\
    && \nullb &= \Hb := \Hb_{\potenergy} + \Hb_{\fluxpotential}\formComma \label{eq:LDA_molecular}&&\\
    &&\dot{\rho} &= -\rho\DivC\Vb\formComma \label{eq:LDA_conti} &&\\
    &&\nullb &= \Cb_{\gamma} \formComma \qquad \forall \gamma\in\Cset:=\Cset_{\potenergy}\cup\Cset_{\fluxpotential} \formComma\label{eq:LDA_constraints}&&
 \end{align}
\end{subequations}
where $ \totalFb\in\tangentR $, $ \Hb\in\tangentQR $ and $ \Cb_{\gamma}\in \Vcal_{\gamma} $ are generalized applied forces
comprising the Lagrangian forces $ \totalFb_{\potenergy} \in\tangentR $, $ \Hb_{\potenergy}\in\tangentQR $, $ \Cb_{\gamma}\in \Vcal_{\gamma} $ for all $ \gamma\in\Cset_{\potenergy} $,
and flux forces $ \totalFb_{\fluxpotential} \in\tangentR $, $ \Hb_{\fluxpotential}\in\tangentQR $, $ \Cb_{\gamma}\in \Vcal_{\gamma} $ for all $ \gamma\in\Cset_{\fluxpotential} $.
These forces are given in such a way that
\begin{subequations}\label{eq:Lagrange_forces_def} \label{eq:flux_forces_def}
\begin{align}
    \innerH{\tangentR}{\totalFb_{\potenergy}, \Wb} &= -\innerH{\tangentR}{\deltafrac{\potenergy}{\para}, \Wb}\formComma
        &\innerH{\tangentR}{ \totalFb_{\fluxpotential} ,\Wb} &= -\innerH{\tangentR}{\deltafrac{\fluxpotential}{\Vb},\Wb}\formComma \\
    \innerH{\tangentQR}{\Hb_{\potenergy}, \Psib} &= -\innerH{\tangentQR}{\deltafrac{\potenergy}{\Qb}, \Psib}\formComma
        &\innerH{\tangentQR}{\Hb_{\fluxpotential}, \Psib} &= -\innerH{\tangentQR}{\deltafrac{\fluxpotential}{\Dphi\Qb}, \Psib} \formComma \\
    \innerH{\Vcal_{\alpha}}{\Cb_{\gamma}, \Psib} &= -\innerH{\Vcal_{\alpha}}{\deltafrac{\potenergy}{\lambdab_{\gamma}}, \thetab_{\gamma}}  \!\!\!\formComma \; \forall \gamma\in\Cset_{\potenergy}
        &\innerH{\Vcal_{\alpha}}{\Cb_{\gamma}, \Psib} &= -\innerH{\Vcal_{\alpha}}{\deltafrac{\fluxpotential}{\lambdab_{\gamma}}, \thetab_{\gamma}} \!\!\! \formComma \; \forall \gamma\in\Cset_{\fluxpotential}\label{eq:Cgamma_def}
\end{align} 
\end{subequations} 
hold for all $ \Wb\in\tangentR $, $ \Psib\in\tangentQR $ and $ \thetab_{\gamma}\in \Vcal_{\gamma} \le \tangentR[^{n_\gamma}] $.
We derive and discuss the Lagrangian forces in Section \ref{sec:lagrangian_forces} and the flux forces in Section \ref{sec:flux_forces}.
These sections reveal a particularly intriguing issue.
Collectively, the fluid force $ \totalFb $ yields
$ \innerH{\tangentR}{\totalFb, \Wb} = \innerH{\tangentR}{\Fb, \Wb} + \innerH{\tangentQR}{\Hb, \gauge{\Wb}\Qb} $ for all $ \Wb\in\tangentR $ and a certain $ \Fb\in\tangentR $. 
The directional deformation $ \gauge{\Wb}\Qb $ \eqref{eq:gaugeR} remains undetermined a priori. The choice of determining $ \gauge{\Wb}\Qb $ in the form of a gauge of surface independence is arbitrary. This is due to the fact that the Q-tensor equation \eqref{eq:LDA_molecular} reveals $ \Hb = \nullb $. Consequently, it follows that $ \totalFb = \Fb $ as long as $ \gauge{\Wb}\Qb $ is finite. 
Therefore, the local Lagrange D'Alembert equations \eqref{eq:LDA} are independent of the choice of a suitable $ \gauge{\Wb}\Qb $, which is in stark contrast to $\hil$-gradient flows on evolving surfaces \cite{NitschkeSadikVoigt_A_2022}. However, even though the net force $ \totalFb = \sum_{\alpha}\totalFb_{\alpha} $ is independent of $ \gauge{\Wb}\Qb $, it's essential to note that its component forces $ \totalFb_{\alpha}\in\tangentR $, presented separately in all subsections of Sections \ref{sec:lagrangian_forces} and \ref{sec:flux_forces}, are not. Considering the component force decomposition
\begin{align}\label{eq:total_force_splitting_weak}
    \totalFb^{\Psi}_{\alpha} &= \Fb_{\alpha} + \gaugeFb^{\Psi}_{\alpha}
    &&\text{\st}
    &\innerH{\tangentR}{\gaugeFb^{\Psi}_{\alpha},\Wb} &= \innerH{\tangentQR}{\Hb_{\alpha},\gauge{\Wb}\Qb }
\end{align}
for all $ \Wb\in\tangentR $. 
The superscript $ \Psi $ indicates a particular choice of $ \gauge{\Wb}\Qb $, which represents a gauge of surface independence for Q-tensor fields \cite{NitschkeSadikVoigt_A_2022}. We consider the following two gauges of surface independence $ \gauge{}^{\Psi}\Qb=\nullb $ for $ \Psi\in\{\mfrak,\jau\} $:
\begin{align}
    \gauge{}^{\mfrak}\Qb &=\nullb &&:\Longleftrightarrow  \label{eq:material_gauge}
        & \forall\Wb\in\tangentR:\quad\quad &\gauge{\Wb}\Qb = \nullb \formComma \\
    \gauge{}^{\jau}\Qb &=\nullb &&:\Longleftrightarrow    \label{eq:jaumann_gauge}
        & \forall\Wb\in\tangentR:\quad\quad& \gauge{\Wb}\Qb - \Abcal[\Wb]\Qb + \Qb\Abcal[\Wb] = \nullb \formPeriod
\end{align}
Here $ \Abcal[\Wb]\in\tangentAR $ \eqref{eq:Abcal} is the skew-symmetric deformation gradient.
For $ \Psi = \mfrak $, we refer to $ \gauge{}^{\mfrak}\Qb =\nullb $ as the material gauge of surface independence. This implies that $ \Qb $ is independent of any suitable deformations $ \para + \eps\Wb $ from an Euclidean perspective. In other words, its director fields and scalar orders remain constant under deformations. This differs from the tangential material gauge of surface independence given in \cite{NitschkeSadikVoigt_A_2022}, where only the tangential part of $ \gauge{\Wb}\qb $ vanished for $ \qb\in\tangentQS $. This distinction is similar to the difference between $ \Dmat\Qb $  and $ \dot{\qb}=\projS[^2](\Dmat\qb) $, as described in \cite{NitschkeVoigt_2023}. 
For $ \Psi = \jau $, we refer to $ \gauge{}^{\jau}\Qb =\nullb $ as the Jaumann gauge of surface independence. This implies a more intrinsic independence with respect to deformations. While scalar orders remain constant, director fields are effectively "frozen" in the deformation, meaning that they appear constant to an observer on the surface under deformations. This represents a genuine extension of the tangential Jaumann gauge introduced in \cite{NitschkeSadikVoigt_A_2022}.
For the decomposition $ \Qb=\Qdepl(\qb,\etab,\beta) $ \eqref{eq:qtensor_decomposition} the Jaumann gauge \eqref{eq:jaumann_gauge} is equivalent to $ \gauge{}^{\jau}\qb=\nullb $,  $ \gauge{}^{\jau}\etab=\nullb $, and $ \gauge{}\beta = 0 $. 
Eventually, the gauge-dependent force $ \gaugeFb^{\Psi}_{\alpha} $ vanishes for $ \Psi=\mfrak $. 
For $ \Psi = \jau $ we use Corollary \ref{cor:jaumann_force_weak}, which yields
\begin{align}
    \totalFb^{\mfrak}_{\alpha} &= \Fb_{\alpha} \formComma \label{eq:total_material_force}\\
    \totalFb^{\jau}_{\alpha} &= \Fb_{\alpha}
                                + \DivC \left( \left( \Id + \normal\otimes\normal \right) (\Qb\Hb_{\alpha} - \Hb_{\alpha}\Qb) \IdS \right) \formPeriod \label{eq:total_Jaumann_force}
\end{align}
It's important to reiterate that \eqref{eq:total_material_force} and \eqref{eq:total_Jaumann_force} do not have any differential impact on the local equations \eqref{eq:LDA}. This is because the sum of all surface gauge-dependent forces cancels out due to $ \Hb=\sum_{\alpha}\Hb_{\alpha} = \nullb $ \eqref{eq:LDA_molecular}. 
In any case, the distinction between different gauge-dependent forces only applies in situations where $ \Hb_{\alpha} \neq \nullb $, such as when $ \alpha\in\{ \EL, \THT, \SC, \CB, \UN, \IM, \NV \} $, as specified in the associated subsections.

We find it consistent to either employ a fully tangential calculus or a $ \R^3 $-related calculus, and strive to avoid mixing both approaches. This principle extends to decompositions that respect tangential spaces and the choice of differential operators.
For surface-conforming models, particularly those formulated without Lagrange parameters, it is convenient to employ a tangential calculus. When the component fluid forces are described by stress tensor fields $ \Sigmab_{\alpha}\in\tangentR\otimes\tangentS $ such that $ \Fb_{\alpha} = \DivC\Sigmab_{\alpha} $, we use an orthogonal decomposition as follows:
\begin{gather}\label{eq:stress_decomposition}
    \Sigmab_{\alpha} = \sigmab_{\alpha} + \normal\otimes\varsigmab_{\alpha}\in  \tangentR\otimes\tangentS \formComma \\
    \begin{aligned}
        \sigmab_{\alpha} &=\IdS\Sigmab_{\alpha}\in\tangentS[^2]\formComma 
                &&&\varsigmab_{\alpha} &=\normal\Sigmab_{\alpha}\in\tangentS\formPeriod \notag
    \end{aligned}
\end{gather}
The forces themselves can be decomposed as:
\begin{gather}\label{eq:force_decomposition}
     \Fb_{\alpha} = \fb_{\alpha} + \fnor[\alpha]\normal = \DivC\Sigmab_{\alpha}\in\tangentR \formComma\\
     \begin{aligned}
        \fb_{\alpha} 
            &= \IdS\Fb_{\alpha}
            = \div\sigmab_{\alpha} - \shop\varsigmab_{\alpha} \in\tangentS\formComma
         &&&\fnor[\alpha] 
            &= \normal\Fb_{\alpha}
             = \div\varsigmab_{\alpha} + \shop\dbdot\sigmab_{\alpha} \in\tangentS[^0] \formPeriod\notag
     \end{aligned}
\end{gather}
These expressions are valid when the force is represented by an associated stress tensor field, as defined in \eqref{eq:DivCSigma_decomposed}. The momentum force decomposition for $ \rho\Dmat\Vb $ is already provided in \eqref{eq:DmatV}.
An orthogonal decomposition of Q-tensor fields is formulated in \eqref{eq:qtensor_decomposition}, and we express Q-tensor force fields as:
\begin{gather}\label{eq:Qtensorforce_decomposition}
    \Hb_{\alpha}
         =\Qdepl(\hb_{\alpha},\zetab_{\alpha},\omega_{\alpha})
         = \hb_{\alpha} + \zetab_{\alpha}\otimes\normal + \normal\otimes\zetab_{\alpha} + \omega_{\alpha}\left( \normal\otimes\normal -\frac{1}{2}\IdS \right) \in \tangentQR\formComma \\
     \hb_{\alpha}
        = \projQS\Hb_{\alpha}
        = \IdS\Hb_{\alpha}\IdS - \frac{\Hb_{\alpha}\dbdot\IdS}{2}\IdS \in\tangentQS \formComma \notag\\
     \begin{aligned}
        \zetab_{\alpha} &= \IdS\Hb_{\alpha}\normal = \normal\Hb_{\alpha}\IdS \in\tangentS\formComma
           &&&\omega_{\alpha} &= \Hb_{\alpha}(\normal,\normal) \in\tangentS[^0] \formPeriod \notag
     \end{aligned}
\end{gather}

\subsection{Lagrangian Forces}\label{sec:lagrangian_forces}

\subsubsection{Elastic Forces}\label{sec:elastic}

As part of the potential energy, we consider the surface elastic free energy, which is expressed in operator and mixed-proxy formulation as follows:
\begin{align}\label{eq:energy_elastic}
   \energyEL
         &:= \frac{L}{2}\normHsq{\tangentQR\otimes\tangentS}{\nablaC\Qb}
          = \frac{L}{2} \int_{\surf} g^{ij} \partial_i Q^{AB} \partial_j Q_{AB} \dS \formComma
\end{align}
This energy term can be viewed as a surface modification of a one-constant approximation ($L := L_1$ and $L_2 = L_3 = L_4 = L_6 = 0$) of the elastic free energy \cite{Schiele_1983,Berreman_1984,Mori_1999,MottramNewton_2014} in a bulk material and represents the thin film limit of this energy under homogeneous Neumann boundary conditions which follows by applying Lemma \ref{lem:thinfilmgradient} to the gradient and using the rectangle rule for the integral. The thin film limit has been previously explored in \cite{Nestler_2020} for conforming Q-tensor fields and in \cite{Nitschke_2018} for conforming Q-tensor fields with a constant eigenvalue field in the normal direction.

Variation \wrt\ $ \Qb $ in direction of $ \Psib\in\tangentQR $ yields
\begin{align}
    \innerH{\tangentQR}{\deltafrac{\energyEL}{\Qb},\Psib}
        &= L \innerH{\tangentQR\otimes\tangentS}{\nablaC\Qb, \nablaC\Psib}
         = -\innerH{\tangentQR}{\HbEL,\Psib} \notag\\
     \HbEL \label{eq:elastic_molucular_field}
        &= L \DeltaC\Qb = L\divC\nablaC\Qb \in \tangentQR \formPeriod
\end{align}
Note that the space of Q-tensor fields is closed by the surface Laplace operator $ \DeltaC $, as demonstrated in \cite{NitschkeVoigt_2023}.
Consequently, there is no need for additional projections in this context.
With deformation formula \eqref{eq:varXF} and the commutator of $ \gauge{\Wb} $ and $ \nablaC $ given in Lemma \ref{lem:ttensor_commutator_gauge_nablaS},
the energy $ \energyEL $ yields
\begin{align}
    \innerH{\tangentR}{ \deltafrac{\energyEL}{\para} , \Wb } \hspace{-50pt}\notag\\
        &= L \left(\innerH{\tangentQR\otimes\tangentS}{\nablaC\Qb, \gauge{\Wb}\nablaC\Qb}
                        + \frac{1}{2}\innerH{\tangentS[^0]}{ \normsq{\tangentQR\otimes\tangentS}{\nablaC\Qb}  ,\Tr\Gb[\Wb]}\right)\notag\\
        &= L\left( \innerH{\tangentQR\otimes\tangentS}{\nablaC\Qb, \nablaC\gauge{\Wb}\Qb - (\nablaC\Qb)\nablaC\Wb} 
                        + \frac{1}{2}\innerH{\tangentS[^0]}{ \normsq{\tangentQR\otimes\tangentS}{\nablaC\Qb}  ,\Tr\nablaC\Wb}\right)\notag\\
        &=  -\innerH{\tangentQR}{ \HbEL , \gauge{\Wb}\Qb }
               + \innerH{\tangentR\otimes\tangentS}{  \SigmabEL ,\nablaC\Wb}\formComma\notag\\
    \SigmabEL \label{eq:elastic_stress_field}
         &= -L \projQS\left( (\nablaC\Qb)^{T} \dbdot\nablaC\Qb \right) 
          = - L \left( (\nablaC\Qb)^{T} \dbdot\nablaC\Qb - \frac{\normsq{\tangentR[^3]}{\nablaC\Qb}}{2}\IdS \right)
                  \in\tangentQS\formPeriod
\end{align}
Therefore the elastic force holds $ \FbEL = \DivC\SigmabEL \in \tangentR  $.
To prevent misinterpretations of the notation, it is  $ [(\nablaC\Qb)^{T} \dbdot\nablaC\Qb]_{ij} = \partial_i Q^{AB} \partial_j Q_{AB} $ valid  in mixed-proxy notation.

\paragraph{Elastic Forces for Surface Conforming Q-Tensor Fields:}
We substitute the surface conforming decomposition $ \Qb=\CQdepl(\qb,\beta)\in\tangentCQR $ \eqref{eq:surface_conforming_ansatz} into $ \HbEL $ \eqref{eq:elastic_molucular_field}.
The surface conforming part $ \proj_{\tangentCQR}\HbEL $ is already derived in \cite{NitschkeVoigt_2023}. 
This reference also provides a fully orthogonal decomposition of the Laplace operator $ \DeltaC:\tangentR[^2] \rightarrow \tangentR[^2] $, which restricted to $ \tangentCQR $ yields the complementary part $ \proj^{\bot}_{\tangentCQR}\HbEL $. 
Here we use that $ \shop^2=\meanc\shop -\gaussc\IdS $ and $ \nabla(\qb-\frac{\beta}{2}\IdS) = \nabla\qb-\frac{1}{2}\IdS\otimes\nabla\beta $ is valid by metric compatibility, which results in
\begin{subequations}\label{eq:elastic_conforming_molecular_field}
\begin{align}
    \HbEL &= \Qdepl(\hbEL,\zetabEL,\omegaEL) \in \tangentQR\\ 
    \hbEL &= L \left(  \Delta\qb - (\meanc^2-2\gaussc)\qb + 3\beta \meanc \left( \shop - \frac{\meanc}{2}\IdS \right) \right) \in \tangentQS\\
    \zetabEL \label{eq:elastic_conforming_molecular_field_zeta}
        &= L\left( 2(\nabla\qb)\dbdot\shop + \qb\nabla\meanc - 3\shop\nabla\beta - \frac{3}{2}\beta\nabla\meanc \right) \in \tangentS \\
    \omegaEL &= L \left(  2 \meanc\shop \dbdot \qb + \Delta\beta  - 3\beta \left( \meanc^2 - 2\gaussc \right)  \right) \in\tangentS[^0]
\end{align}
\end{subequations}
for decomposition \eqref{eq:Qtensorforce_decomposition}.
Please note that the non-conforming component, which is exclusively represented by $ \zetabEL $, does not typically vanish, particularly for curved surfaces. As a result, it plays a role in the surface-conforming constraint stress discussed in Section \ref{sec:surface_conforming_constrain}.
Based on the proof of the fully orthogonal decomposition of the Laplace operator $ \DeltaC $ in \cite{NitschkeVoigt_2023}, we can extract the results for the mixed-proxy formulation of first-order derivatives for both tangential and purely normal tensor fields. 
Consequently, this leads to
\begin{align*}
    \partial_k Q^{AB} 
        &= \partial_k \left( q^{AB}-\frac{\beta}{2}\IdSC^{AB} \right) + \partial_k(\beta\normalC^A\normalC^B)\\
        &= \left(\tensor{q}{^{ij}_{|k}} -\frac{1}{2}\beta_{|k}g^{ij}\right)\partial_i\paraC^A\partial_j\paraC^B
           + \beta_{|k} \normalC^A\normalC^{B} \\
        &\quad + \left( q^{ij}\shopC_{ki} - \frac{3}{2}\beta \shopC^j_k \right) \normalC^A\partial_j\paraC^B
           + \left( q^{ij}\shopC_{kj} - \frac{3}{2}\beta \shopC^i_k \right) \partial_i\paraC^A \normalC^{B}
\end{align*}
for $ \Qb=\CQdepl(\qb,\beta) $.
After renaming some summation indices and exploiting orthogonality $ \normal\bot\partial_i\para $, we obtain the following expression:
\begin{align*}
\MoveEqLeft
    \left[  (\nablaC\Qb)^{T} \dbdot\nablaC\Qb \right]_{kl} = (\partial_k\Qb)\dbdot(\partial_l\Qb)\\
        &= g_{im}g_{jn} \left( \tensor{q}{^{ij}_{|k}} - \frac{1}{2}\beta_{|k}g^{ij} \right)\left( \tensor{q}{^{mn}_{|l}} - \frac{1}{2}\beta_{|l}g^{mn} \right)
          +  \beta_{|k} \beta_{|l}
          + 2 g_{jn} \left( q^{ij}\shopC_{ki} - \frac{3}{2}\beta \shopC^j_k \right) \left( q^{mn}\shopC_{lm} - \frac{3}{2}\beta \shopC^n_l \right) \\
        &= \tensor{q}{^{ij}_{|k}}\tensor{q}{_{ij|k}} 
            +  \frac{3}{2}\beta_{|k} \beta_{|l}
            + 2 q^{ij} q^m_j\shopC_{ki}\shopC_{lm}
            - 6\beta q^{ij} \shopC_{ki}\shopC_{lj}
            + \frac{9}{2} \beta^2 \shopC^{j}_{k}\shopC_{lj}\\
        &= \left[ (\nabla\qb)^{T} \dbdot\nabla\qb 
                    + \frac{3}{2} \nabla\beta \otimes \nabla\beta 
                    - 6 \beta\shop\qb\shop
                    + \left(\Tr\qb^2 + \frac{9}{2}\beta^2 \right)\shop^2 \right]_{kl} \formPeriod 
\end{align*}
Here, we use the fact that $ 2\qb^2=(\Tr\qb^2)\IdS $.
The Gauss equation for surfaces\footnote{We employ the representation of the Riemannian curvature tensor $ \boldsymbol{\mathcal{R}} $ as
                        $ \mathcal{R}_{ijkl} = \gaussc\left( g_{ik}g_{jl} - g_{il}g_{jk} \right) = \shopC_{ik}\shopC_{jl} - \shopC_{il}\shopC_{jk}  $ on surfaces.}
gives $ \shop\qb\shop = \gaussc\qb + (\qb \dbdot \shop)\shop $, and $ \projQS\shop^2 = \meanc\projQS\shop $.
As a result, the elastic stress field  \eqref{eq:elastic_stress_field} becomes:
\begin{align}\label{eq:sigmabEL}
    \SigmabEL = \sigmabEL
        &= -L \Bigg( \projQS((\nabla\qb)^{T} \dbdot\nabla\qb ) 
                    + \frac{3}{2}\projQS( \nabla\beta \otimes \nabla\beta ) 
                    - 6\gaussc\beta\qb \\
        &\quad\quad\quad\quad + \frac{1}{2}\left( 2\meanc\Tr\qb^2 - 12\beta \qb \dbdot \shop  + 9 \meanc\beta^2 \right)\left( \shop - \frac{\meanc}{2}\IdS \right)\Bigg)
          \in \tangentQS \formPeriod \notag
\end{align}

\subsubsection{Thermotropic Forces}\label{sec:thermotropic}

As part of the potential energy, we consider the surface thermotropic free energy described by
\begin{align}\label{eq:energy_thermotropic}
    \energyTH
        &= \innerH{\tangentSymR}{a\Qb + \frac{2b}{3}\Qb^2 + c\Qb^3, \Qb}
\end{align}
which includes thermotropic coefficients $ a,b,c\in\R $.
This energy is analogous to that of a bulk system, as found in references such as \cite{Schiele_1983,Berreman_1984,Mori_1999,MottramNewton_2014}, but restricted to the surface.
Trivially, \eqref{eq:energy_thermotropic} is the thin film limit of the thermotropic energy within a bulk, regardless of the choice of natural thin-film boundary conditions.
A thin-film limit analysis has previously been performed in \cite{Nestler_2020} for conforming Q-tensor fields and in \cite{Nitschke_2018} for conforming Q-tensor fields with a constant eigenvalue field in normal direction.
Note that we can simplify our calculations using the identity $ \Tr\Qb^4 = \frac{1}{2}\left(\Tr\Qb^2\right)^2 $.

We have already calculated the variation of  \eqref{eq:energy_thermotropic} \wrt\ $ \Qb $ in direction of $ \Psib\in\tangentQR $ in \cite{NitschkeVoigt_2023}, which yields the following expression:
\begin{align}
     \innerH{\tangentQR}{\deltafrac{\energyTH}{\Qb},\Psib}
             &= -\innerH{\tangentQR}{\HbTH,\Psib} \notag\\
     \HbTH \label{eq:thermotropic_molucular_field}
            &= -2 \left( a \Qb + b\projQR\Qb^2 + c\Tr(\Qb^2)\Qb \right)
                \in \tangentQR \formPeriod
\end{align}
Note that $ \Qb^2\in\tangentSymR $ is not a Q-tensor field, which justifies the projection $ \projQR\Qb^2 = \Qb^2 - \frac{\Tr\Qb^2}{3}\Id $. 
The spatial variation is very similar to the variation  \wrt\ $ \Qb $ since the deformation derivative is compatible with the $ \R^3 $-metric, \ie\ $ \gauge{\Wb}\Id = 0 $.
Therefore, $  \gauge{\Wb} $ commutes with the local inner product and the trace operator.
The deformation formula \eqref{eq:varXF} yields
\begin{align}
      \innerH{\tangentR}{ \deltafrac{\energyTH}{\para} , \Wb }
            &= 2\innerH{\tangentSymR}{a \Qb + b\Qb^2 + c\Tr(\Qb^2)\Qb, \gauge{\Wb}\Qb} \notag\\
              &\quad  + \innerH{\tangentS[^0]}{ a \Tr\Qb^2 + \frac{2b}{3}\Tr\Qb^3 + c\Tr\Qb^4 , \Tr\Gb[\Wb]}\notag\\
            &= -\innerH{\tangentQR}{\HbTH,\gauge{\Wb}\Qb} +  \innerH{\tangentR\otimes\tangentS}{  \SigmabTH ,\nablaC\Wb}\formComma \notag\\
     \SigmabTH \label{eq:thermotropic_stress_field}
        &= \left( a \Tr\Qb^2 + \frac{2b}{3}\Tr\Qb^3 + c\Tr\Qb^4 \right) \IdS 
        =: \pTH \IdS \in \tangentIdS\formPeriod
\end{align}
Therefore, the thermotropic force yields $ \FbTH = \DivC\SigmabTH = \GradC\pTH \in \tangentR  $.

\paragraph{Thermotropic Forces for Surface Conforming Q-Tensor Fields:}
 Substituting the surface conforming decomposition $ \Qb=\CQdepl(\qb,\beta)\in\tangentCQR $ \eqref{eq:surface_conforming_ansatz}, as shown in \cite{NitschkeVoigt_2023}, results in the surface-conforming tensor field
\begin{subequations}\label{eq:thermotropic_conforming_molecular_field}
 \begin{align}
    \HbTH &= \CQdepl(\hbTH,\omegaTH)\in\tangentCQR\\
    \hbTH &= -\left(2a - 2b\beta + 3c\beta^2  + 2c\Tr\qb^2\right)\qb  \in \tangentQS\\
    \omegaTH &= -\left(2a + b\beta + 3c\beta^2  + 2c\Tr\qb^2\right)\beta
                   +\frac{2}{3} b \Tr\qb^2 \in\tangentS[^0] \formPeriod
 \end{align}
\end{subequations}
 We also adopt the following identities from \cite{NitschkeVoigt_2023}:
 $ \Tr\Qb^2 = \Tr\qb^2 + \frac{3}{2}\beta^2 $, 
 $ \Tr\Qb^3 = \Qb^2\dbdot\Qb = -\frac{3}{2}\beta(\Tr\qb^2 - \frac{1}{2}\beta^2) $
 and $ \Tr\Qb^4 = \frac{1}{2}(\Tr\Qb^2)^2 = \frac{1}{2}((\Tr\qb^2)^2 + 3\beta^2\Tr\qb^2 + \frac{9}{4}\beta^4) $. 
 Substituting these into \eqref{eq:thermotropic_stress_field} yields the thermotropic pressure
 \begin{align}\label{eq:pTH}
    \pTH &= \left( 2a - 2b\beta + c\left( \Tr\qb^2 + 3\beta^2 \right) \right) \frac{\Tr\qb^2}{2}
            +\left( 12 a + 4b\beta + 9c\beta^2 \right)\frac{\beta^2}{8} \in\tangentS[^0]\formPeriod 
 \end{align}
 This pressure gives rise to the stress tensor field $ \SigmabTH = \sigmabTH =\pTH \IdS \in \tangentIdS $.
 In principle, this stress tensor represents the energy density of $ \energyTH $ in terms of pressure, constrained by surface-conforming Q-tensor fields.
 These results are consistent with those for bulk energy in  \cite{Nitschke_2018} and \cite{Nestler_2020}, even when $ \beta $ is constant in the first reference.

\subsubsection{Bending Forces}\label{sec:bending}

As part of the potential energy we consider the surface bending energy
\begin{align}\label{eq:energy_bending}
    \energyBE
        &= \frac{\kappa}{2} \normHsq{\tangentS[^0]}{\meanc - \meanc_{0}} \formComma
\end{align}
where $ \kappa>0 $ is the bending stiffness and $ \meanc_{0}\in\R $ the spontaneous curvature.

This energy does not depend on the state variable $ \Qb $.
Therefore variation \wrt\ $ \Qb $ in direction of $ \Psib\in\tangentQR $ gives
\begin{align*}
    \innerH{\tangentQR}{\deltafrac{\energyBE}{\Qb},\Psib}
        &= -\innerH{\tangentQR}{\HbBE,\Psib} \\
     \HbBE 
        &= \nullb \in \tangentQR \formPeriod
\end{align*}
The spatial variation is already calculated in \cite{Zhong_can_1989,Barrett_SIAMJSC_2008,Dougan_ESAIM_2012,Hauser2013,Tu_2014,Guckenberger_2017,Reuther_2020,BachiniKrauseNitschkeVoigt_2023} and yields
\begin{align*}
    \innerH{\tangentR}{ \deltafrac{\energyBE}{\para} , \Wb }
        &=  \innerH{\tangentR\otimes\tangentS}{  \SigmabBE ,\nablaC\Wb}\formComma \\
    \SigmabBE
        &= -\kappa\left( (\meanc - \meanc_{0})\projQS\shop + \normal\otimes\nabla\meanc \right) -\frac{\kappa}{2}\meanc_{0}(\meanc - \meanc_{0})\IdS 
                \in \tangentR\otimes\tangentS  \formPeriod
\end{align*}
Note that the bending force $ \FbBE=\DivC\SigmabBE $ is a pure normals force, 
\ie\ it is $ \FbBE = \fnorBE \normal \in\tangentnormal $ valid, namely with
\begin{align}\label{eq:fnorBE}
    \fnorBE
        &= -\kappa\left( \Delta\meanc + \left( \meanc - \meanc_{0} \right)\left( \normsq{\tangentQS}{\projQS\shop} + \frac{\meanc_{0}}{2}\meanc \right) \right) \in\tangentS[^0]\formComma
\end{align}  
where $ \normsq{\tangentQS}{\projQS\shop} = \shop\dbdot\shop -\frac{1}{2}\meanc^2 = \frac{1}{2}\meanc^2 - 2\gaussc $.

\subsubsection{Surface Conforming Constraint (optional)} \label{sec:surface_conforming_constrain}

In addition to the common quadratic ansatz $ \Qb=\sum_{\alpha=1}^{2} S_\alpha(\pb_\alpha\otimes\pb_\alpha-\frac{1}{3}\Id) $, where $ S_\alpha $ are scalar order parameter fields and $ \pb_\alpha $ normalized directional fields, there is  also the linear ansatz \eqref{eq:qtensor_decomposition}, which is uniquely defined.
This approach may be less vivid than the quadratic one, but it provides a very useful orthogonal decomposition in the context of tangential related subspaces of $ \tangentQR $.
Especially for surface conforming Q-tensor fields \eqref{eq:surface_conforming_space} it holds that $ \Qdepl(\qb,\etab,\beta)\in\tangentCQR $ if and only if $ \etab=0 $.
Therefore, the orthogonal decomposition for surface conforming Q-tensor fields $ \Qb\in\tangentCQR $ yields the ansatz
\begin{align}\label{eq:surface_conforming_ansatz}
    \Qb &= \CQdepl(\qb,\beta) := \Qdepl(\qb,0,\beta) = \qb + \beta\left( \normal\otimes\normal - \frac{1}{2}\IdS \right)\in\tangentCQR
\end{align}  
with mutually independent and uniquely given $ \qb\in\tangentQS $ and $ \beta\in\tangentS[^0] $.
To force a solution to be surface conformal we take advantage of the Lagrange multiplier technique and consider the Lagrange function
\begin{align}\label{eq:surface_conforming_lagrange}
    \energySC 
        &:= \innerH{\tangentS}{\lambdabSC,\IdS\Qb\normal}
        = \innerH{\tangentS}{\lambdabSC,\etab} \formComma
\end{align}
where $ \lambdabSC\in\tangentS $ is the associated Lagrange parameter, yielding an additional degree of freedom, and $ \etab\in\tangentS $ determines the non-conforming part in the decomposition $ \Qb=\Qdepl(\qb,\etab,\beta) $ \eqref{eq:qtensor_decomposition}.
More precisely, $ \lambdabSC $ is an additional state variable and comes with the new equation
$ \innerH{\tangentS}{\deltafrac{\energySC}{\lambdabSC},\thetab} = 0 $ for all virtual displacements $ \thetab\in\tangentS $, since $ \energySC $ is the only energetic term considered to contain $ \lambdabSC $. 
The strong form yields the desired constraint equation
\begin{align}\label{eq:surface_conforming_constrain}
    0 &= \IdS\Qb\normal = \etab  \in\tangentS
\end{align}
for $ \Qb=\Qdepl(\qb,\etab,\beta)\in\tangentQR $.
Note that $ \energySC $ states a potential energy technically.
However, under the constraint \eqref{eq:surface_conforming_constrain} it has not any impact on the net system, \ie\ it holds $ \energySC\vert_{\Qb\in\tangentCQR} = 0 $.
Nevertheless, its variation \wrt\ state variables $ \Qb $ and $ \para $ does not vanish and has to be considered.

The variation \wrt\ $ \Qb\in\tangentQR $ in direction of $ \Psib\in\tangentQR $ yields
\begin{align}
    \innerH{\tangentQR}{\deltafrac{\energySC}{\Qb}, \Psib}
        &= \innerH{\tangentR[^2]}{\lambdabSC\otimes\normal, \Psib}
         = -\innerH{\tangentQR}{\HbSC, \Psib} \formComma\notag\\
    \HbSC &= \Qdepl(\nullb, \zetabSC, 0) \label{eq:surface_conforming_molecular}
           = -\frac{1}{2}\left( \lambdabSC\otimes\normal + \normal\otimes\lambdabSC \right) 
               \in \tangentCQRbot < \tangentQR\formComma\\
    \zetabSC \label{eq:surface_conforming_molecular_zetab}
        &= -\frac{1}{2}\lambdabSC \in \tangentS \formComma
\end{align}
where $ \tangentCQRbot:= \{ \Qb\in\tangentQR \mid \Qb\bot\tangentCQR \} =  \proj_{\tangentSymR}(\normal\otimes\tangentS) $ is the space of Q-tensor fields, which are orthogonal to the space of surface conforming Q-tensor fields.
We can take advantage of this orthogonality if we want to derive a reduced system without the additional state variable $ \lambdabSC $. 
The generalized force balance for the non-conforming part obtained by the Lagrange-D'Alembert principle reads $ \nullb = \sum_{\alpha}\zetab_{\alpha} $, where $ \{\zetab_{\alpha}\} $ determine all non-conforming forces considered in this paper \wrt\ the decomposition $ \Hb_{\alpha} = \Qdepl(\hb_{\alpha},\zetab_{\alpha},\omega_{\alpha}) $.
Most of these non-conforming forces are vanishing under the constraint \eqref{eq:surface_conforming_constrain} and we end up with the identity
\begin{align}\label{eq:lambdaSC_determined}
    \lambdabSC &= 2\left( \zetabEL + \zetabIM \right) \formComma
\end{align}
where $ \zetabEL\in\tangentS $ determines the non-conforming elastic force \eqref{eq:elastic_conforming_molecular_field} and $ \zetabIM[\mfrak]\in\tangentS $ determines the non-conforming immobility force \eqref{eq:immobility_conforming_molucular_field} for the material model ($ \Phi=\mfrak $).
For the Jaumann model ($ \Phi=\jau $) it is $ \zetabIM[\jau]=\nullb $ and in turn $ \lambdabSC $ depends only on the nematic elasticity.
As a consequence, if we are providing ansatz \eqref{eq:surface_conforming_ansatz}, which fulfills the constraint \eqref{eq:surface_conforming_constrain} a priori,
then it only remains to consider the tangential Q-tensor force balance $ \nullb = \sum_{\alpha}\hb_{\alpha} $ \eqref{eq:model_conforming_qtensor} and normal force balance $ 0 =  \sum_{\alpha}\omega_{\alpha} $ \eqref{eq:model_conforming_betaeq}
for $ \hb_{\alpha}\in\tangentQS $ and $ \omega_\alpha\in\tangentS[^0] $ given in their associated surface conforming subsection within this paper.
Note that if we consider uniaxiality in addition, the non-conforming uniaxiality constraint force $ \zetabUN $ has to be added to \eqref{eq:lambdaSC_determined}.
However, $ \zetabUN $ has no effect on the derivation below. 
We discuss this in more details in the surface conforming paragraph in Section \ref{sec:uniaxiality_constrain}.  
The variation \wrt\ $ \para $ in direction of $ \Wb\in\tangentR $ yields
\begin{align*}
    \innerH{\tangentR}{\deltafrac{\energySC}{\para},\Wb}
        &= \innerH{\tangentR}{ \gauge{\Wb}\lambdabSC + (\Tr\Gb[\Wb])\lambdabSC, \IdS\Qb\normal } \\
        &\quad + \innerH{\tangentR}{\lambdabSC, (\gauge{\Wb}\IdS)\Qb\normal + \IdS(\gauge{\Wb}\Qb)\normal + \IdS\Qb\gauge{\Wb}\normal} \formComma
\end{align*}
where the first summand vanishes by constraint \eqref{eq:surface_conforming_constrain}.
From \cite{NitschkeSadikVoigt_A_2022} we take the identities
\begin{align}
    \gauge{\Wb}\normal \label{eq:gauge_normal}
        &= -\normal\nablaC\Wb \in\tangentS \formComma\\ 
    \gauge{\Wb}\IdS
        &= \gauge{\Wb}\left( \Id - \normal\otimes\normal \right)
         = \normal\nablaC\Wb\otimes\normal + \normal\otimes\normal\nablaC\Wb \in\tangentCQRbot \formPeriod \notag 
\end{align}
Using that $ (\nablaC\Wb)\Qb\normal=0 $ holds by constraint \eqref{eq:surface_conforming_constrain}, results in
\begin{align}
    \innerH{\tangentR}{\deltafrac{\energySC}{\para},\Wb}
        &= \innerH{\tangentS}{\lambdabSC, \Qb(\normal,\normal)\normal\nablaC\Wb + \IdS(\gauge{\Wb}\Qb)\normal - \IdS\Qb(\normal\nablaC\Wb)} \notag\\
        &= -\innerH{\tangentQR}{\HbSC, \gauge{\Wb}\Qb} + \innerH{\tangentR\otimes\tangentS}{\SigmabSC, \nablaC\Wb} \notag\\
    \SigmabSC \label{eq:surface_conforming_stress}
        &= \normal \otimes \left( \Qb(\normal,\normal)\lambdabSC - \IdS\Qb\lambdabSC \right) \in \normal\otimes\tangentS \formPeriod
\end{align}
If we provide the ansatz \eqref{eq:surface_conforming_ansatz} and substitute $ \lambdabSC $ \eqref{eq:lambdaSC_determined} while taking $ \zetabEL $ \eqref{eq:elastic_conforming_molecular_field_zeta} and $ \zetabIM[\mfrak] $ \eqref{eq:immobility_conforming_molucular_field_zeta} into account, we obtain
the surface conformal constraint stress field
\begin{align}
     \SigmabSCPhi \label{eq:surface_conforming_stress_without_lambda}
        &:= \SigmabSC\vert_{\Phi}
          = \normal \otimes \varsigmabSCPhi \in \normal\otimes\tangentS \formComma\\
     \varsigmabSCPhi
         &= \beta\lambdabSC-(\qb-\frac{\beta}{2}\IdS)\lambdabSC
         = - 2 \left( \qb - \frac{3}{2}\beta\IdS \right)\left( \zetabEL + \zetabIM \right) \notag\\
        &= - L \left( 4 \left( \qb - \frac{3}{2}\beta\IdS \right) \nabla\qb \dbdot \shop 
                                        -6 \left( \qb - \frac{3}{2}\beta\IdS \right)\shop\nabla\beta
                                        -6 \beta\qb\nabla\meanc +  \left( \Tr\qb^2 + \frac{9}{2}\beta^2 \right)\nabla\meanc\right) \notag\\
        &\quad - \begin{cases}
                      M  \normal\nablaC\Vb \left( 6\beta\qb - \left( \Tr\qb^2 + \frac{9}{2}\beta^2 \right)\IdS \right) & \text{for }  \Phi=\mfrak \formComma \\
                      \nullb & \text{for }  \Phi=\jau \formComma         
                  \end{cases} \notag 
\end{align}
where $ L,M \ge 0 $ are the elastic and immobility parameter, see sections \ref{sec:elastic} and \ref{sec:immobility}.

\paragraph{Constant Normal Eigenvalue (optional)}\label{sec:constant_beta}
 Within several surface conforming Q-tensor models the eigenvalue $ \beta $ in normal direction is prescribed by a constant $ \beta_0\in\R $. There are two special cases discussed in \cite{Nestler_2020}.
 The so-called flat degenerated case ($ \beta_0=0 $) describing biaxial Q-tensor fields, which are intrinsically uniaxial.
 This means the quadratic director field ansatz for tangential principal director $ \pb\in\tangentS $ reads
 \begin{align*}
    \beta \equiv \beta_0 = 0:\quad\quad
    \Qb = \qb = S(\pb\otimes\pb-\frac{1}{2}\IdS) = S \pb\otimes\pb + \frac{S}{2}\normal\otimes\normal -\frac{1}{3}(S+\frac{S}{2})\Id \in \tangentQS\formComma
 \end{align*}
 where $ S\in\tangentS[^0] $ is the scalar order in $ \pb $-direction and $ \frac{S}{2} $ in normal direction.
 Therefore flat degenerated Q-tensor fields are only genuinely uniaxial for the trivial case $ \Qb=\nullb $.
 The second case relates to the surface conforming uniaxial setting $ \Qb = S(\pb\otimes\pb-\frac{1}{3}\Id) = \CQdepl(\qb,\beta) $, with the same $ \qb $  as given above and 
 $ \beta= -\frac{S}{3} = \pm \frac{\sqrt{2\Tr\qb^2}}{3} $, \cf\ relation \eqref{eq:uniaxial_constrain_conforming_simplified}. 
 By approximating the scalar order $ S\approx S_0 $ in $ \beta $ at the thermotropic equilibrium, where $ \HbTH = 0 $ \eqref{eq:thermotropic_molucular_field}, we obtain
 $ \beta \approx \frac{1}{12 c} (b - \sqrt{b^2 - 24 ac})=\beta_0 $ if we only consider the globally stable nematic state  for $ a < \frac{b^2}{27c} $, see \eg\ \cite{MottramNewton_2014}.
 We stipulate the Lagrange function for $ \beta_0\in\R $ by
 \begin{align*}
    \energyCB &:= \innerH{\tangentS[^0]}{\lambdaCB, \Qb(\normal,\normal)- \beta_0}
              =  \innerH{\tangentS[^0]}{\lambdaCB, \beta - \beta_0} \formPeriod
 \end{align*}
 This yields the Lagrange parameter $ \lambdaCB\in\tangentS[^0] $ as an additional degree of freedom.
 Variation \wrt\ this Lagrange parameter results in the identity $ \innerH{\tangentS[^0]}{\deltafrac{\energyCB}{\lambdaCB},\theta} = 0 $ for all $ \theta\in\tangentS[^0] $ and thus the constraint
 \begin{align}\label{eq:constant_beta_constrain}
    0 &= \Qb(\normal,\normal) - \beta_0 = \beta - \beta_0
 \end{align}
 for $ \Qb=\CQdepl(\qb,\beta)\in\tangentCQR $.
 Note that under this constraint all spatial and temporal derivatives $ \nabla\beta $ and $ \dot{\beta} $ can be neglected, since $ \beta $ is  a constant scalar field.
 
 The variation \wrt\ $ \Qb\in\tangentQR $ in direction of $ \Psib\in\tangentQR $ yields
 \begin{align*}
     \innerH{\tangentQR}{\deltafrac{\energyCB}{\Qb}, \Psib}
        &= \innerH{\tangentR[^2]}{\lambdaCB \normal\otimes\normal, \Psib}
         = -\innerH{\tangentQR}{\HbCB, \Psib} \formComma\\
     \HbCB &= \Qdepl(\nullb, \nullb, \omegaCB) = - \lambdaCB \left( \normal\otimes\normal - \frac{1}{3}\Id \right) \in \tangentQR \formComma\\
     \omegaCB &= -\frac{2}{3}\lambdaCB \in\tangentS[^0] \formPeriod
 \end{align*}
 Similar to the conforming constraint, we could use the orthogonality $ \HbCB\bot\Qdepl(\tangentQS,\tangentS, 0) $ if we want to derive a reduced model without Lagrange parameter. 
 The pure normal force balance $ 0 = \sum_{\alpha}\omega_{\alpha} $ yields
 \begin{align}\label{eq:lambdaCB_determined}
    \lambdaCB &= \frac{2}{3}\left(\omegaEL + \omegaTH + \omegaNV\right) \in \tangentS[^0]\formComma
 \end{align}
 where $\omegaEL$ determines the pure normal part of the elastic force \eqref{eq:elastic_conforming_molecular_field}, $  \omegaTH $ the pure normal part of the thermotropic force \eqref{eq:thermotropic_conforming_molecular_field} and  $\omegaNV$ the pure normal part of the nematic viscous force \eqref{eq:nematic_viscous_conforming_molecular_field}.
 The variation \wrt\ $ \para $ in direction of $ \Wb\in\tangentR $ yields
 \begin{align*}
  \innerH{\tangentR}{\deltafrac{\energyCB}{\para},\Wb}
      &= \innerH{\tangentS[^0]}{ \gauge{\Wb}\lambdaCB + \lambdaCB\Tr\Gb[\Wb], \Qb(\normal,\normal)- \beta_0 } \\
      &\quad + \innerH{\tangentS[^0]}{\lambdaCB,  2\Qb(\gauge{\Wb}\normal,\normal) + (\gauge{\Wb}\Qb)(\normal,\normal) } \formComma\\
      &= -\innerH{\tangentQR}{\HbCB, \gauge{\Wb}\Qb}
 \end{align*}
 by constraint \eqref{eq:constant_beta_constrain} and  $\Qb(\gauge{\Wb}\normal,\normal) = 0$, 
 since $ \gauge{\Wb}\normal\in\tangentS $ \eqref{eq:gauge_normal} and surface conforming constraint \eqref{eq:surface_conforming_constrain}.
 As a consequence, there are no additional constraint forces in the fluid equation \eqref{eq:LDA_fluid}, which is thus independent of the Lagrange parameter $ \lambdaCB $.   
 Therefore, we could also omit the pure normal force balance \eqref{eq:lambdaCB_determined} without further ado and it remains only the tangential Q-tensor force balance $ \nullb = \sum_{\alpha}\hb_{\alpha} $ \eqref{eq:model_conforming_qtensor} for all considered $ \hb_{\alpha}\in\tangentQS $.

\subsubsection{Uniaxiality Constraint (optional)}\label{sec:uniaxiality_constrain}

A Q-tensor field $ \Qb\in\tangentQR $ is called uniaxial if and only if two of its eigenvalue fields are equal, see \cite{MottramNewton_2014}.
It does not matter that $ \Qb $ is only defined on $ \surf $, since it is a pure local property.
To force uniaxiality, the original biaxiality measure $  (\Tr\Qb^2)^{-3}\mathfrak{b}(\Qb)  $ (\cite{Kaiser_1992,Majumdar_2010,Nestler_2020}), containing the eigenvalue polynomial
\begin{align}\label{eq:biaxiality_measure_dimensionalized}
    \mathfrak{b}(\Qb) 
        &:= (\Tr\Qb^2)^{3} - 6(\Tr\Qb^3)^2
         = 2(\lambda_1-\lambda_2)^2  (\lambda_1-\lambda_3)^2 (\lambda_2-\lambda_3)^2 
\end{align}
for eigenvalue fields $ \lambda_\alpha \in \tangentS[^0] $ of $ \Qb $, is not suitable for the Lagrange multiplier technique.
Due to its double root structure \wrt\ its eigenvalues, the variation \wrt\ $ \Qb $ would always result in a vanishing constraint force at $ \mathfrak{b}(\Qb)=0 $, which is not allowed. 
Instead of processing the biaxiality measure in any way, we introduce the biaxiality Q-tensor polynomial
\begin{align}\label{eq:biaxiality_qt}
    \mathfrak{B}(\Qb) 
        &:=  \Qb^4 - \frac{5}{6}(\Tr\Qb^2)\Qb^2 + \frac{1}{9}(\Tr\Qb^2)^2\Id \in\tangentQR \formPeriod
\end{align}
The vanishing trace $ \Tr\mathfrak{B}(\Qb) = 0 $ is given by $ 2\Tr\Qb^4=(\Tr\Qb^2)^2 $ and $ \Tr\Id=3 $.
Considering the square of the norm and using Lemma \ref{lem:trace_of_some_Q_powers} results in
\begin{align*}
    \normsq{\tangentQR}{\mathfrak{B}(\Qb)}
        &= \Tr\mathfrak{B}(\Qb)^2
         = \Tr\Qb^8 - \frac{5}{3}(\Tr\Qb^2)\Tr\Qb^6 + \frac{11}{12}(\Tr\Qb^2)^2\Tr\Qb^4 - \frac{4}{27}(\Tr\Qb^2)^4 \\
        &= \frac{1}{54}(\Tr\Qb^2)^4 - \frac{1}{9}(\Tr\Qb^2)(\Tr\Qb^3)^2
         = \frac{\Tr\Qb^2}{54} \mathfrak{b}(\Qb) \formPeriod
\end{align*}
As a consequence it holds:
\begin{align*}
    \mathfrak{B}(\Qb) &= \nullb  && \Longleftrightarrow && \Qb\in\tangentQR \text{ is uniaxial.} 
\end{align*}
Basically, we found a tensorial square root of \eqref{eq:biaxiality_measure_dimensionalized}, which overcomes the quadratic eigenvalue polynomial issue mentioned above.
This motivates us to define the Lagrange function
\begin{align*}
    \energyUN
        &:= -3 \innerH{\tangentQR}{\LambdabUN, \mathfrak{B}(\Qb)}
          = - 3\innerH{\tangentQR}{\LambdabUN, \Qb^4 - \frac{5}{6}(\Tr\Qb^2)\Qb^2 + \frac{1}{9}(\Tr\Qb^2)^2\Id}
\end{align*}
as a potential energy, where the Lagrange parameter $ \LambdabUN\in\tangentQR $ appears as an additional degree of freedom.
Therefore, the variation $ \innerH{\tangentQR}{\deltafrac{\energyUN}{\LambdabUN}, \Thetab} $ in arbitrary directions of $ \Thetab\in\tangentQR $ gives the Q-tensor valued constraint
\begin{align}\label{eq:uniaxiality_constrain}
    \mathfrak{B}(\Qb) &= \nullb \in\tangentQR\formComma
\end{align}
since $ \energyUN $ is the only considered energetic term comprising  $ \LambdabUN $.

The variation \wrt\ $ \Qb $ in direction of $ \Psib\in\tangentQR $ yields
\begin{align}
    \MoveEqLeft
    \innerH{\tangentQR}{\deltafrac{\energyUN}{\Qb},\Psib} \notag\\
        &= -3 \innerH{\tangentR[^2]}{\LambdabUN, \Qb^3\Psib + \Qb^2\Psib\Qb + \Qb\Psib\Qb^2 + \Psib\Qb^3} \notag\\
        &\quad + \frac{5}{2} \innerH{\tangentR[^2]}{\LambdabUN, (\Tr\Qb^2)(\Qb\Psib+\Psib\Qb) + 2(\Qb\dbdot\Psib)\Qb^2 }
             -\frac{4}{3} \innerH{\tangentR[^2]}{ \LambdabUN, (\Tr\Qb^2) (\Qb\dbdot\Psib)\Id }\notag\\
        &= -\innerH{\tangentR[^2]}{6\left(\LambdabUN\Qb^3 + \Qb\LambdabUN\Qb^2\right)- 5\left( (\Tr\Qb^2)\Psib\Qb + (\LambdabUN\dbdot\Qb^2)\Qb \right), \Psib}
        = -\innerH{\tangentQR}{\HbUN,\Psib} \formComma\notag\\
    \HbUN \label{eq:uniaxial_constrain_molecular_field}
        &= 6\projQR\left( \LambdabUN\Qb^3 + \Qb\LambdabUN\Qb^2 \right)
          -5(\Tr\Qb^2) \projQR\left( \LambdabUN\Qb \right)
          -5 (\LambdabUN\dbdot\Qb^2) \Qb\formComma
\end{align}
where we use that $ \Id\bot\tangentQR $ and $ \tangentQR < \tangentSymR $ holds to omit the rear summand and reduce symmetric summations in the first identity.
As a consequence $ \HbUN\in\tangentQR $ is the generalized constraint force to maintain uniaxiality \eqref{eq:uniaxiality_constrain} in the Q-tensor equation.
However, there are no effective constraint forces in the fluid equation, since the variation \wrt\ $ \para $ in direction of $ \Wb\in\tangentR $ results in
\begin{align}
    \innerH{\tangentR}{\deltafrac{\energyUN}{\para}, \Wb}
        &= -3 \left(\innerH{\tangentR[^2]}{\LambdabUN, \gauge{\Wb}\mathfrak{B}(\Qb)}
                    +\innerH{\tangentR[^2]}{\gauge{\Wb}\LambdabUN + \LambdabUN\Tr\Gb[\Wb], \mathfrak{B}(\Qb) } \right) \notag\\
        &= - \innerH{\tangentQR}{\HbUN,\gauge{\Wb}\Qb}
            \label{eq:uniaxiality_constrain_variation_wrt_X}
\end{align}
at the constraint $ \mathfrak{B}(\Qb)= \nullb $.

\paragraph{Uniaxiality Constraint for Surface Conforming Q-Tensor Fields:}

Substituting the surface conforming ansatz $ \Qb = \CQdepl(\qb,\beta) = \qb + \beta\left( \normal\otimes\normal - \frac{1}{2}\IdS \right)\in\tangentCQR $ into the summands of \eqref{eq:biaxiality_qt} yields
\begin{align*}
    \Qb^4 
        &= -\frac{\beta}{2}\left( 2\Tr\qb^2 + \beta^2 \right)\qb
            + \frac{1}{16}\left( 4(\Tr\qb^2)^2 + 12\beta^2\Tr\qb^2 + \beta^4 \right)\IdS
            + \beta^4 \normal\otimes\normal \formComma \\
    (\Tr\Qb^2)\Qb^2
        &= -\frac{\beta}{2}\left( 2\Tr\qb^2 + 3\beta^2 \right)\qb
            +\frac{1}{8}\left( 4(\Tr\qb^2)^2 + 8\beta^2\Tr\qb^2 + 3\beta^4 \right)\IdS
            + \frac{\beta^2}{2}\left( 2\Tr\qb^2 + 3\beta^2 \right) \normal\otimes\normal\formComma \\
    (\Tr\Qb^2)^2\Id
        &= \frac{1}{4}\left( 4(\Tr\qb^2)^2 + 12 \beta^2 \Tr\qb^2 + 9\beta^4 \right)\left( \IdS + \normal\otimes\normal  \right) \formComma
\end{align*} 
where we use that $ 2\qb^2= (\Tr\qb^2)\IdS $ is valid.
Eventually, we obtain the biaxiality Q-tensor polynomial
\begin{align}
     \mathfrak{B}(\Qb)
        &= \frac{1}{36}(2\Tr\qb^2 - 9\beta^2)\left( -3\beta\qb + 2 (\Tr\qb^2)\left( \normal\otimes\normal - \frac{1}{2}\IdS \right) \right) \notag\\
        &= \frac{1}{36}(2\Tr\qb^2 - 9\beta^2) \CQdepl(-3\beta\qb ,2 \Tr\qb^2) \in \tangentCQR \formPeriod
                \label{eq:biaxiality_qt_conforming}
\end{align}
As a consequence of $ \mathfrak{B}(\Qb)= \nullb $, it holds:
\begin{align}\label{eq:uniaxial_constrain_conforming_simplified}
    2\Tr\qb^2 = 9\beta^2 &\quad\text{ or }\quad \qb=0  && \Longleftrightarrow && \Qb= \CQdepl(\qb,\beta)\in\tangentCQR \text{ is uniaxial.} 
\end{align}
The first condition ($ 2\Tr\qb^2 = 9\beta^2 $) refers to uniaxiality with a tangential principal director, see \cite{Nestler_2020}, whereas the second condition ($ \qb=0 $) is associated  to uniaxiality with the normal field as principal director.
More precisely, it holds $ \Qb=\CQdepl(\nullb,\beta)= \frac{3}{2}\beta\left( \normal\otimes\normal - \frac{1}{3}\Id \right) $ in the latter situation.
Note that while it is possible for uniaxial surface conforming Q-tensor fields to bear tangential as well as normal principal directors simultaneously at the surface, there must be a spatial phase transition between both configurations for reasons of spatial continuity, where $\Qb= \nullb$ is orderless.
Analogous applies for temporal changes between both  director configurations.

To substitute $ \Qb = \CQdepl(\qb,\beta)$ into the constraint force $ \HbUN $ \eqref{eq:uniaxial_constrain_molecular_field}, we also have to decompose the Lagrange parameter.
A priori $ \LambdabUN\in\tangentQR $ is not restricted to be surface conforming.
Therefore, we use the Q-tensor ansatz \eqref{eq:qtensor_decomposition} and stipulate
\begin{align*}
    \LambdabUN 
        &= \Qdepl(\lambdabUN,\lbUN,\lambdabotUN) = \lambdabUN + \lbUN\otimes\normal + \normal\otimes\lbUN + \lambdabotUN\left( \normal\otimes\normal -\frac{1}{2}\IdS \right)
            \in \tangentQR
\end{align*}
for uniquely given $ \lambdabUN\in\tangentQS $, $ \lbUN\in\tangentS $ and $ \lambdabotUN\in\tangentS[^0] $. 
Applying these representations to \eqref{eq:uniaxial_constrain_molecular_field} and using $ 2\qb^2= (\Tr\qb^2)\IdS $ yields $ \HbUN = \Qdepl(\hbUN,\zetabUN,\omegaUN) \in\tangentQR $, where
\begin{align*}
    \hbUN
        &= -6\beta\qb\lambdabUN\qb
           + \frac{1}{2}\left( 2\Tr\qb^2 + 3\beta^2 \right)\projQS(\lambdabUN\qb)\\
         &\quad + \frac{1}{2}\left( 10\beta\lambdabUN\dbdot\qb + \lambdabotUN\left(4\Tr\qb^2 - 9\beta^2 \right)\right)\qb
            - \frac{\beta}{4}\left( 14\Tr\qb^2 - 9\beta^2 \right)\lambdabUN \in\tangentQS\formComma\\
    \zetabUN
        &= -\frac{1}{2} \left( 2\Tr\qb^2 + 3\beta^2 \right)\qb\lbUN 
            - 2\beta(\Tr\qb^2)\lbUN \in\tangentS\formComma\\
    \omegaUN
        &= -\frac{1}{6} \left( 2\Tr\qb^2 - 27\beta^2 \right)\lambdabUN\dbdot\qb
            + 18\lambdabotUN\beta\Tr\qb^2 \in\tangentS[^0]   \formPeriod
\end{align*} 
The constraint $ \mathfrak{B}(\Qb)=0 $ is not yet included.
The equivalent uniaxiality condition given in \eqref{eq:uniaxial_constrain_conforming_simplified} can be read as substitution rule: $ \beta^2 \fb(\qb) = \frac{2}{9}(\Tr\qb^2)\fb(\qb) $ 
for all $ \fb:\tangentQS \rightarrow \tangentR[^n] $ with $ \fb(\nullb)=\nullb $.
Utilizing this rule results in
\begin{subequations}\label{eq:HbUN_decomposed}
\begin{align}
    \HbUN
        &= \Qdepl(\hbUN,\zetabUN,\omegaUN) \in\tangentQR \formComma\\
    \hbUN
        &= -6\beta\qb\lambdabUN\qb
           + \frac{4}{3}(\Tr\qb^2)\projQS(\lambdabUN\qb)\\
        &\quad  + \left( 5\beta\lambdabUN\dbdot\qb + \lambdabotUN\Tr\qb^2\right)\qb \notag
            - \frac{\beta}{4}\left( 14\Tr\qb^2 - 9\beta^2 \right)\lambdabUN \in\tangentQS\formComma\\
    \zetabUN
        &= -\frac{2}{3} (\Tr\qb^2) \left( 2\qb\lbUN + 3\beta\lbUN \right) \in\tangentS\formComma\\
    \omegaUN
        &= \frac{1}{3}(\Tr\qb^2)\left( 2\lambdabUN\dbdot\qb - 9\lambdabotUN\beta \right) \in\tangentS[^0]   \formPeriod
\end{align} 
\end{subequations}
At first glance, this might look problematic, since the uniaxiality constraint equation $ \mathfrak{B}(\Qb)=0 $ becomes surface conforming valued as we see in \eqref{eq:biaxiality_qt_conforming},
though we still end up with a fully Q-tensor valued Lagrange parameter $ \LambdabUN = \Qdepl(\lambdabUN,\lbUN,\lambdabotUN) $ implemented in the constraint force $ \HbUN $.
This could mean that $ \LambdabUN $ is underdetermined through all available equations.
For a non-hydrodynamical system one could simply omit the surface non-conforming part of the Q-tensor equation, which would remove $  \zetabUN $, and therefore also $ \lbUN $, from the system.
However, as we can see in Section \ref{sec:surface_conforming_constrain}, the non-conforming part couples to the fluid equation in a hydrodynamical system, at least through the surface conforming Lagrange parameter $ \lambdabSC $.
Therefore it is not allowed to just omit the non-conforming part, instead we have to include this part into the fluid equation similar to what we demonstrate in Section \ref{sec:surface_conforming_constrain}.
Fortunately, the Lagrange function $ \energyUN $ does not induce an effective fluid constraint force depending on $ \LambdabUN $, see \eqref{eq:uniaxiality_constrain_variation_wrt_X},
\ie\ the coupling of the non-conforming part is still exclusive via $ \lambdabSC $.
Hence the non-conforming force balance $ \nullb = \sum_{\alpha}\zetab_{\alpha} $ alters \eqref{eq:lambdaSC_determined} to $\lambdabSC = 2\left( \zetabEL + \zetabIM + \zetabUN \right)$,
which have to be substituted into the surface conforming constraint stress $ \SigmabSCPhi $ \eqref{eq:surface_conforming_stress_without_lambda}.
Applying $  2\qb^2= (\Tr\qb^2)\IdS $ and  uniaxiality \eqref{eq:uniaxial_constrain_conforming_simplified}, this results in
\begin{align*}
    \SigmabSCPhi &= -2\normal\otimes\left( \qb - \frac{3}{2}\beta\IdS \right)\left( \zetabEL + \zetabIM + \zetabUN \right)\\
                 &= -2\normal\otimes\left(\left( \qb - \frac{3}{2}\beta\IdS \right)\left( \zetabEL + \zetabIM \right) - \frac{2}{3} (\Tr\qb^2)\left( \Tr\qb^2 - \frac{9}{2}\beta^2 \right)\lbUN\right)\\
                 &= -2\normal\otimes\left( \qb - \frac{3}{2}\beta\IdS \right)\left( \zetabEL + \zetabIM\right)\formComma
\end{align*}
which is still the same stress tensor field as given in \eqref{eq:surface_conforming_stress_without_lambda}.
As a consequence, $ \lbUN\in\tangentS $ is indeed arbitrary and thus has not any effects on the solution.
Therefore, in order to prevent underdetermination, we could either set $ \lbUN = \nullb $ \oeda\ or we drop the non-conforming equation $ \nullb = \sum_{\alpha}\zetab_{\alpha} $, which is already included in $ \SigmabSCPhi $, altogether according to Section \ref{sec:surface_conforming_constrain}.

\subsubsection{Isotropic State Constraint (optional)}\label{sec:isotropic_state}

In this section we consider the nematic isotropic state, \ie\ $ \Qb=0 $.
This intention is purely for verification reasons. 
We could have derived the resulting model without considering the Q-tensor field in the first place.
In order to implement this constraint by the Lagrange multiplier technique,
we use the Lagrange function
$ \energyIS := -\innerH{\tangentQR}{\LambdabIS,\Qb} $,
where $ \LambdabIS\in\tangentQR $ is the associated Lagrange parameter, yielding a state variable as an additional degree
of freedom. We obtain the constraint equation
\begin{align}\label{eq:isotropic_constrain}
    \nullb &= \Qb
\end{align}
by $ \innerH{\tangentQR}{\deltafrac{\energyIS}{\LambdabIS},\Thetab}=0 $ for all virtual displacements $ \Thetab\in\tangentQR $,
since $ \energyIS $ is the only energetic term considered to contain $ \LambdabIS $.
The variation $ \innerH{\tangentQR}{\deltafrac{\energyIS}{\Qb}, \Psib} = -\innerH{\tangentQR}{\HbIS, \Psib}$ for all $ \Psib\in\tangentQR $ yields the Q-tensor force
\begin{align*}
    \HbIS &= \LambdabIS \in \tangentQR\formPeriod 
\end{align*}
Material constraint forces are not involved, since the constraint \eqref{eq:isotropic_constrain} results in
\begin{align*}
    \innerH{\tangentR}{\deltafrac{\energyIS}{\para}, \Wb}
        &= - \innerH{\tangentQR}{\HbIS, \gauge{\Wb}\Qb} - \innerH{\tangentQR}{\gauge{\Wb}\LambdabIS + (\Tr\Gb[\Wb])\LambdabIS, \Qb}\\
        &= - \innerH{\tangentQR}{\HbIS, \gauge{\Wb}\Qb}
\end{align*}
for all $ \Wb\in\tangentR $ and thus $ \FbIS=\nullb $.
Note that $ \LambdabIS $ only participate in the Q-tensor force balance \eqref{eq:LDA_molecular} and is fully determined by it.
Therefore, if we wanted to avoid the Lagrange parameter formulation, then we could omit this force balance and the constraint equation \eqref{eq:isotropic_constrain} while satisfying the isotropic state $ \Qb=\nullb $ for all remaining equations.

\subsection{Flux Forces}\label{sec:flux_forces}

\subsubsection{Immobility Forces}\label{sec:immobility}

An approach to control the relaxation rate of the nematic field is to consider the immobility potential
\begin{align}
     \energyIM \label{eq:immobility_flux}
        &:= \frac{M}{2}\normHQRsq{\Dphi\Qb}  
\end{align}
as a part of the flux potential, with $ M \ge 0 $ and $ \Dphi\Qb\in\{\Dmat\Qb , \Djau\Qb\} $ given in \eqref{eq:DmatQ} and \eqref{eq:DjauQ}.
In simple terms, we are assuming a kind of damping in the temporal change of the Q-tensor field depending on the chosen rate,
which causes an emission of energy to the surrounding. 
The choice of the Q-tensor rate by $\Phi\in\{\mfrak,\jau\}$ not only influences the equilibrium state to a certain extent, but above all the dynamics leading to it.
For instance, the minimum equation $\Dphi\Qb=\nullb$  of \eqref{eq:immobility_flux} imply that the director fields of $\Qb$ are frozen into the surface motion for the Jaumann rate ($\Phi=\jau$), whereas for the material rate ($\Phi=\mfrak$), they are frozen into the embedding space. 
We discussed some examples for the kernel of some tangential time derivative on tensor fields in more details in \cite{NitschkeVoigt_JoGaP_2022}.

Variation \wrt\ process variable $ \Dphi\Qb\in\tangentQR $ in direction of $ \Psib\in\tangentQR $ yields
\begin{align}
    \innerH{\tangentQR}{ \deltafrac{\energyIM}{\Dphi\Qb} , \Psib }
        &= - \innerH{\tangentQR}{ \HbIM, \Psib } \formComma\notag\\
    \HbIM \label{eq:HbIM}
        &= - M \Dphi\Qb \in\tangentQR \formPeriod
\end{align}
Since the $ \hil $-inner product does not depend on the surface velocity, lemma \ref{lem:ttensor_veloder_of_timeD} and \ref{lem:ttensor_weak_skewsymmetric_deformation} yield
\begin{align*}
    \innerH{\tangentR}{ \deltafrac{\energyIM}{\Vb} , \Wb }
        &= M\innerH{\tangentQR}{\Dphi\Qb, \veloder{\Wb}\Dphi\Qb}\\
        &= M\left( \innerH{\tangentQR}{\Dphi\Qb, \gauge{\Wb}\Qb} - \innerH{\tangentR[^2]}{(\Dphi\Qb)\Qb + (\Qb \Dphi\Qb )^T, \Phib[\Wb]} \right)\\
        &= M\left( \innerH{\tangentQR}{\Dphi\Qb, \gauge{\Wb}\Qb} - 2\innerH{\tangentR[^2]}{(\Dphi\Qb)\Qb, \Phib[\Wb]} \right)\\
        &=  -\innerH{\tangentQR}{ \HbIM , \gauge{\Wb}\Qb }
                + \innerH{\tangentR\otimes\tangentS}{  \SigmabIM ,\nablaC\Wb}   \formComma
\end{align*}
\begin{align} \label{eq:immobility_stress_field}
    \SigmabIM[\mfrak] &= \nullb \formComma
        &\SigmabIM[\jau] &= M (\Id + \normal\otimes\normal)(\Qb\Djau\Qb - (\Djau\Qb)\Qb)\IdS \in \tangentAS\oplus(\tangentnormal\otimes\tangentS) \formPeriod
\end{align}
Therefore the immobility force field yields $ \FbIM = \DivC\SigmabIM \in \tangentR  $.
Note that the total immobility force vanishes for a consistent choice of process variable and gauge of surface independence,
\ie\ $ \totalFb_{\IM}^{\Phi} = \FbIM + \gaugeFbIM{\Psi} = \nullb $ for $ \Phi=\Psi $.
Hence especially the value of the Jaumann force field $ \FbIM[\jau] $ appears kinda artificial in order to maintain $ \Psi $-independent models as shown in Section \ref{sec:lagrangedalembert}.
Apart from our derivation of the surface Beris-Edwards models, we can relate $ \FbIM[\jau] $ to the ``broken symmetry'' force for tangential polar nematics in \cite{Morris_2022,Kruse_2005}.
Such polar nematic/director fields are very similar to flat degenerated Q-tensor fields $ \Qb=\qb\in\tangentQS $, which in turn represent tangential apolar nematics with a scalar order.
In order to be able to compare this force we also have to relate  $ \Djau\Qb $ to the molecular field.
This relation is given by the Q-tensor force balance $ \nullb = \sum_{\alpha}\Hb_{\alpha} $ \eqref{eq:LDA_molecular}, which contains the Jaumann derivative by $ \HbIM $ \eqref{eq:HbIM}.
A comparison with \cite{Beris_1994} reveals that using \eqref{eq:immobility_flux} as a flux potential instead of a contribution to the kinetic energy (with $M=\rho$) leads to a non-inertia instead of an inertia model.
The term "inertia" here only refers to the kinetics of the nematic field and not to those of the point masses.

\paragraph{Immobility Forces for Surface Conforming Q-Tensor Fields:}
We are considering the surface conforming decomposition $ \Qb=\CQdepl(\qb,\beta)\in\tangentCQR $ \eqref{eq:surface_conforming_ansatz}.
Note that $ \tangentCQR $ is not closed by the material derivative ($ \Phi=\mfrak $) contrarily to the Jaumann derivative ($\Phi=\jau$).
Therefor the non-conforming part of $ \HbIM[\mfrak] $ is not vanishing generally.
In \cite{NitschkeVoigt_2023}, the orthogonal decomposition of both time derivatives restricted to surface conforming Q-tensor fields is sufficiently clarified.
This reference yields
\begin{subequations}\label{eq:immobility_conforming_molucular_field}
\begin{align}
    \HbIM[\mfrak] &= \Qdepl(\hbIM[\mfrak],\zetabIM[\mfrak],\omegaIM[\mfrak])\in\tangentQR \formComma
        &\HbIM[\jau] &= \CQdepl(\hbIM[\jau],\omegaIM[\jau])\in\tangentCQR\formComma\\
    \hbIM[\mfrak] &= -M\dot{\qb} \in\tangentQS\formComma
        &\hbIM[\jau] &= -M\timeJ\qb \in\tangentQS \formComma\\
    \zetabIM[\mfrak] \label{eq:immobility_conforming_molucular_field_zeta}
            &= -M \left(\nabla\vnor + \vb\shop\right)\left( \qb - \frac{3}{2}\beta\IdS \right) \in \tangentS \formComma
        &\zetabIM[\jau] &= \nullb \in\tangentS    \\
    \omegaIM[\mfrak] &= -M\dot{\beta} \in\tangentS[^0] \formComma
        &\omegaIM[\jau] &= -M\dot{\beta} \in\tangentS[^0] \formComma
\end{align}
\end{subequations}
where $  \nabla\vnor + \shop\vb = \normal\nablaC\Vb \in\tangentS $ for $ \Vb=\vb+\vnor\normal\in\tangentR $.
The tangential Q-tensor fields $ \hbIM $ are given in terms of the tangential material derivative $ \dot{\qb} $ and tangential Jaumann derivative $ \timeJ\qb $ of $ \qb\in\tangentQS < \tangentS[^2] $.
The space of tangential Q-tensor fields $ \tangentQS $ is closed by both time derivatives,
see \cite{NitschkeVoigt_JoGaP_2022,NitschkeVoigt_2023} for more details.
We only list $ \zetabIM[\jau] = \nullb $ formally for reasons of consistency. 
The only non-symmetric term in
\begin{align}\label{eq:QDjauQ_conforming}
    \Qb\Djau\Qb
        &= \qb\timeJ\qb - \frac{1}{2}\left( \beta\timeJ\qb + \dot{\beta}\qb \right) + \beta\dot{\beta}\left( \normal\otimes\normal + \frac{1}{4}\IdS \right)
\end{align}
is the first summand, which is also tangential.
As a consequence we obtain
\begin{align*}
    \SigmabIM[\mfrak]  = \sigmabIM[\mfrak] &= \nullb \formComma
      & \SigmabIM[\jau] = \sigmabIM[\jau] &= M \left( \qb\timeJ\qb - (\timeJ\qb)\qb \right) \in\tangentAS
\end{align*}
for substituting  $ \Qb=\CQdepl(\qb,\beta) $ into stress fields \eqref{eq:immobility_stress_field}.

\subsubsection{Nematic Viscous Forces}\label{sec:nematic_viscosity}

In order to model viscous forces depending anisotropically on the nematic order, we consider the linear ansatz 
\begin{align}\label{eq:aniso_metric}
    \IbxiQ 
        &:= \Id - \xi\Qb \in \tangentSymR
\end{align}
for an anisotropic inner metric, where $ \xi\in\R $ is the anisotropy coefficient.
Note that this tensor is locally isotropic, \ie\ $ \IbxiQ\in\tangentIdR $, if and only if $ \xi\Qb = 0 $ locally, since $ \Qb\bot\tangentIdR $.
The fully lower-convected rate \cite{NitschkeVoigt_2023} of $\IbxiQ $ stipulates the nematic viscous potential
\begin{align}
    \energyNV 
        &:= \frac{\coeffIF}{4}\normHsq{\tangentSymR}{\Dlow\IbxiQ}
          = \frac{\coeffIF}{4}\normHsq{\tangentSymR}{\Gb[\Vb]+\Gb^T[\Vb]-\xi\Dlow\Qb} \notag\\
         &= \frac{\coeffIF}{4}\normHsq{\tangentSymR}{\Sb[\Vb]\IbxiQ + \IbxiQ\Sb[\Vb] - \xi\Djau\Qb} \formComma
            \label{eq:disspot_aniso_visc}
\end{align}
where $ \coeffIF \ge 0 $ is the isotropic viscosity of the material.
Ultimately, it quantify the rate of the local inner product given by the anisotropic internal metric tensor \eqref{eq:aniso_metric} when its arguments were frozen in the material flow, \ie
\begin{align*}
    \forall\Rb_1,\Rb_2\in\ker\Dupp<\tangentR:\quad \frac{d}{dt} \IbxiQ(\Rb_1,\Rb_2) &= (\Dlow\IbxiQ)(\Rb_1,\Rb_2)\formComma
\end{align*}
where $ \ker\Dupp $ is the kernel of the upper-convected time derivative \cite{NitschkeVoigt_2023}.
The fully lower-convected rate of $ \IbxiQ $ yields the relations
\begin{subequations}\label{eq:tmp_01}
\begin{align}
     \Dlow\IbxiQ     
        &= -\xi\Dmat\Qb + \Gbcal^T[\Vb]\IbxiQ + \IbxiQ\Gbcal[\Vb] \\ 
        &= -\xi\Djau\Qb + \Sb[\Vb]\IbxiQ + \IbxiQ\Sb[\Vb]  \label{eq:DlowIbxiQ_jauform} 
\end{align}
\end{subequations}
to its corresponding Jaumann and material rate in terms of time derivatives \eqref{eq:DmatQ} and \eqref{eq:DjauQ} of $ \Qb $, 
where we use the $ \R^3 $-metric compatibility $ \Djau\Id = \Dmat\Id = 0 $.
Note that for $ \xi=0 $ the potential \eqref{eq:disspot_aniso_visc} specifies the usual isotropic viscosity potential 
$ \frac{\coeffIF}{4}\| \Gb[\Vb]+\Gb^T[\Vb] \|_{\hilspace{\tangentSymR}}^2 $,
see \eg\ \cite{BachiniKrauseNitschkeVoigt_2023,Arroyo_2009}.

With Lemma \ref{lem:ttensor_variational_independence_of_process_variable} the variation \wrt\ process variable $ \Dphi\Qb$ in direction of $ \Psib\in\tangentQR $ yields
\begin{subequations} \label{eq:nematicvisc_molecular_field}
\begin{align}
    \innerH{\tangentQR}{\deltafrac{\energyNV}{\Dphi\Qb},\Psib}
        &= \innerH{\tangentQR}{\deltafrac{\energyNV}{\Djau\Qb},\Psib}
         = -\frac{\coeffIF\xi}{2} \innerH{\tangentSymR}{\Sb[\Vb]\IbxiQ + \IbxiQ\Sb[\Vb] - \xi\Djau\Qb, \Psib}\notag\\
        &= -\innerH{\tangentQR}{\HbNV,\Psib} \formComma \notag\\
    \HbNV
        &= -\coeffIF\xi\left( \frac{\xi}{2}\Djau\Qb - \projQR \left( \Sb[\Vb]\IbxiQ \right) \right)  \label{eq:nematicvisc_molecular_field_jauform}\\ 
        &= -\coeffIF\xi\left( \frac{\xi}{2}\Dmat\Qb - \projQR \left( \IbxiQ\Gbcal[\Vb] \right) \right)  \label{eq:nematicvisc_molecular_field_matform} 
        \in\tangentQR\formComma
\end{align}
\end{subequations}
where we represent $ \HbNV $ in terms of the Jaumann as well as the material derivative.
Using the initial formulation of the potential \eqref{eq:disspot_aniso_visc}, the relations  of $ \Dlow\IbxiQ $ to its material and Jaumann rate \eqref{eq:tmp_01}, and Lemma \ref{lem:ttensor_veloder_of_timeD} for evaluating the velocity gradient $ \veloder{\Wb}\Dlow\IbxiQ $, the variation \wrt\ material velocity $ \Vb $ in direction of $ \Wb\in\tangentR $ yields
\begin{align} 
\MoveEqLeft
    \innerH{\tangentR}{\deltafrac{\energyNV}{\Vb},\Wb}
        = \frac{\coeffIF}{2} \innerH{\tangentSymR}{\Dlow\IbxiQ, \veloder{\Wb}\Dlow\IbxiQ} \notag\\
      \label{eq:tmp01}
        &=  \frac{\coeffIF}{2} \innerH{\tangentSymR}{ -\xi\Djau\Qb + \Sb[\Vb]\IbxiQ + \IbxiQ\Sb[\Vb], -\xi\gauge{\Wb}\Qb  + \Gbcal^T[\Wb]\IbxiQ + \IbxiQ\Gbcal[\Wb]} \\ 
        &= -\innerH{\tangentQR}{ \HbNV , \gauge{\Wb}\Qb } 
            + \coeffIF \innerH{\tangentR[^2]}{ -\xi\IbxiQ\Djau\Qb + (\IbxiQ)^2\Sb[\Vb] + \IbxiQ\Sb[\Vb]\IbxiQ , \Gbcal[\Wb] } \formComma \notag
\end{align}
where we take advantage of the symmetry, \ie\ $ \Djau\Qb,\IbxiQ,\Sb[\Vb]\in\tangentSymR $.
Moreover, for 
\begin{align*}
    \Rb = -\xi\IbxiQ\Djau\Qb + (\IbxiQ)^2\Sb[\Vb] + \IbxiQ\Sb[\Vb]\IbxiQ
\end{align*}
Lemma \ref{lem:ttensor_weak_deformation} states
\begin{align*}
    \innerH{\tangentR[^2]}{\Rb,\Gbcal[\Wb]}
                &= \innerH{\tangentR\otimes\tangentS}{ \Rb\IdS - \normal\otimes\IdS\Rb\normal , \nablaC\Wb}\formPeriod
\end{align*}
Unfolding $ \IbxiQ $ in the right-handed side term by term yields
\begin{align*}
    -\xi\IbxiQ(\Djau\Qb)\IdS
        &= -\xi  (\Djau\Qb)\IdS + \xi^2  (\Qb\Djau)\Qb\IdS\formComma \\
    (\IbxiQ)^2\Sb[\Vb]\IdS
        &= \Sb[\Vb] - 2\xi \Qb\Sb[\Vb] + \xi^2 \Qb^2\Sb[\Vb]\formComma\\ 
    \IbxiQ\Sb[\Vb]\IbxiQ\IdS
        &= \Sb[\Vb] - \xi\left( \Qb\Sb[\Vb] + \Sb[\Vb]\Qb\IdS \right) + \xi^2 \Qb\Sb[\Vb]\Qb\IdS 
\end{align*}
for the first summand and
\begin{align*}
    \xi\IdS\IbxiQ(\Djau\Qb)\normal
        &= \xi \normal (\Djau\Qb) \IdS -\xi^2\IdS \Qb (\Djau\Qb) \normal \\ 
    -\IdS (\IbxiQ)^2\Sb[\Vb] \normal
        &= \nullb \\
    -\IdS \IbxiQ\Sb[\Vb]\IbxiQ \normal
        &= \xi \normal\Qb\Sb[\Vb] - \xi^2 \normal\Qb\Sb[\Vb]\Qb\IdS 
\end{align*}
for the second summand, where we use that $ \Djau\Qb, \Qb \in \tangentSymR  $ and $ \Sb[\Vb]\in\tangentSymS[^2] $.
Adding all this up and utilizing $ \IdS = \Id-\normal\otimes\normal $ result in
\begin{subequations} \label{eq:nematicvisc_stress_field}
\begin{align}
    \innerH{\tangentR}{\deltafrac{\energyNV}{\Vb},\Wb}
        &= -\innerH{\tangentQR}{ \HbNV , \gauge{\Wb}\Qb } + \innerH{\tangentR\otimes\tangentS}{\SigmabNV , \nablaC\Wb} \formComma \label{eq:tmp02}\\
    \SigmabNV
        &= \coeffIF\Big( 2\Sb[\Vb] 
                    -\xi\left( \projS[^2] \Djau\Qb + 2\proj_{\tangentSymS}(\Qb\Sb[\Vb])  + 2\Qb\Sb[\Vb] \right) \notag\\
        &\quad\quad\quad +\xi^2 \left( \projS[^2](\Qb\Sb[\Vb]\Qb) 
            -\normal\otimes \IdS \Qb (\Djau\Qb) \normal  + \Qb\Djau\Qb\IdS + \Qb^2\Sb[\Vb]\right)\Big) \label{eq:nematicvisc_stress_field_jauform} \\  
        &= \coeffIF\Big( 2\Sb[\Vb] 
                     -\xi\left( \projS[^2] \Dmat\Qb + 2\proj_{\tangentSymS}(\Qb\nablaC\Vb)  + 2\Qb\Sb[\Vb] \right) \notag\\
        &\quad\quad\quad +\xi^2 \big( \projS[^2](\Qb\Gbcal^{T}[\Vb]\Qb) 
               -\normal\otimes\IdS\Qb\left( (\Dmat\Qb)\normal - 2\Abcal[\Vb]\Qb\normal - \Qb\bb[\Vb] \right) \notag\\
        &\quad\quad\quad\quad\quad\quad + \Qb\Dmat\Qb\IdS + \Qb^2\nablaC\Vb\big)\Big) \label{eq:nematicvisc_stress_field_matform} 
                \in\tangentR\otimes\tangentS
\end{align}
\end{subequations}
in terms of Jaumann as well as material derivatives of $ \Qb $.
For the latter identity we apply the time derivative relation \eqref{eq:DjauQ}, $ \nablaC\Vb = \Sb[\Vb]+\Abcal[\Vb]\IdS = ( \Sb[\Vb]-\IdS\Abcal[\Vb] )^T $ and $ \Sb[\Vb]- \Abcal[\Vb] = \Gbcal^T[\Vb] $. 

\begin{remark}
    For Q-tensor fields $ \Qb $ with scalar order parameters $ S_1,S_2\in[0,1] $,
    the anisotropic internal metric $ \IbxiQ = \Id - \xi\Qb $ is symmetric and positive definite if and only if $ \xi\in(-\frac{3}{2}, \frac{3}{2}) $.
\end{remark}
\begin{proof}
    Let $ \Pb\in\tangentR $ be an eigenvector of $ \Qb\in\tangentQR $ to the corresponding eigenvalue $ \lambda\in\tangentS[^0] $, \ie\ it holds $ \Qb\Pb=\lambda\Pb $.
    Since the scalar order parameter of $ \Qb $ are in $ [0,1] $, the range of eigenvectors is restricted to $ \lambda\in [-\frac{2}{3}, \frac{2}{3}] $, see \cite{MottramNewton_2014}.
    The eigenvector space of $ \Qb $ and $ \IbxiQ $ are equal and the corresponding eigenvalue for eigenvector $ \Pb $ is $ 1-\xi\lambda $, \ie\ it is $ \IbxiQ\Pb=(1-\xi\lambda)\Pb $ valid.
    For $ \IbxiQ $ to be positive definite, it must hold $ 1-\xi\lambda > 0 $, which gives the assertion.
\end{proof}

\paragraph{Nematic Viscous Forces for Surface Conforming Q-Tensor Fields:}
Substituting $ \Qb=\CQdepl(\qb,\beta)\in\tangentCQR $ \eqref{eq:surface_conforming_ansatz} into $ \Sb[\Vb]\IbxiQ $ yields
\begin{align*}
    \Sb[\Vb]\IbxiQ &= -\xi\Sb[\Vb]\qb + \left( 1 + \frac{\xi}{2}\beta \right)\Sb[\Vb] \in\tangentSymS \formPeriod
\end{align*}
For its Q-tensor part we can invoke Lemma \ref{lem:qtensorpart_of_tangentialtensor_is_conforming} as a consequence.
Since the Jaumann derivative is compatible with the surface conforming decomposition, \ie\ $ \Djau\CQdepl(\qb,\beta) =\CQdepl(\timeJ\qb,\dot{\beta}) $ \cite{NitschkeVoigt_2023}, 
the Q-tensor field \eqref{eq:nematicvisc_molecular_field_jauform} becomes
\begin{subequations}\label{eq:nematic_viscous_conforming_molecular_field}
\begin{align}
    \HbNV &= \CQdepl(\hbNV,\omegaNV) \in \tangentCQR \formComma\\
    \hbNV &= -\coeffIF\xi\left( \frac{\xi}{2}\timeJ\qb + \xi\projQS(\Sb[\Vb]\qb) - \left( 1 + \frac{\xi}{2}\beta \right)\projQS\Sb[\Vb] \right)  \\ 
          &= -\coeffIF\xi\left( \frac{\xi}{2}\dot{\qb} + \xi\projQS(\qb\Gb[\Vb]) - \left( 1 + \frac{\xi}{2}\beta \right)\projQS\Sb[\Vb] \right) \in\tangentQS \formComma \\ 
    \omegaNV &= -\coeffIF\xi\left( \frac{\xi}{2}\dot{\beta} - \frac{\xi}{3} \Sb[\Vb]\dbdot\qb  + \frac{1}{3}\left( 1 + \frac{\xi}{2}\beta \right)\Tr\Gb[\Vb] \right) \in\tangentS[^0] \formPeriod 
\end{align}
\end{subequations}
The second identity for the tangential Q-tensor part $ \hbNV $ can be obtained by substituting the tangential time derivative relation \eqref{eq:jauq} into the first one or equivalently by evaluating   \eqref{eq:nematicvisc_molecular_field_matform} directly for the given ansatz $ \Qb=\CQdepl(\qb,\beta)$.
To evaluate the stress tensor field \eqref{eq:nematicvisc_stress_field_jauform} we use that
\begin{align*}
    \projS[^2]\Djau\Qb &= \timeJ\qb - \frac{1}{2}\dot{\beta}\IdS \formComma
    &\Qb\Sb[\Vb] &= \qb\Sb[\Vb]-\frac{\beta}{2}\Sb[\Vb]
\end{align*}
holds to get the $ \xi(\ldots) $ part.
For the $ \xi^2(\ldots) $ part, we note that $ \IdS \Qb (\Djau\Qb) \normal $ vanishes by \eqref{eq:QDjauQ_conforming}.
Moreover, it holds
\begin{align*}
     \Qb (\Djau\Qb)\IdS
        &= \qb\timeJ\qb - \frac{1}{2}\left( \beta\timeJ\qb + \dot{\beta}\qb \right) + \frac{1}{4}\beta\dot{\beta}\IdS \formComma\\
     \projS[^2](\Qb\Sb[\Vb]\Qb)
        &= \qb\Sb[\Vb]\qb-\frac{\beta}{2}\left( \qb\Sb[\Vb] + \Sb[\Vb]\qb \right) + \frac{1}{4}\beta^2\Sb[\Vb] \formComma \\
     \Qb^2\Sb[\Vb]
        &= -\beta\qb\Sb[\Vb] + \frac{1}{2}\left( \Tr\qb^2 + \frac{1}{2}\beta^2 \right)\Sb[\Vb] \formComma
\end{align*}
where we use that $ 2\qb^2=(\Tr\qb^2)\IdS $ is valid.
Eventually, the stress field \eqref{eq:nematicvisc_stress_field_jauform} becomes
\begin{align}
    \SigmabNV &= \sigmabNV\notag\\
        &= \coeffIF\bigg( 2\Sb[\Vb] 
                -\xi\left( \timeJ\qb - \frac{\dot{\beta}}{2}\IdS + 3\qb\Sb[\Vb] + \Sb[\Vb]\qb - 2\beta\Sb[\Vb] \right) 
            +\xi^2 \bigg( \qb\timeJ\qb - \frac{1}{2}\left( \beta\timeJ\qb + \dot{\beta}\qb \right)\notag\\ 
        &\quad\quad\quad\quad+ \frac{\beta\dot{\beta}}{4}\IdS 
          + \qb\Sb[\Vb]\qb-\frac{\beta}{2}\left( 3\qb\Sb[\Vb] + \Sb[\Vb]\qb \right)  + \frac{\Tr\qb^2 + \beta^2}{2}\Sb[\Vb] \bigg)
            \bigg) \label{eq:nematicvisc_stress_field_jauform_conforming}\\ 
        &= \coeffIF\bigg( 2\Sb[\Vb] 
                -\xi\left( \dot{\qb} - \frac{\dot{\beta}}{2}\IdS + \qb(\Gb[\Vb]+2\Sb[\Vb]) + \Gb^T[\Vb]\qb - 2\beta\Sb[\Vb] \right) 
            +\xi^2 \bigg( \qb\dot{\qb} - \frac{1}{2}\left( \beta\dot{\qb} + \dot{\beta}\qb \right) \label{eq:nematicvisc_stress_field_materialform_conforming}\\ 
        &\quad\quad+ \frac{\beta\dot{\beta}}{4}\IdS 
          + \qb\Gb^T[\Vb]\qb-\frac{\beta}{2}\left( \qb(\Gb[\Vb]+2\Sb[\Vb]) + \Gb^T[\Vb]\qb \right)  + \frac{\Tr\qb^2}{2}\Gb[\Vb] + \frac{\beta^2}{2}\Sb[\Vb] \bigg)
            \bigg)
        \in\tangentS[^2] \formComma \notag
\end{align}
where we use the tangential time derivative relation \eqref{eq:jauq}.
Alternatively, evaluating \eqref{eq:nematicvisc_stress_field_matform} directly for the given ansatz $ \Qb=\CQdepl(\qb,\beta)$ gives also the second representation of $  \SigmabNV $ in terms of $ \dot{\qb} $.

\subsubsection{Inextensibility Constraint}\label{sec:incompressibility}

If we stipulate the fluid to be inextensible then the material rate $ \dot{\rho} $ of the mass density $ \rho $ vanishes.
As a consequence of mass conservation
\begin{align*}
    0 
        &= \dot{m}\vert_{\mathcal{M}} 
         = \frac{d}{dt}\int_{\mathcal{M}}\rho\dS
         = \int_{\mathcal{M}} \dot{\rho} + \rho\Tr\Gb[\Vb] \dS 
\end{align*}
for all 2-dimensional subsurfaces $ \mathcal{M}\subseteq\surf $ and $ \rho > 0 $, the material velocity $ \Vb\in\tangentR $ has to be divergence-free
according to $\DivC\Vb = \Tr\Gb[\Vb] = 0 $.
To archive this constraint we consider the flux potential
\begin{align*}
    \energyIC 
        &:= -\innerH{\tangentS[^0]}{p , \Tr\Gb[\Vb]}
\end{align*}
as the  Lagrange function for inextensibility with Lagrange parameter $ p\in\tangentS[^0] $, which yields a new degree of freedom.
The same method has been previously published in \cite{Arroyo_2009}.
Formally we declare $ p $ as a process variable, since it is associated to the process variable $ \Vb $ in the Lagrange function.
Therefore, the variation $ \innerH{\tangentS[^0]}{\deltafrac{\energyIC}{p},\theta} $ has to vanish for all virtual displacements $ \theta\in\tangentS[^0] $,
since $ \energyIC $  is the only energetic term considered to contain $ p $.
This yields the desired inextensibility constraint as an additional equation
\begin{align}\label{eq:conti_incompressible}
    0 &= \Tr\Gb[\Vb] = \DivC\Vb = \div\vb - \meanc\vnor
\end{align}
for $ \Vb=\vb+\vnor\normal\in\tangentR $. 
More specifically, \eqref{eq:conti_incompressible} equals the continuity equation for $ \dot{\rho}=0 $ anyway and replaces it hereby.

Variation \wrt\ $ \Vb $ in direction of $ \Wb\in\tangentR $ results in
\begin{align}
    \innerH{\tangentR}{\dfrac{\energyIC}{\Vb},\Wb}
        &= - \innerH{\tangentS[^0]}{\Tr\Gb[\Vb],\veloder{\Wb}p} - \innerH{\tangentS[^0]}{p, \Tr\Gb[\Wb]}
         =  \innerH{\tangentR}{\SigmabIC,\nablaC\Wb} \formComma\notag\\
    \SigmabIC \label{eq:SigmabIC}
        &= \sigmabIC
         = -p \IdS \in\tangentIdS \formComma
\end{align}
where the first summand vanishes due to inextensibility.
Therefore, the inextensibility force is $ \FbIC = \DivC\SigmabIC = -\GradC p $ by relations of the adjoint gradient \eqref{eq:GradC}.
This justifies to interpret $ p $ as the pressure.
Even though we actually derive inextensible models in this paper, all other stress and generalized force fields in section \ref{sec:derivation} do not relay on the inextensibility constraint \eqref{eq:conti_incompressible} and could also be used for compressible models. 

\subsubsection{No Normal Flow Constraint (optional)}\label{sec:no_normal_flow}

In order to force a geometrically stationary surface we stipulate a vanishing normal velocity, \ie\ $ \vnor = \Vb\normal = 0 $.
Note that this is the only situation where a purely Eulerian observer exists at a moving surface.
We implement this constraint by the Lagrange function
\begin{align*}
    \energyNN
        &:= -\innerH{\tangentS[^0]}{\lambdaNN,\vnor}
          = -\innerH{\tangentR}{\lambdaNN\normal,\Vb} \formComma
\end{align*}
where $ \lambdaNN\in\tangentS[^0] $ is the associated Lagrange parameter, yielding a process variable as an additional degree of freedom.
We obtain the constraint equation
\begin{align}\label{eq:no_normal_flow_constrain}
    0 &= \Vb\normal = \vnor
\end{align}
by $ \inner{\tangentS[^0]}{\deltafrac{\energyNN}{\lambdaNN},\theta}=0 $ for all virtual displacements $ \theta\in\tangentS[^0] $,
since $ \energyNN $ is the only energetic term considered to contain $ \lambdaNN $. 
The variation $ \innerH{\tangentR}{\deltafrac{\energyNN}{\Vb}, \Wb} = -\innerH{\tangentR}{\FbNN, \Wb} $ in arbitrary direction $ \Wb\in\tangentR $, yields the constraint force
\begin{align} \label{eq:FbNN}
    \FbNN &= \lambdaNN\normal \in \tangentnormal < \tangentR \formPeriod
\end{align}
The Lagrange parameter $ \lambdaNN $ appears solely in the normal direction of the fluid equation and is therefore only determined by it.
Due to this we could drop the normal fluid equation and omit $ \vnor $ from all remaining equations by constraint \eqref{eq:no_normal_flow_constrain}, if we are not interested in the constraint force $ \FbNN $.

\paragraph{Flat Surfaces:}

Since the condition \eqref{eq:no_normal_flow_constrain} yields a geometrically stationary surface, geometric quantities stay stationary also from an Eulerian perspective.
Other observer only perceive advection of such quantities along the observer flow.
Therefore, if a geometrical quantity is globally zero from the beginning, then it stays that way for every observer with $ \vnor=0 $.
As a consequence, a flat surface can be stipulated by the initial constraint
\begin{align}\label{eq:flat_surface_initial_constrain}
    \nullb = \shop\vert_{t=t_0}
\end{align}
additionally to the constraint $ \vnor=0 $.
This is equivalent to the constraint $ \nullb = \shop $ at all times under \eqref{eq:no_normal_flow_constrain}.
To make it clear, there is no need to augment the Lagrange multiplier technique on $ \vnor=0 $ for a flat surface.
We only have to consider the initial constraint \eqref{eq:flat_surface_initial_constrain} and substitute $\shop = \nullb $ wherever it appears.
Any constraint forces that may occur are already included in $ \FbNN $.

\subsubsection{No Flow Constraint (optional)}\label{sec:no_flow}

In this section we consider materially stationary surfaces, \ie\ the surface is at rest by stipulating a vanishing material velocity.
This intention is purely for verification reasons. 
We could have derived the resulting model without considering the velocity in the first place, \eg\ by formulating the associated $ \hil $-gradient flow.
In order to implement this constraint by the Lagrange multiplier technique,
we use the Lagrange function
\begin{align*}
    \energyNF
        &:= -\innerH{\tangentR}{\LambdabNF,\Vb} \formComma
\end{align*}
where $ \LambdabNF\in\tangentR $ is the associated Lagrange parameter, yielding a process variable as an additional degree of freedom.
We obtain the constraint equation
\begin{align}\label{eq:no_flow_constrain}
    \nullb &= \Vb
\end{align}
by $ \innerH{\tangentR}{\deltafrac{\energyNF}{\LambdabNF},\Thetab}=0 $ for all virtual displacements $ \Thetab\in\tangentR $,
since $ \energyNF $ is the only energetic term considered to contain $ \LambdabNF $. 
The variation $ \innerH{\tangentR}{\deltafrac{\energyNF}{\Vb}, \Wb} = -\innerH{\tangentR}{\FbNF, \Wb} $ in arbitrary direction $ \Wb\in\tangentR $  yields the constraint force
\begin{align}\label{eq:FbNF}
    \FbNF &= \LambdabNF\in \tangentR \formPeriod
\end{align}
The Lagrange function does not depend on $ \Qb $ and thus the Q-tensor constraint force is $ \HbNF=\nullb $.
The Lagrange parameter $ \LambdabNF $ appears solely in the fluid equation and is therefore only determined by it.
Due to this we could drop the fluid equation and omit $ \Vb $ from all remaining equations by constraint \eqref{eq:no_flow_constrain}, if we are not interested in the constraint force $ \FbNF $.
Eventually, only the Q-tensor equation \eqref{eq:L2flow} remains, which describes the usual local $ \hil $-gradient flow equation for the free energy $ \energyEL + \energyTH $ under possibly considered constraints on $ \Qb $.
There is no longer a choice of $ \Phi\in\{\mfrak,\jau\} $ to describe the dissipation mechanism for $ \Qb $, since it holds $ \Dmat\Qb=\Djau\Qb =:\Dt\Qb $.
Moreover the Lagrangian and Eulerian observer are equal \wrt\ stationary materials and therefore $ \Dt\Qb=\partial_t\Qb $.

\subsection{Total Energy Rate}\label{sec:energy_rate}

In this section we show thermodynamical consistency of the total energy $ \energyTOT $ by revealing that $ \ddt\energyTOT $  solely depends on the energy flux potential $ \fluxpotential $.
The total energy
\begin{align*}
    \energyTOT 
        &:= \energyK + \potenergy
          =  \energyK + \energyEL + \energyTH + \energyBE  
\end{align*}
comprises the kinetic energy $ \energyK $ \eqref{eq:kinetic_energy} and potential energy $ \potenergy $,
which in turn includes the elastic energy $ \energyEL $ \eqref{eq:energy_elastic}, the thermotropic energy $ \energyTH $ \eqref{eq:energy_thermotropic} and the bending energy $ \energyBE  $ \eqref{eq:energy_bending}.
Note that all Langrage functions implemented as potential energy are not considered in $ \potenergy $, since they all are vanishing trivially under their associated constraint.
By the transport formula $ \ddt\int_{\surf} f \dS = \int_{\surf} \dot{f} + f\Tr\Gb[\Vb] \dS $ and local mass conservation $ \dot{\rho} + \rho\Tr\Gb[\Vb] = 0 $ the rate of kinetic energy yields
\begin{align}\label{eq:kinetic_energy_rate}
    \ddt\energyK
        &= \innerH{\tangentR}{\frac{1}{2}\left( \dot{\rho} + \rho\Tr\Gb[\Vb] \right)\Vb + \rho\Dmat\Vb, \Vb}
         = \innerH{\tangentR}{\rho\Dmat\Vb, \Vb} \formPeriod
\end{align}
Calculating the potential energy rate $ \ddt\potenergy $ directly is about as costly and quite similar to its variation procedures given in Sections \ref{sec:elastic}--\ref{sec:bending}.
For this reason we use the chain rule 
\begin{align}\label{eq:temporal_energetic_chain_rule_partial_form}
    \ddt\potenergy
        &= \innerH{\tangentR}{\frac{\partial \potenergy}{\partial \para}, \Vb} + \innerH{\tangentQR}{\frac{\partial \potenergy}{\partial \Qb}, \Dmat\Qb}
\end{align}
instead.
We obtain this chain rule for $ \ddt\potenergy = \lim_{\tau\rightarrow 0} \frac{1}{\tau}\left( \potenergy\vert_{t+\tau} - \potenergy\vert_t \right) $ by Taylor expansion at $ \tau=0 $ and the fact that $ \potenergy $ is instantaneous, which results in
\begin{align*}
    \potenergy\vert_{t+\tau}
        &= \potenergy[\para\vert_{t+\tau}, \Qb[\para]\vert_{t+\tau}]
         = \potenergy[ \para\vert_t + \tau\Vb\vert_t + \landau(\tau^2) , \Qb[\para]\vert_t + \tau\Dmat\Qb[\para]\vert_t + \landau(\tau^2) ]\\
        &= \potenergy\vert_t + \tau\left( \innerH{\tangentR}{\frac{\partial \potenergy}{\partial \para}, \Vb} + \innerH{\tangentQR}{\frac{\partial \potenergy}{\partial \Qb}, \Dmat\Qb} \right)\Bigg\vert_t + \landau(\tau^2)
\end{align*}
for $ \potenergy = \potenergy[\para,\Qb[\para]] $, \cf\ \cite{NitschkeSadikVoigt_A_2022} for more details in this kind of process.
Note that \eqref{eq:temporal_energetic_chain_rule_partial_form} uses partial variations, which variate only the partial arguments of $ \potenergy[\para,\Qb[\para]] $ without considering mutual dependencies of the arguments $ \para $ and $ \Qb=\Qb[\para] $. 
Moreover these partial variations depend on the syntactical description of $ \potenergy $, see \cite{NitschkeSadikVoigt_A_2022}.
Since the parameterization $ \para $ does not depend on $ \Qb $ a priori, it holds 
$ \innerH{\tangentQR}{\frac{\partial \potenergy}{\partial \Qb}, \Psib} = \innerH{\tangentQR}{\deltafrac{\potenergy}{\Qb}, \Psib} $ anyway for all $ \Psib\in\tangentQR $.
However, $ \Qb $ is not a priori independent of $ \para $. 
Its first order dependency is determined by the choice of a gauge of surface independence, such as \eqref{eq:material_gauge} or \eqref{eq:jaumann_gauge}.
With the same process as in \cite{NitschkeSadikVoigt_A_2022} we 
obtain\footnote{The only noteworthy difference to \cite{NitschkeSadikVoigt_A_2022} is that we have to omit the tangential projection, since $ \Qb $ is a $ \R^3 $-quantity.
                Therefore the Taylor expansion at the unperturbed surface just reads
                $
                    \Qb[\para+\eps\Wb] = (\Qb[\para])^{*_{\eps\Wb}} + (\gauge{\Wb}\Qb[\para])^{*_{\eps\Wb}} + \landau(\eps^2) \formPeriod
                $
                }  
the identity
\begin{align*}
    \innerH{\tangentR}{\deltafrac{\potenergy}{\para}, \Wb}
        &= \innerH{\tangentR}{\frac{\partial \potenergy}{\partial \para}, \Wb} + \innerH{\tangentQR}{\deltafrac{\potenergy}{\Qb}, \gauge{\Wb}\Qb} \formPeriod
\end{align*}
As a consequence, the chain rule \eqref{eq:temporal_energetic_chain_rule_partial_form} holds
\begin{align*}
    \ddt\potenergy
        &= \innerH{\tangentR}{\deltafrac{\potenergy}{\para}, \Vb}
            + \innerH{\tangentQR}{\deltafrac{\potenergy}{\Qb}, \Dmat\Qb - \gauge{\Vb}\Qb}
\end{align*}   
in terms of syntax 
independent\footnote{For instance and contrarily to partial variations, the total variation \wrt\ $\para$ is always yielding the same results for $ \potenergy=\potenergy[\para,\Qb[\para]] $, $ \potenergy = \potenergy[\para, \{Q^\alpha[\para]\}_{\alpha=1}^{5} ]$, $ \potenergy=\potenergy[\para,\Qb[\para], \{g_{ij}[\para]\}, \shop[\para]] $, \etc, as long as $ \potenergy $ describes the same energy, see \cite{NitschkeSadikVoigt_A_2022} for more details.} 
total variations.
This leads to the total energy rate
\begin{align}
    \ddt\energyTOT
        &= \innerH{\tangentR}{\rho\Dmat\Vb - \totalFb_{\potenergy}, \Vb}
            - \innerH{\tangentQR}{\Hb_{\potenergy}, \Dmat\Qb - \gauge{\Vb}\Qb} \notag\\
        &= \innerH{\tangentR}{\totalFb_{\fluxpotential}, \Vb} + \innerH{\tangentQR}{\Hb_{\fluxpotential}, \Dmat\Qb - \gauge{\Vb}\Qb}\notag\\
        &= \innerH{\tangentR}{\Fb_{\fluxpotential}, \Vb} + \innerH{\tangentQR}{\Hb_{\fluxpotential}, \Dmat\Qb}
            \label{eq:total_energy_rate_general}
\end{align}
by definition of Lagrange forces in \eqref{eq:Lagrange_forces_def}, fluid equation \eqref{eq:LDA_fluid}, Q-tensor equation \eqref{eq:LDA_molecular} and the force separation \eqref{eq:total_force_splitting_weak}.
A first statement about thermodynamic consistency is thus already answered.
The total energy rate, \wrt\ solutions of the surface Beris-Edwards models, does not depend on the potential energy $ \potenergy $ and even vanishes for closed systems, \ie\ in the absence of a flux potential $ \fluxpotential $ determining the energy transfer between the system under consideration and untreated energy reservoirs.
Additionally, \eqref{eq:total_energy_rate_general} validates that  the rate is independent of $ \gauge{\Wb}\Qb $.
Note that the surface Beris-Edwards models,  as well as its solution, do not depend on $ \gauge{\Wb}\Qb $.
Therefore, this holds also for the state measurement $ \energyTOT $ and its rate, which is confirmed by \eqref{eq:total_energy_rate_general}.

In the following we calculate the total energy rate in more detail \wrt\ all in this paper given flux potentials $ \fluxpotential_{\alpha} $ with 
$ \alpha\in\textup{A}_{\fluxpotential}:= \{\IM,\AC,\NV,\IC,\NN,\NF\} $, \ie\ it is $ \fluxpotential = \sum_{\alpha\in\textup{A}_{\fluxpotential}}\fluxpotential_{\alpha}  $ valid.
To save writing efforts we separate the total energy rate by $ \ddt\energyTOT = \sum_{\alpha\in\textup{A}_{\fluxpotential}} \dot{\energy}_{\alpha} $
with summands $ \dot{\energy}_{\alpha} := \innerH{\tangentR}{\Fb_{\alpha}, \Vb} + \innerH{\tangentQR}{\Hb_{\alpha}, \Dmat\Qb} $.
For the considered constraints, solely given by the incompressibility stress \eqref{eq:SigmabIC}, no normal flow force \eqref{eq:FbNN} and no flow force \eqref{eq:FbNF}, we trivially get
\begin{align*}
    \dot{\energy}_{\IC} &= -\innerH{\tangentSymS}{\SigmabIC, \Gb[\Vb]} = 0\formComma
    &\dot{\energy}_{\NN} &= \innerH{\tangentR}{\FbNN, \Vb} = 0\formComma
    &\dot{\energy}_{\NF} &= \innerH{\tangentR}{\FbNF, \Vb} = 0
\end{align*}
as expected under their associated constraints \eqref{eq:conti_incompressible}, \eqref{eq:no_normal_flow_constrain} and \eqref{eq:no_flow_constrain}, if applied respectively.
The material immobility ($ \Phi=\mfrak $) yields
\begin{align*}
     \dot{\energy}_{\IM}\vert_{\Phi=\mfrak}
         &= -M \normHsq{\tangentQR}{\Dmat\Qb}
         = -2 \energyIM[\mfrak]
\end{align*}
by the vanishing material immobility stress $ \SigmabIM[\mfrak] $ \eqref{eq:immobility_stress_field}, Q-tensor force $ \HbIM[\mfrak] $ \eqref{eq:HbIM} and flux potential $ \energyIM[\mfrak] $ \eqref{eq:immobility_flux}.
A similar result is obtained for the Jaumann immobility ($ \Phi=\jau $).
Using the Jaumann immobility stress $ \SigmabIM[\jau] $ \eqref{eq:immobility_stress_field}, Q-tensor force $ \HbIM[\jau] $ \eqref{eq:HbIM}, flux potential $ \energyIM[\jau] $ \eqref{eq:immobility_flux},
time derivatives relation between $ \Dmat\Qb $ and $ \Djau\Qb $ \eqref{eq:DjauQ},
and the identity $ \Abcal[\Vb] = \proj_{\tangentAR}(\nablaC\Vb + \normal\otimes\normal\nablaC\Vb) $, following from \eqref{eq:Gbcal} and \eqref{eq:Abcal}, we calculate
\begin{align*}
    \dot{\energy}_{\IM}\vert_{\Phi=\jau}
        &= - \innerH{\tangentR[^2]}{\SigmabIM[\jau], \nablaC\Vb} - M \innerH{\tangentQR}{\Djau\Qb , \Dmat\Qb}\\
        &= - M \left( \innerH{\tangentAR}{ \Qb\Djau\Qb - (\Djau\Qb)\Qb , \Abcal[\Vb]} 
                    + \innerH{\tangentQR}{\Djau\Qb , \Djau\Qb + \Abcal[\Vb]\Qb - \Qb\Abcal[\Vb]\Qb }\right) \\
        &= -M \normHsq{\tangentQR}{\Djau\Qb}
         = -2 \energyIM[\jau] \formPeriod
\end{align*}
For the nematic viscosity we shorten the way a bit.
We use the backward calculation from \eqref{eq:tmp02} to \eqref{eq:tmp01}, which is still valid for $ \gauge{\Wb}\Qb=\nullb $ and $ \Wb=\Vb $.
This leads to
\begin{align*}
    \dot{\energy}_{\NV}
        &= -\innerH{\tangentR[^2]}{\SigmabNV, \nablaC\Vb}
           + \innerH{\tangentQR}{\HbNV , \Djau\Qb + \Abcal[\Vb]\Qb - \Qb\Abcal[\Vb]\Qb } \\
        &= -\frac{\coeffIF}{2}\bigg( \innerH{\tangentR[^2]}{\Dlow\IbxiQ, \Gbcal^T[\Vb]\left( \Id - \xi\Qb \right) + \left( \Id - \xi\Qb \right)\Gbcal[\Vb]} \\
                    &\quad\quad\quad\quad - \xi \innerH{\tangentR[^2]}{\Dlow\IbxiQ, \Djau\Qb + \frac{1}{2}\left( \Gbcal[\Vb] - \Gbcal^T[\Vb] \right)\Qb - \frac{1}{2}\Qb\left( \Gbcal[\Vb] - \Gbcal^T[\Vb] \right)}\bigg) \\
        &= -\frac{\coeffIF}{2} \innerH{\tangentSymR}{ \Dlow\IbxiQ, -\xi \Djau\Qb - \xi\Sb[\Vb]\Qb - \xi \Qb\Sb[\Vb] + 2\Sb[\Vb] }
         = - 2\energyNV
\end{align*}
by the stress $ \SigmabNV $ \eqref{eq:nematicvisc_stress_field_jauform}, Q-tensor force $ \HbNV $ \eqref{eq:nematicvisc_molecular_field_jauform}, flux potential $ \energyNV $ \eqref{eq:disspot_aniso_visc}, time derivatives relation \eqref{eq:DjauQ}, anisotropic metric $ \IbxiQ $ \eqref{eq:aniso_metric} and its lower-convected rate  $ \Dlow\IbxiQ $ \eqref{eq:DlowIbxiQ_jauform}.
Finally, the total energy rate \eqref{eq:total_energy_rate_general} results in
\begin{align}\label{eq:total_energy_rate}
    \ddt\energyTOT
            &= -2 \left( \energyIM + \energyNV \right) \le 0 \formPeriod
\end{align}
As a consequence, changes of the energy of our system under consideration are solely determined by the energy transfer mechanisms applied in this paper.
Moreover, the system is purely dissipative, since $ \energyIM, \energyNV \ge 0 $.
 
\subsection{Physical relations} 
\label{sec:phys}

For space dimension $d = 2, 3$ a connection can be established between Beris-Edwards models \cite{Beris_1994} and Ericksen-Leslie models \cite{ericksen1961conservation,Leslie_1968}. For rigorous mathematical proofs see \cite{Wang_SIAMJMA_2015}. To establish this connection for the derived surface models, see, e.g., \cite{nitschke2019hydrodynamic} for corresponding surface Ericksen-Leslie models, is beyond the scope of this paper. However, one essential physical implication can also be shown for the surface models. This relation is concerned with the nematic viscous stress tensor field. We demonstrate for a surface conforming uniaxial field with a tangential principal director consistency with the Parodi-Leslie relations \cite{Leslie_1966,Parodi_1970,Currie_1974,Beris_1994}. Let $ \Qb=s(\mathbf{d}\otimes\mathbf{d}-\frac{1}{3}\Id) $, where $ s\in\tangentS[^0] $ is a temporally constant scalar order field, \ie\ $ \dot{s}=0 $, and $ \mathbf{d}\in\tangentS $ is a tangential director field, \ie\ $ \mathbf{d}\mathbf{d}=1 $.
A comparison with ansatz \eqref{eq:surface_conforming_ansatz} results in $ \Qb=\CQdepl(\qb,\beta)\in\tangentCQR $ with $ \qb=s(\mathbf{d}\otimes\mathbf{d}-\frac{1}{2}\IdS)\in\tangentQS $ and $ \beta=-\frac{s}{3}\in\tangentS[^0] $.
Moreover, we use that $ \timeJ\qb=s((\timeJ\mathbf{d})\otimes\mathbf{d}+\mathbf{d}\otimes\timeJ\mathbf{d}) $, $ \dot{\beta}=0 $ and $ \mathbf{d}\bot\timeJ\mathbf{d} $ holds, see \cite{NitschkeVoigt_JoGaP_2022,NitschkeVoigt_2023}.
Substituting $ \qb $ and $ \beta $ into \eqref{eq:nematicvisc_stress_field_jauform_conforming}, or \eqref{eq:nematicvisc_stress_field_materialform_conforming} equivalently, yields
\begin{gather*}
    \sigmabNV
        = \alpha_1 \Sb[\Vb](\mathbf{d},\mathbf{d})\mathbf{d}\otimes\mathbf{d}
          +\alpha_2 (\timeJ\mathbf{d})\otimes\mathbf{d}
          +\alpha_3 \mathbf{d}\otimes\timeJ\mathbf{d}
          +\alpha_4 \Sb[\Vb]
          +\alpha_5 \Sb[\Vb]\mathbf{d}\otimes\mathbf{d}
          +\alpha_6 \mathbf{d}\otimes\Sb[\Vb]\mathbf{d} \formComma\\
\begin{align*}
    \alpha_1 &= \coeffIF s^2\xi^2 \formComma
        &\alpha_2 &= -\coeffIF s\xi \left( 1 + \frac{s\xi}{3} \right) \formComma
            &\alpha_3 &= -\coeffIF s\xi \left( 1 - \frac{2}{3}s\xi \right) \formComma \\
    \alpha_4 &= 2\coeffIF\left( 1 + \frac{s\xi}{3} \right)^2 \formComma
        &\alpha_5 &= -\coeffIF s\xi \left( 1 + \frac{s\xi}{3} \right) \formComma
            &\alpha_6 &= -3 \coeffIF s\xi \formPeriod
\end{align*}
\end{gather*}
Time-reversal invariance leads to the Parodi relation \cite{Parodi_1970}
\begin{align*}
    \alpha_2 + \alpha_3  &= \alpha_6 - \alpha_5 &&\left(= \coeffIF s\xi \left( \frac{s\xi}{3} - 2 \right)\right) \formComma
\end{align*}
which is fulfilled for $ \sigmabNV $.
Please note that neither the material $\sigmabIM[\mfrak]=\nullb$ nor the Jaumann immobility stress tensor field $\sigmabIM[\jau]= M s^2 ( \mathbf{d}\otimes\timeJ\mathbf{d} - (\timeJ\mathbf{d})\otimes\mathbf{d} )$ makes any contributions in this context.
Furthermore, $ \sigmabNV $ respects the Leslie relations
\begin{gather*}
\begin{align*}
    0 &\le \alpha_3 - \alpha_2 &&\left(= 2\coeffIF\left( 1 + \frac{s\xi}{3} \right)^2\right) \formComma
  & 0   &\le 2\alpha_4 + \alpha_5 + \alpha_6 &&\left(=\coeffIF\left( 2 - \frac{s\xi}{3} \right)^2\right)\formComma \\
    0   &\le\alpha_4            &&\left(= 2\coeffIF\left( 1 + \frac{s\xi}{3} \right)^2\right)\formComma
  &0 &\le \alpha_1 + \alpha_4 + \alpha_5 + \alpha_6 &&\left(=2\coeffIF\left( 1 - \frac{2}{3}s\xi \right)^2\right)\formComma
\end{align*}\\
\begin{align*}
    0   &\le \left( \alpha_3 - \alpha_2 \right)\left( 2\alpha_4 + \alpha_5 + \alpha_6 \right) - \left(\alpha_6 - \alpha_5\right)^2  &&\left(=0\right)\formPeriod
\end{align*}
\end{gather*}
as given in \cite{Parodi_1970}.
The relation $ \alpha_1 + \alpha_4 + \alpha_5 + \alpha_6 \ge 0 $ is slightly less restrictive in \cite{Leslie_1966,Currie_1974,Beris_1994} and therefore holds as well.
With the inclusion of immobility, the relations remain almost unaltered.
Only the last relation is no longer trivial for  the Jaumann immobility, but still satisfied by $ 2M\coeffIF s^2\left( 2 - \frac{s\xi}{3} \right)^2 \ge 0 $.
The reason for separately considering immobility, rather than integrating it together  with nematic viscosity in first place, stems from the fact that these relations are derived for inertia models in which $ \frac{\rho}{2}\dot{\mathbf{d}}\dot{\mathbf{d}} $ is incorporated into the kinetic energy. 
The examination of dynamics through immobility can be regarded as a kind of limit of over-damped rotational kinetics.
In our view, these relations should still be considered valuable in this situation, especially if the rotation of the director field is inhibited relatively to the material's local rotation.

\section{Discussion}\label{sec:discussion}

Using a Lagrange-D'Alembert principle we have derived surface Beris-Edwards models on evolving surfaces. 
The derivation leads to thermodynamically consistent formulations allowing for the simultaneous relaxation of the surface Q-tensor field and the shape of the surface by taking the tight interplay between geometry and flow field and flow field and Q-tensor field into account. 
We have considered two different models $ \Phi\in\{\mfrak,\jau\} $.
We do not like to judge which choice of $ \Phi $ is the ``correct'' one. But in our understanding we believe that the material models ($ \Phi=\mfrak $) are suitable if the nematic field lays on the fluid membrane, and the Jaumann models ($ \Phi=\jau $) if the nematic field describes the fluid membrane. In any case, this is not a mathematical but a physical issue. Differences in the solution will be small if the nematic field is sufficiently close to its equilibrium, \eg\ for a small immobility coefficient $ M $ with a convenient initial condition.
However, as already shown for simpler models, out of equilibrium the dynamics can strongly differ \cite{NitschkeSadikVoigt_A_2022}. The derived surface Beris-Edwards models \eqref{eq:model_lagrange_multiplier} and \eqref{eq:model_conforming}, with potential constraints have been shown to reduce to known model equations in the literature if considered in special simplified situations, e.g., for nematic shells, fluid deformable surfaces or nematodynamics in flat space, see Secs. \ref{sec:model_general_cases} and \ref{sec:surface_conforming_special_cases}.

While most model formulations for surface liquid crystals consider a surface conforming formulation, we consider next to this formulation also a general formulation without the necessity of tangential anchoring. This general surface Beris-Edwards model can be seen as a modular model allowing for incorporating constraints by Lagrange multipliers. Surface conformity is one of these constraints. The obtained formulations are equivalent. But we believe the general form to be computationally more feasible. This has been demonstrated for fluid deformable surfaces \cite{Krause_2023,Krause_PAMM_2023,BachiniKrauseNitschkeVoigt_2023}. However, a computational realization for the surface Beris-Edwards model remains subject of future work. Several comments on the Lagrange multiplier approach should be made. As one can easily see, the surface conforming model \eqref{eq:model_conforming} exhibits imperfections under the uniaxial condition. The constraint equation $ \nullb=\Cb_{\UN} $ stipulates $ 0 = (2\Tr\qb^2-9\beta^2)\Tr\qb^2 $ (Tab.~\ref{tab:constraints_conforming}), which in turn makes the second part $ \nullb = (2\Tr\qb^2-9\beta^2)\beta\qb $ of the constraint redundant, since $ \Tr\qb^2=0 $ if and only if $ \qb=\nullb $. However, removing this second part would result in fewer equations than degrees of freedom, indicating an underdetermined system. As there are no uniaxial constraint forces in the fluid equation, and the Q-tensor field is sufficiently determined by the Q-tensor and the constraint equation, we believe that this underdetermination mainly affects the Lagrange parameter $ \LambdabUN\cong\left( \lambdabUN, \lambdabotUN \right) $, though further investigation is required. Nevertheless, if this is indeed the case, it does not impact the actual solution $ (\vb,\vnor,\qb,\beta) $. We speculate that this issue may also arise in the more general model \eqref{eq:model_lagrange_multiplier}, although it may not be as readily apparent. From an analytical standpoint, this does not pose significant challenges. However, numerical discretizations should take this into account. For instance, we believe that careful selection of linearization techniques and linear solvers can mitigate the underdetermination numerically. Alternatively a penalty method could be used instead of the Lagrange multiplier technique. E.g. by incorporating the biaxiality potential energy given by
        \begin{align}\label{eq:biaxiality_energy}
            \potenergy_{\UN} 
                &:= \varpi\int_{\surf} \mathfrak{b}(\Qb) \dS 
                = \varpi\int_{\surf}(\Tr\Qb^2)^{3} - 6(\Tr\Qb^3)^2 \dS
                \ge 0
        \end{align} 
where $ \varpi \ge 0 $ is a penalization parameter and $ \mathfrak{b}(\Qb) $ is a biaxial measure \eqref{eq:biaxiality_measure_dimensionalized}. This potential energy compels the solution to approximate uniaxiality for $ \varpi \gg 0 $ due to the resulting Lagrangian forces
        \begin{align*}
            \FbUN &= \varpi\GradC \mathfrak{b}(\Qb)\formComma
            & \HbUN &= -6\varpi\left( (\Tr\Qb^2)^{2}\Qb - 6(\Tr\Qb^3)\Qb^2 \right)
        \end{align*}
which need to be added to the general fluid equation \eqref{eq:model_lagrange_multiplier_floweq} and Q-tensor equation \eqref{eq:model_lagrange_multiplier_moleculareq} in place of the constraint terms used in the Lagrange multiplier technique. For the surface conforming model \eqref{eq:model_conforming}, it is necessary to substitute $ \Qb=\CQdepl(\qb,\beta) $ \eqref{eq:surface_conforming_ansatz} and perform the corresponding decompositions \eqref{eq:force_decomposition} and \eqref{eq:Qtensorforce_decomposition}. A similar approach was previously published in \cite{Nestler_2020}, albeit for pure tangential uniaxiality, which does not permit phase transitions between tangential and normal director fields. However, based on our experience, penalty methods require careful selection of the penalty parameter to achieve both a sufficiently accurate approximation  and numerical stability. This means that we have to find a ``sweet spot'' for $ \varpi $  at first place or by adaption through multiple simulations. Parameter studies for similar problem formulations can be found in \cite{Nestler_2019}, where  penalty approaches were used to enforce tangentiality of director fields and Q-tensor fields. Examples for penalized surface conformity can be found in \cite{BouckNochettoYushutin_2022,Nestler_2020}. Overall, we would like to emphasize that penalty methods are generally compatible with the Lagrange-D'Alembert principle. 

Our last comments concern potential model extensions. We have chosen a one-constant approximation for the sake of simplicity, e.g., $ L = L_1 $ as the only non-vanishing elastic constant. Considering an elastic energy $ \energyEL $ \eqref{eq:energy_elastic} also comprising $ L_2 $,$ L_3 $,$ L_4 $ and $ L_6 $-terms \cite{Golovaty_JNS_2017,Nitschke_2018} would have made this paper even more technically. 
However, the provided tools should be sufficient to derive the remaining elastic forces and adapt the surface Beris-Edwards models appropriately. 
We omitted the treatment of boundary conditions as it would have opened up a separate and extensive topic.
The local Lagrange-d'Alembert equations \eqref{eq:LDA} and everything derived from it holds only for boundaryless surfaces or with sufficient periodicity at the boundary for all involved degrees of freedom, especially including the geometry.
However, boundary conditions could be included in the step going from the spatially global equation \eqref{eq:LDA_time_local} to the local equations \eqref{eq:LDA}.
The obtained boundary integral has to vanish either by natural or essential boundary conditions, or a combination of both.
Furthermore we only discussed inextensible fluids. For compressible fluids one can use $ \fluxpotential_{PR}=\innerH{\tangentS[^0]}{p,\Tr\Gb[\Vb]} $ as a genuine flux potential, where $ p\in\tangentS[^0] $ is not a degree of freedom as for the Lagrange multiplier technique. Moreover $ p $  serves as a pressure and could depend on $ \Vb $, $ \Qb $ or other additional degrees of freedom. In any case, this pressure has to be determined in some way in order to close the equations. A more profound investigation to develop a compressible model according to the Lagrange-D'Alembert principle is needed. 
The last extension concerns the incorporation of active terms. This is current research and will be published elsewhere.

\appendix

\section{Toolbox} \label{app:toolbox}

\begin{lemma}\label{lem:geo_qwertz01}
    For a parameterization $ \para $ and the metric tensor  given by covariant proxy components $ g_{ij} = \inner{\tangentR}{\partial_i\para,\partial_j\para} $
    holds
    \begin{align*}
        \partial_j\left( g^{ij}\partial_i\para \right)
            &= \meanc\normal - \Gamma_{jk}^jg^{ik}\partial_i\para \formPeriod
    \end{align*}
\end{lemma}
\begin{proof}
    Firstly, we notice that
    \begin{align}\label{eq:geo_qwerts01_tmp1}
        \partial_i g_{jk}
            &= \Gamma_{ijk} + \Gamma_{ikj}
    \end{align}
    holds as a consequence of the definition of Christoffel symbols.
    For the inverse metric we obtain
    \begin{align*}
        \partial_i g^{jk}
            &= \partial_i \left( g^{jl} g^{km} g_{lm} \right)
             = \delta_l^k \partial_i g^{jl} + \delta_m^j \partial_i g^{km} + g^{jl} g^{km} \partial_i g_{lm}
             = 2 \partial_i g^{jk} + g^{jl} g^{km} \partial_i g_{lm} \formPeriod
    \end{align*}
    Therefore \eqref{eq:geo_qwerts01_tmp1} yields
    \begin{align}\label{eq:geo_qwerts01_tmp2}
        \partial_i g^{jk}
            &= -\left( g^{jl}\Gamma_{il}^k +  g^{kl}\Gamma_{il}^j \right) \formPeriod
    \end{align}
    Since $ \Gamma_{ijk} = \inner{\tangentR}{\partial_i\partial_j\para,\partial_k\para} $ and $ \shopC_{ij} = \inner{\tangentR}{\partial_i\partial_j\para,\normal} $
    is valid, it holds
    \begin{align}\label{eq:geo_qwerts01_tmp3}
        \partial_i \partial_j \para
            &= \Gamma_{ij}^k \partial_k\para + \shopC_{ij}\normal \formPeriod
    \end{align}
    Finally, \eqref{eq:geo_qwerts01_tmp2} and \eqref{eq:geo_qwerts01_tmp3} yield
    \begin{align*}
         \partial_j\left( g^{ij}\partial_i\para \right)
            &=  g^{ij} \left(  \Gamma_{ij}^k \partial_k\para + \shopC_{ij}\normal \right)
                - \left( g^{ik}\Gamma_{jk}^j +  g^{jk}\Gamma_{jk}^i \right) \partial_i\para
             = \meanc\normal - \Gamma_{jk}^jg^{ik}\partial_i\para \formPeriod
    \end{align*}
    Alternatively, we could also use that $ \Delta\paraC^A = \meanc\normalC^A $, 
    where  $ \Delta\paraC^A = \partial_j ( g^{ij}\partial_i\paraC^A ) + \Gamma_{jk}^jg^{ik}\partial_i\paraC^A $ is the Laplace-Beltrami operator  
    on the scalar fields $ \paraC^A\in\tangentS[^0] $.
\end{proof}

\begin{lemma}\label{lem:scalar_veloder_of_timeD}
    For scalar fields $ f=f[\para]\in\tangentS[^0] $ with $ \veloder{\Vb}f =0 $, \ie\ $ f $ does not depend on the velocity $ \Vb = \partial_t\para $,
    holds
    \begin{align*}
        \veloder{\Wb} \dot{f} &= \gauge{\Wb}f
    \end{align*}
    for all $ \Wb\in\tangentR $.
\end{lemma}
\begin{proof}
    A Taylor expansion for $ \para(t+\tau,y^1,y^2) $ at $ \tau=0 $ yields
    \begin{align*}
        f[\para](t+\tau,y^1,y^2)
            &= f[\para\vert_{t} + \tau(\Vb\vert_{t} + \landau(\tau))](t+\tau,y^1,y^2)
             = f[\para\vert_{t}](t+\tau,y^1,y^2) + \tau\gauge{\Vb\vert_{t}+\landau(\tau)} f[\para\vert_{t}](t+\tau,y^1,y^2)\formPeriod
    \end{align*}
    Therefore, the time derivative of $ f $ becomes
    \begin{align}\label{eq:scalar_dot_chain_rule}
        \dot{f}[\para,\Vb](t,y^1,y^2)
            &= \lim_{\tau\rightarrow 0} \frac{f[\para](t+\tau,y^1,y^2) - f[\para](t,y^1,y^2)}{\tau}
             = \dot{f}[\para\vert_{t}](t,y^1,y^2) +\gauge{\Vb}f[\para](t,y^1,y^2) \formComma
    \end{align}
    where we cancel out the restriction $ \vert_{t} $ in the second summand, since the term is temporally locally defined.
    This does not apply to the first summand.
    The term $ \dot{f}[\para\vert_{t}] $ is to be understood as the temporal rate of $ f $, which is not induced by $ \para $.
    Poorly written, \eqref{eq:scalar_dot_chain_rule} reads as the chain rule ``$ \frac{d}{dt}f(t,\para(t))=\partial_t f + (\partial_t \para)\cdot\partial_{\para} f $''
    for partial dependencies.
    Since the first term of \eqref{eq:scalar_dot_chain_rule} does not depend on $ \Vb $, due to the temporal fixation of $ \para $
    and linearity of the deformation derivative in its direction argument, we obtain the assertion.
\end{proof}

\begin{lemma}\label{lem:ttensor_veloder_of_timeD}
    For 2-tensor fields $ \Rb=\Rb[\para]\in\tangentS[^2] $ with $ \veloder{\Wb}\Rb = \nullb $ for all $ \Wb\in\tangentR $, 
    \ie\ $ \Rb $ does not depend on the velocity $ \Vb = \partial_t\para $,
    holds
    \begin{align*}
        \veloder{\Wb} \Dphi\Rb &= \gauge{\Wb}\Rb - \Phib[\Wb]\Rb - \Rb\Phib^T[\Wb]
    \end{align*}
    for all $ \Wb\in\tangentR $, where $ \Dphi\Rb = \Dmat\Rb - \Phib[\Vb]\Rb - \Rb\Phib^T[\Vb] $ is one of the time derivatives in \cite{NitschkeVoigt_2023}.
\end{lemma}
\begin{proof}
    With lemma \ref{lem:scalar_veloder_of_timeD} the material derivative yields
    \begin{align*}
        \veloder{\Wb} \Dmat\Rb
            &= ( \veloder{\Wb} \dot{R}^{AB} ) \eb_A\otimes\eb_B
             = ( \gauge{\Wb} R^{AB}) \eb_A\otimes\eb_B
             =  \gauge{\Wb} \Rb \formPeriod
    \end{align*}
    Since $ \Phib[\Vb] $  is linear in $ \Vb $, it holds
    \begin{align*}
        \veloder{\Wb} \Dphi\Rb
            &= \gauge{\Wb} \Rb -  \veloder{\Wb}\left( \Phib[\Vb]\Rb + \Rb\Phib^T[\Vb] \right)
             = \gauge{\Wb}\Rb - \Phib[\Wb]\Rb - \Rb\Phib^T[\Wb] \formPeriod
    \end{align*}
\end{proof}

\begin{lemma}\label{lem:ttensor_variational_independence_of_process_variable}
    For all 2-tensor fields $ \Rb\in\tangentR[^2] $ and time derivatives $ \Dt[\Phi],\Dt[\widetilde{\Phi}]:\tangentR[^2]\rightarrow\tangentR $
    given in \cite{NitschkeVoigt_2023}, 
    let the two functionals $ \functional[\Dt[\Phi]\Rb] $ and $  \widetilde{\functional}[\Dt[\widetilde{\Phi}]\Rb] $
    be semantically equal,
    \ie\ it is $ \functional[\Dt[\Phi]\Rb]= \widetilde{\functional}[\Dt[\widetilde{\Phi}]\Rb] $ valid.
    For all $ \Psib\in\tangentR[^2] $ holds
    \begin{align*}
        \innerH{\tangentR[^2]}{\frac{\delta\functional}{\delta\Dt[\Phi]\Rb} , \Psib}
            &= \innerH{\tangentR[^2]}{\frac{\delta\widetilde{\functional}}{\delta\Dt[\widetilde{\Phi}]\Rb} , \Psib} \formPeriod
    \end{align*}
    This means verbalized that the variation of a functional, \wrt\ to a process variable for $ \Rb $, does not depend on a certain choice of this process variable.
\end{lemma}
\begin{proof}
Since all time derivatives on 2-tensor fields, which are considered in \cite{NitschkeVoigt_2023}, are relatable by 
$ \Dt[\widetilde{\Phi}]\Rb = \Dt[\Phi]\Rb + (\Phib-\widetilde{\Phib})\Rb + \Rb(\Phib-\widetilde{\Phib})^T $
for certain $ \Phib, \widetilde{\Phib} \in\tangentR$ depending on $ \Gbcal[\Vb] $, we calculate
\begin{align*}
    \innerH{\tangentR[^2]}{\frac{\delta\functional}{\delta\Dt[\Phi]\Rb} , \Psib} \hspace{-15pt}
        &= \innerH{\tangentR[^2]}{\frac{\delta\widetilde{\functional}}{\delta\Dt[\Phi]\Rb} , \Psib}\\
        &= \lim_{\eps\rightarrow 0} \frac{\widetilde{\functional}[\Dt[\Phi]\Rb + \eps\Psib +(\Phib-\widetilde{\Phib})\Rb + \Rb(\Phib-\widetilde{\Phib})^T]
                                        - 	\widetilde{\functional}[\Dt[\Phi]\Rb +(\Phib-\widetilde{\Phib})\Rb + \Rb(\Phib-\widetilde{\Phib})^T]}{\eps}\\
        &= \lim_{\eps\rightarrow 0} \frac{\widetilde{\functional}[\Dt[\widetilde{\Phi}]\Rb + \eps\Psib]- \widetilde{\functional}[\Dt[\widetilde{\Phi}]\Rb]}{\eps}
         = \innerH{\tangentR[^2]}{\frac{\delta\widetilde{\functional}}{\delta\Dt[\widetilde{\Phi}]\Rb} , \Psib} \formPeriod
\end{align*}
\end{proof}

\begin{lemma}\label{lem:ttensor_commutator_gauge_nablaS}
    The commutator of the deformation and componentwise derivative yields
    \begin{align*}
        \left[ \gauge{\Wb}, \nablaC \right]\Rb
            &:=  \gauge{\Wb} \nablaC \Rb - \nablaC \gauge{\Wb}\Rb
             = -(\nablaC\Rb)\Gbcal[\Wb]
             = -(\nablaC\Rb)\left( \nablaC\Wb - \normal(\nablaC\Wb)\otimes\normal \right)
    \end{align*}
    for all 2-tensor fields $ \Rb=\Rb[\para]\in\tangentS[^2] $ and vector fields $ \Wb\in\tangentR $.
\end{lemma}
\begin{proof}
    The pure Cartesian proxy of
    $ \nablaC\Rb = g^{ij}(\partial_j R^{AB}) \eb_A\otimes\eb_B\otimes\partial_i \para  $
    yields
    \begin{align*}
        [\gauge{\Wb} \nablaC \Rb[\para]]^{ABC}
            &= \gauge{\Wb}\left( g^{ij}[\para] \partial_j R^{AB}[\para] \partial_i \paraC^{C} \right) \\
            &= g^{ij}(\partial_j \gauge{\Wb}R^{AB}[\para]) \partial_i \paraC^C
                + (\gauge{\Wb}g^{ij}[\para] ) (\partial_j R^{AB}) \partial_i \paraC^C
                +  g^{ij}(\partial_j R^{AB}) \partial_i W^{C} \formComma
    \end{align*}
    since $ \gauge{\Wb}\eb_A=0 $, $  \gauge{\Wb}\paraC^A= W^{A}  $ and $ \left[ \gauge{\Wb}, \partial_i \right]\vert_{\tangentS[^0]} = 0 $.
    The first summand is just the proxy of $ \nablaC \gauge{\Wb}\Rb $.
    In the second summand, the deformation derivative of the inverse metric tensor proxy results in 
    $ \gauge{\Wb}g^{ij}[\para] = -(G^{ij}[\Wb] + G^{ji}[\Wb]) $, see \cite{NitschkeSadikVoigt_A_2022}.
    The last summand comprises the componentwise derivative of $ \Wb $, 
    which gives $ \partial_i W^{C} = \tensor{[\nablaC\Wb]}{^C_i} = \tensor{G}{^{C}_i}[\Wb] + \normalC^C b_{i}[\Wb] $ with $\bb[\Wb]=\normal\nablaC\Wb\in\tangentS$.
    Eventually, we obtain
    \begin{align*}
        \gauge{\Wb} \nablaC \Rb[\para]
            &= \nablaC \gauge{\Wb}\Rb[\para]
                - \nablaC\Rb\left( \Gb[\Wb] +  \Gb^T[\Wb] \right)
                + \nablaC\Rb\left( \Gb^T[\Wb] + \bb[\Wb]\otimes\normal \right)\\
            &= \nablaC \gauge{\Wb}\Rb[\para] - \nablaC\Rb\left( \Gb[\Wb] - \bb[\Wb]\otimes\normal \right) \formComma
    \end{align*}
    which in turn gives the assertion, since $ \Gbcal[\Wb] = \Gb[\Wb] + \normal\otimes\bb[\Wb] - \bb[\Wb]\otimes\normal $ \eqref{eq:Gbcal} and $ (\nablaC\Rb)\normal = 0 $.
\end{proof}

\begin{lemma}\label{lem:ttensor_weak_skewsymmetric_deformation}
    For all deformation directions $ \Wb=\wb + \wnor\normal\in\tangentR $, $ \wb\in\tangentS $, $ \wnor\in\tangentS[^0] $ and the skew-symmetric deformation gradient $ \Abcal[\Wb] \in \tangentAR $ \eqref{eq:Abcal} holds
    \begin{align*}
        \forall \Rb\in\tangentR[^2] &:
            &\inner{\tangentR[^2]}{\Rb, \Abcal[\Wb]}
                &= \inner{\tangentR\otimes\tangentS}{\left( \Id + \normal\otimes\normal \right) (\proj_{\tangentAR}\Rb) \IdS , \nablaC\Wb } \formComma\\
        \forall \tilde{\Rb} \in \tangentAR &:
            &\inner{\tangentAR}{\tilde{\Rb}, \Abcal[\Wb]}
                &= -\left( \phi\rot\wb + 2\nabla_{\etab}\wnor + 2\shop(\etab, \wb)\right) \formComma
    \end{align*}
    where $ \tilde{\Rb} = \phi\Eb + \etab\otimes\normal - \normal\otimes\etab $, $ \phi\in\tangentS[^0] $ and $ \etab\in\tangentS $. 
\end{lemma}
\begin{proof}
    Let $\tilde{\Rb}:=\proj_{\tangentAR}\Rb$. 
    For the tangential skew-symmetric surface deformation  gradient $ \Ab[\Wb] = \proj_{\tangentAS}\nablaC\Wb \in\tangentAS $ \eqref{eq:Ab}, this yields
    \begin{align*}
        \inner{\tangentR[^2]}{\Rb, \Ab[\Wb]}
            &= \inner{\tangentAS}{\projS[^2]\tilde{\Rb}, \Ab[\Wb]}
             = \inner{\tangentR\otimes\tangentS}{\projS[^2]\tilde{\Rb}, \nablaC\Wb} \formPeriod
    \end{align*}
    For the remaining part $ \Abcal[\Wb]-\Ab[\Wb] = 2\proj_{\tangentAR}(\normal\otimes\normal\nablaC\Wb) $, we get
    \begin{align*}
        2\inner{\tangentR[^2]}{\Rb, \proj_{\tangentAR}(\normal\otimes\normal\nablaC\Wb)}
            &= 2\inner{\tangentAR}{\tilde{\Rb}\IdS, \normal\otimes\normal\nablaC\Wb}
             = 2\inner{\tangentR\otimes\tangentS}{\normal\otimes\normal\tilde{\Rb}\IdS, \nablaC\Wb} \formPeriod
    \end{align*}
    It holds $ \projS[^2]\tilde{\Rb} + \normal\otimes\normal\tilde{\Rb}\IdS = \tilde{\Rb}\IdS $, which gives the first assertion.
    Moreover, substituting $ \nablaC\Wb = \nabla\wb - \wnor\shop + \normal\otimes(\nabla\wnor+\shop\wb) $ and the uniquely given decomposition
    $ \tilde{\Rb} = \phi\Eb + \etab\otimes\normal - \normal\otimes\etab $ into this yields
    \begin{align*}
        \inner{\tangentAR}{\tilde{\Rb}, \Abcal[\Wb]}
            &=  \inner{\tangentR\otimes\tangentS}{ \phi\Eb - 2\normal\otimes\etab, \nabla\wb - \wnor\shop + \normal\otimes(\nabla\wnor+\shop\wb) }\\
            &= \phi\inner{\tangentS[^2]}{\Eb,\nabla\wb} - 2 \inner{\tangentS}{\etab, \nabla\wnor+\shop\wb} \formPeriod
    \end{align*}
    The curl $  \rot{\wb} = -\inner{\tangentS[^2]}{\Eb,\nabla\wb} \in\tangentS[^0] $ gives the second assertion.
\end{proof}

\begin{corollary}\label{cor:jaumann_force_weak}
    Assuming the Jaumann gauge of surface independence for all $ \Qb\in\tangentQR $, \ie\ it is $ \gauge{\Wb}\Qb - \Abcal[\Wb]\Qb + \Qb\Abcal[\Wb] = 0 $ valid 
    for all deformation directions $ \Wb\in\tangentR $, it holds
    \begin{align*}
        \inner{\tangentR[^2]}{\Rb,  \gauge{\Wb}\Qb }
                &= \inner{\tangentR\otimes\tangentS}{\left( \Id + \normal\otimes\normal \right) ((\proj_{\tangentQR}\Rb)\Qb - \Qb(\proj_{\tangentQR}\Rb)) \IdS , \nablaC\Wb } 
    \end{align*}
    for all $ \Rb\in\tangentR[^2] $.
\end{corollary}
\begin{proof}
    Follows from $ \inner{\tangentR[^2]}{\Rb,  \gauge{\Wb}\Qb } = \inner{\tangentQR}{ \proj_{\tangentQR} \Rb,  \gauge{\Wb}\Qb } $  and lemma \ref{lem:ttensor_weak_skewsymmetric_deformation}.
\end{proof}

\begin{lemma}\label{lem:ttensor_weak_deformation}
    For all  2-tensor fields $ \Rb\in\tangentR[^2] $ and deformation directions $ \Wb\in\tangentR $, with deformation gradient $ \Gbcal[\Wb] \in \tangentR[^2] $ \eqref{eq:Gbcal}, holds
    \begin{align*}
        \innerH{\tangentR[^2]}{\Rb,\Gbcal[\Wb]}
            &=\innerH{\tangentR\otimes\tangentS}{ \Rb\IdS - \normal\otimes\IdS\Rb\normal , \nablaC\Wb}\\
            &= \innerH{\tangentR\otimes\tangentS}{\projS[^2]\Rb + \normal\otimes\left( \normal\Rb\IdS - \IdS\Rb\normal\right) , \nablaC\Wb}\formPeriod
    \end{align*}
\end{lemma}
\begin{proof}
    The deformation gradient \eqref{eq:Gbcal} yields
    \begin{align*}
        \innerH{\tangentR[^2]}{\Rb,\Gbcal[\Wb]}
            &= \innerH{\tangentR[^2]}{\Rb , \nablaC\Wb - \normal\nablaC\Wb\otimes\normal}
             =  \innerH{\tangentR\otimes\tangentS}{ \Rb\IdS - \normal\otimes\IdS\Rb\normal , \nablaC\Wb}\formPeriod
    \end{align*}
    Using the orthogonal decomposition $ \tangentR[^2]=\tangentSymR\oplus\tangentAR $ results in
    \begin{align*}
        \Rb\IdS - \normal\otimes\IdS\Rb\normal
            &= \Rb\IdS - \normal\otimes\IdS\left( \proj_{\tangentSymR}\Rb + \proj_{\tangentAR}\Rb \right)\normal\\
            &= \Rb\IdS - \normal\otimes\normal\left( \proj_{\tangentSymR}\Rb - \proj_{\tangentAR}\Rb \right)\IdS\\
            &= \left( \Id-\normal\otimes\normal \right)\Rb\IdS + 2 \normal\otimes\normal(\proj_{\tangentAR}\Rb)\IdS \formComma
    \end{align*}
    which gives with $ \IdS=\Id-\normal\otimes\normal $ the assertion.
\end{proof}

\begin{lemma}\label{lem:qtensorpart_of_tangentialtensor_is_conforming}
    The Q-tensor part of all tangential 2-tensor fields $ \rb\in\tangentS[^2] $ is surface conforming and yields
    \begin{align*}
        \proj_{\tangentQR}\rb 
            &= \CQdepl\left(\projQS\rb,-\frac{\Tr\rb}{3}\right)
             = \projQS\rb - \frac{\Tr\rb}{3}\left( \normal\otimes\normal - \frac{1}{2}\IdS \right)
             \in\tangentCQR \formComma 
    \end{align*}
    where $ \CQdepl: \tangentQS \times \tangentS[^0] \rightarrow \tangentCQR $ is given in \eqref{eq:surface_conforming_ansatz}.
\end{lemma}
\begin{proof}
    Since $ \IdS+\normal\otimes\normal = \Id $, we calculate
    \begin{align*}
        \CQdepl\left(\projQS\rb,-\frac{\Tr\rb}{3}\right)
            &= \frac{1}{2}\left( \rb + \rb^T - (\Tr\rb)\IdS\right) - \frac{\Tr\rb}{3}\left( \normal\otimes\normal - \frac{1}{2}\IdS \right)\\
            &= \frac{1}{2}\left( \rb + \rb^T\right) - \frac{\Tr\rb}{3} \Id
             = \proj_{\tangentQR}\rb \formPeriod
    \end{align*}
\end{proof}

\begin{lemma}\label{lem:trace_of_some_Q_powers}
    For all Q-tensor fields $ \Qb\in\tangentQR $ holds
    \begin{align*}
        \Tr\Qb^4 
            &= \frac{1}{2}(\Tr\Qb^2)^2 \formComma\\
        \Tr\Qb^6
            &=  \frac{1}{4}(\Tr\Qb^2)^3 + \frac{1}{3}(\Tr\Qb^3)^2 \formComma\\
        \Tr\Qb^8
            &= \frac{1}{8}(\Tr\Qb^2)^4 + \frac{4}{9}(\Tr\Qb^2)(\Tr\Qb^3)^2 \formPeriod
    \end{align*}
\end{lemma}
\begin{proof}
    Let be $ \lambda_1,\lambda_2,\lambda_3\in\tangentS[^0] $ the eigenvalue fields of $ \Qb\in\tangentQR<\tangentSymR $, 
    where, \oeda, $\lambda_3 = -(\lambda_1+\lambda_2)$ holds, since $ \Tr\Qb=\sum_{\alpha=1}^{3} \lambda_\alpha =0 $ is valid.
    For a better clarity we collect eigenvalue monomials symmetrically by $ a_{n,m} := \lambda_1^n\lambda_2^m + \lambda_2^n\lambda_1^m $.
    The trace of $ n $th power $ \Tr\Qb^n = \lambda_1^n + \lambda_2^n + (-1)^n(\lambda_1+\lambda_2)^n $ results in
    \begin{align*}
        \Tr\Qb^2 &= 2a_{2,0} + a_{1,1} \formComma
            & (\Tr\Qb^2)^2 &= 4 a_{4,0} + 8a_{3,1} + 6a_{2,2} \formComma \\
        (\Tr\Qb^2)^3 &=  8a_{6,0} + 24a_{5,1}+48a_{4,2}+28a_{3,3} \formComma
            & (\Tr\Qb^2)^4 &= 16a_{8,0} + 64a_{7,1} + 160a_{6,2} + 256a_{5,3} + 152a_{4,4} \formComma\\
         \Tr\Qb^3 &= -3a_{2,1} \formComma
            & (\Tr\Qb^3)^2 &= 9a_{4,2} + 9a_{3,3} \formComma \\
         (\Tr\Qb^2)(\Tr\Qb^3)^2 &= 18a_{6,2} + 54a_{5,3} + 36a_{4,4} \formPeriod
    \end{align*}
    Moreover, we obtain
    \begin{align*}
        \Tr\Qb^4 
            &=  2 a_{4,0} + 4a_{3,1} + 3a_{2,2} 
             =  \frac{1}{2}\left(4 a_{4,0} + 8a_{3,1} + 6a_{2,2}\right) \formComma \\
        \Tr\Qb^6
            &= 2a_{6,0} + 6a_{5,1} + 15a_{4,2} + 10a_{3,3}
             = \frac{1}{4} \left( 8a_{6,0} + 24a_{5,1}+48a_{4,2}+28a_{3,3} \right) + \frac{1}{3}\left( 9a_{4,2} + 9a_{3,3} \right) \formComma\\
        \Tr\Qb^8
            &= 2a_{8,0} + 8a_{7,1} + 28a_{6,2} + 56a_{5,3} + 35a_{4,4}\\
            &= \frac{1}{8}\left( 16a_{8,0} + 64a_{7,1} + 160a_{6,2} + 256a_{5,3} + 152a_{4,4} \right) + \frac{4}{9}\left( 18a_{6,2} + 54a_{5,3} + 36a_{4,4} \right) \formComma
    \end{align*}
    which gives the assertions.
\end{proof}

\begin{lemma} \label{lem:tensor_divC_to_DivC}
    The relation between \mbox{$ \hil $-adjoint} and trace componentwise divergence is
    \begin{align*}
        \divC\Rb 
            &= \DivC\Rb + \meanc\Rb\normal \in\tangentR[^{(n-1)}]
    \end{align*}
    for all $ n $-tensor fields $ \Rb\in\tangentR[^n] $, where $ \divC=-\nablaC^{*} $ is the \mbox{$ \hil $-adjoint} 
    and $ \DivC = \Tr\circ\nablaC $ the trace componentwise divergence.  
\end{lemma}
\begin{proof}
    Let be $ \Psib\in\tangentR[^{(n-1)}] $ and $ \Psib\vdotsn[n-1]\Rb := \Psi_{A_{1}\ldots A_{n-1}} R^{A_{1} \ldots A_{n}}\eb_{A_{n}}\in\tangentR $ 
    the left-sided contraction of $ \Psib $ with $ \Rb $.
    Lemma \ref{lem:geo_qwertz01} and the (covariant) divergence 
    $ \div\wb = -\nabla^{*}\wb = \Tr\nabla\wb = \partial_i w^i + \Gamma_{ik}^{i}w^k\in\tangentS[^0] $ for all tangential vector fields $ \wb\in\tangentS $
    yield
    \begin{align*}
        \inner{\tangentR[^n]}{ \Rb, \nablaC\Psib}
            &= R_{A_{1}\ldots A_{n}} \left( \partial_j \Psi^{A_{1}\ldots A_{n-1}} \right) g^{ij} \partial_i \paraC^{A_{n}} \\
            &=   \partial_j  \left( R_{A_{1}\ldots A_{n}}   \Psi^{A_{1}\ldots A_{n-1}} g^{ij} \partial_i \paraC^{A_{n}}  \right)
                            - R_{A_{1}\ldots A_{n}}   \Psi^{A_{1}\ldots A_{n-1}} \partial_j  \left( g^{ij} \partial_i \paraC^{A_{n}}  \right)\\
            &\quad   - \left( \partial_j R_{A_{1}\ldots A_{n}} \right) \Psi^{A_{1}\ldots A_{n-1}} g^{ij} \partial_i \paraC^{A_{n}}\\
            &= \partial_j \left[ \projS( \Psib\vdotsn[n-1]\Rb ) \right]^j + \Gamma_{jk}^j\left[ \projS( \Psib\vdotsn[n-1]\Rb ) \right]^k
               - \meanc \Psib\vdotsn[n-1]\Rb\normal - \inner{\tangentR[^{(n-1)}]}{\DivC\Rb, \Psib }\\
            &=  \div\projS( \Psib\vdotsn[n-1]\Rb )  - \inner{\tangentR[^{(n-1)}]}{\DivC\Rb + \meanc\Rb\normal, \Psib } \formPeriod
    \end{align*} 
    Eventually, Gauss's theorem and convenient boundary conditions reveal
    \begin{align*}
        \innerH{\tangentR[^n]}{ \Rb, \nablaC\Psib}
            &= - \innerH{\tangentR[^{(n-1)}]}{\DivC\Rb + \meanc\Rb\normal, \Psib }
             = \innerH{\tangentR[^{(n-1)}]}{\nablaC^{*}\Rb , \Psib }
    \end{align*}
    for all $ \Psib $.
\end{proof}

\begin{lemma}\label{lem:thinfilmgradient}
    Let be $ \hat{\Rb}\in\bulktangentR[^n]{\surf_h} $ a smooth extension of the $ n $-tensor field $ \Rb\in\tangentR[^n] $, 
    defined in a thin film $ \surf_h $ of thickness $ h>0 $ around a surface $ \surf $ without boundaries,
    \ie\ it holds $ \hat{\Rb}\vert_{\surf}=\Rb $ and $ \partial\surf=\emptyset $.
    If $\hat{\Rb}$ follows the homogeneous Neumann boundary condition, \ie\ it is $\nablahat_{\hat{\normal}}\hat{\Rb}\vert_{\partial\surf_h} = 0$,
    then the usual $ \R^3 $--derivative $ \nablahat $ yields
    \begin{align*}
        (\nablahat\hat{\Rb})\vert_{\surf} 
            &= \nablaC\Rb +  \landau(h^2) \formPeriod
    \end{align*} 
\end{lemma}
\begin{proof}
 Following \cite[Lemma A.7.]{Nitschke_2018} the Neumann boundary condition yields $\nablahat_{\hat{\normal}}\hat{\Rb}\vert_{\surf} = \landau(h^2)$.
 Therefore, decomposing the $ \R^3 $--derivative into a tangential and normal part gives
 \begin{align*}
    (\nablahat\hat{\Rb})\vert_{\surf}
        &=  (\nablahat\hat{\Rb})\vert_{\surf}\left( \IdS + \normal\otimes\normal \right)
        = \nablaC\Rb + ( \nablahat_{\hat{\normal}}\hat{\Rb})\vert_{\surf} \otimes\normal
        = \nablaC\Rb +  \landau(h^2) \formPeriod
 \end{align*}
\end{proof}

\noindent
{\textbf{Funding:}} A.V. was supported by DFG through FOR3013. \\
 
\noindent
{\textbf{Competing interests:}} There are no competing interests. \\
 
\noindent
{\textbf{Authors' contributions:}} This project was conceived by I.N. and A.V.; I.N. derived the theory, the results were analysed by I.N. and A.V. and the main text was written by I.N. and A.V..

\bibliography{bib}

\end{document}

%% file: symbols.tex
\newcommand{\splitcelltab}[1]{\begin{tabular}{@{}c@{}}#1\end{tabular}} 

\newcommand{\wrt}{w.\,r.\,t.}
\newcommand{\eg}{e.\,g.}
\newcommand{\st}{s.\,t.} 
\newcommand{\ie}{i.\,e.}

\newcommand{\oeda}{w.\,l.\,o.\,g.}
\newcommand{\etc}{etc.}
\newcommand{\resp}{resp.}
\newcommand{\cf}{cf.}

\newcommand{\formComma}{\,\text{,}}
\newcommand{\formPeriod}{\,\text{.}}

\newcommand{\surf}{\mathcal{S}}

\newcommand{\R}{\mathbb{R}}

\newcommand{\tangent}[2][]{\tensor{\operatorname{T}\!}{#1}#2}
\newcommand{\tangentS}[1][]{\tangent[#1]{\surf}}
\newcommand{\tangentR}[1][]{\tangent[#1]{\R^3\vert_{\surf}}}
\newcommand{\tangentQS}{\tensor{\operatorname{\mathcal{Q}}\!}{^2}\surf}
\newcommand{\tangentQR}{\tensor{\operatorname{\mathcal{Q}}\!}{^2}\R^3\vert_{\surf}}
\newcommand{\tangentCQR}{\tensor{\operatorname{C_{\surf}\mathcal{Q}}\!}{^2}\R^3\vert_{\surf}}
\newcommand{\tangentCQRbot}{\tensor{\operatorname{C_{\surf}^{\bot}\mathcal{Q}}\!}{^2}\R^3\vert_{\surf}}
\newcommand{\tangentAS}[1][^2]{\tensor{\operatorname{\mathcal{A}}\!}{#1}\surf}
\newcommand{\tangentAR}[1][^2]{\tensor{\operatorname{\mathcal{A}}\!}{#1}\R^3\vert_{\surf}}
\newcommand{\tangentSymS}[1][^2]{\tensor{\operatorname{Sym}\!}{#1}\surf}
\newcommand{\tangentSymR}[1][^2]{\tensor{\operatorname{Sym}\!}{#1}\R^3\vert_{\surf}}
\newcommand{\tangentnormal}[1][]{\tangent[#1]{\normal}}
\newcommand{\tangentIdR}{\tangent{\Id}}
\newcommand{\tangentIdS}{\tangent{\IdS}}

\newcommand{\bulktangentR}[2][]{\tangent[#1]{\R^3\vert_{#2}}}

\newcommand{\paraC}{X}
\newcommand{\para}{\boldsymbol{\paraC}}
\newcommand{\normalC}{\nu}
\newcommand{\normal}{\boldsymbol{\normalC}}
\newcommand{\shopC}{I\!I}
\newcommand{\shop}{\boldsymbol{\shopC}}
\newcommand{\meanc}{\mathcal{H}}
\newcommand{\gaussc}{\mathcal{K}}

\newcommand{\dS}{\textup{d}\surf}
\newcommand{\dt}{\textup{d}t}

\newcommand{\eps}{\varepsilon}

\newcommand{\landau}{\mathcal{O}}

\newcommand{\curlyd}{\eth}
\newcommand{\gauge}[2][]{\tensor*{\curlyd}{_{#2}^{#1}}}
\newcommand{\veloder}[2][]{\tensor*{\mathfrak{V}}{_{#2}^{#1}}}

\newcommand{\Comp}{\mathsf{C}}

\newcommand{\nablaC}{\nabla_{\!\Comp}}
\newcommand{\nablahat}{\hat{\nabla}}

\newcommand{\rot}{\operatorname{rot}}

\newcommand{\Tr}{\operatorname{Tr}}
\renewcommand{\div}{\operatorname{div}}
\newcommand{\Div}{\operatorname{Div}}

\newcommand{\divC}{\div_{\!\Comp}}
\newcommand{\DivC}{\Div_{\!\Comp}}

\newcommand{\Grad}{\operatorname{Grad}}
\newcommand{\GradC}{\Grad_{\Comp}}

\newcommand{\DeltaC}{\Delta_{\Comp}}

\newcommand{\proj}{\operatorname{\Pi}}
\newcommand{\projS}[1][]{\proj_{\tangentS[#1]}}
\newcommand{\projQS}{\proj_{\tangentQS}}
\newcommand{\projQR}{\proj_{\tangentQR}}

\newcommand{\jau}{\mathcal{J}}
\newcommand{\timeD}[1][\Phi]{\operatorname{d}_{t}^{#1}\!}
\newcommand{\timeJ}{\mathfrak{J}}
\newcommand{\timeLuu}{\mathfrak{L}^{\sharp\sharp}}
\newcommand{\timeLll}{\mathfrak{L}^{\flat\flat}}

\newcommand{\Dt}[1][]{\operatorname{D}^{#1}_t \!}
\newcommand{\Dmat}{\Dt[\mfrak]}
\newcommand{\Dupp}{\Dt[\sharp]}
\newcommand{\Dlow}{\Dt[\flat]}
\newcommand{\Djau}{\Dt[\jau]}
\newcommand{\Dphi}{\Dt[\Phi]}

\newcommand{\Qdepl}{\mathcal{Q}}
\newcommand{\CQdepl}{\Qdepl_{C_{\surf}}}

\newcommand{\functional}{\mathfrak{F}}
\newcommand{\potenergy}{\mathfrak{U}}
\newcommand{\energy}{\mathfrak{E}}
\newcommand{\fluxpotential}{\mathfrak{R}}

\newcommand{\totalFb}{\overline{\Fb}}
\newcommand{\gaugeFb}{\widetilde{\Fb}}

\newcommand{\energyK}{\energy_{\textup{K}}}

\newcommand{\Cset}{\mathcal{C}} 

\newcommand{\EL}{\textup{EL}}
\newcommand{\energyEL}{\potenergy_{\EL}} 
\newcommand{\HbEL}{\Hb_{\EL}}
\newcommand{\hbEL}{\hb_{\EL}}
\newcommand{\zetabEL}{\zetab_{\EL}}
\newcommand{\omegaEL}{\omega_{\EL}}
\newcommand{\SigmabEL}{\Sigmab_{\EL}}
\newcommand{\sigmabEL}{\sigmab_{\EL}}
\newcommand{\FbEL}{\Fb_{\EL}}

\newcommand{\THT}{\textup{TH}} 
\newcommand{\energyTH}{\potenergy_{\THT}} 
\newcommand{\HbTH}{\Hb_{\THT}}
\newcommand{\hbTH}{\hb_{\THT}}
\newcommand{\omegaTH}{\omega_{\THT}}
\newcommand{\SigmabTH}{\Sigmab_{\THT}}
\newcommand{\sigmabTH}{\sigmab_{\THT}}
\newcommand{\pTH}{p_{\THT}}
\newcommand{\FbTH}{\Fb_{\THT}}

\newcommand{\BE}{\textup{BE}}
\newcommand{\energyBE}{\potenergy_{\BE}} 
\newcommand{\HbBE}{\Hb_{\BE}}
\newcommand{\SigmabBE}{\Sigmab_{\BE}}
\newcommand{\FbBE}{\Fb_{\BE}}

\newcommand{\fnorBE}{\fnor[\BE]}

\newcommand{\coeffIF}{\upsilon}

\newcommand{\IM}{\textup{IM}}
\newcommand{\energyIM}[1][\Phi]{\fluxpotential_{\IM}^{#1}} 
\newcommand{\HbIM}[1][\Phi]{\Hb_{\IM}^{#1}}
\newcommand{\hbIM}[1][\Phi]{\hb_{\IM}^{#1}}
\newcommand{\zetabIM}[1][\Phi]{\zetab_{\IM}^{#1}}
\newcommand{\omegaIM}[1][\Phi]{\omega_{\IM}^{#1}}
\newcommand{\SigmabIM}[1][\Phi]{\Sigmab_{\IM}^{#1}}
\newcommand{\sigmabIM}[1][\Phi]{\sigmab_{\IM}^{#1}}
\newcommand{\FbIM}[1][\Phi]{\Fb_{\IM}^{#1}}
\newcommand{\gaugeFbIM}[2][\Phi]{\gaugeFb_{\IM}^{#1,#2}}

\newcommand{\AC}{\textup{AC}}

\newcommand{\NV}{\textup{NV}}
\newcommand{\energyNV}{\fluxpotential_{\NV}} 
\newcommand{\HbNV}{\Hb_{\NV}}
\newcommand{\hbNV}{\hb_{\NV}}
\newcommand{\omegaNV}{\omega_{\NV}}
\newcommand{\SigmabNV}{\Sigmab_{\NV}}
\newcommand{\sigmabNV}{\sigmab_{\NV}}

\newcommand{\SC}{\textup{SC}}
\newcommand{\energySC}{\potenergy_{\SC}}
\newcommand{\lambdabSC}{\lambdab_{\SC}} 
\newcommand{\HbSC}{\Hb_{\SC}}
\newcommand{\zetabSC}{\zetab_{SC}}
\newcommand{\SigmabSC}[1][]{\Sigmab_{\SC}^{#1}}
\newcommand{\SigmabSCPhi}[1][\Phi]{\SigmabSC[#1]}
\newcommand{\varsigmabSCPhi}[1][\Phi]{\varsigmab_{\SC}^{#1}}

\newcommand{\CB}{\textup{CB}}
\newcommand{\energyCB}{\potenergy_{\CB}}
\newcommand{\lambdaCB}{\lambda_{CB}} 
\newcommand{\HbCB}{\Hb_{\CB}}
\newcommand{\omegaCB}{\omega_{CB}}

\newcommand{\IC}{\textup{IC}}
\newcommand{\energyIC}{\fluxpotential_{\IC}} 
\newcommand{\SigmabIC}{\Sigmab_{\IC}}
\newcommand{\sigmabIC}{\sigmab_{\IC}}
\newcommand{\FbIC}{\Fb_{\IC}}

\newcommand{\UN}{\textup{UN}}
\newcommand{\energyUN}{\potenergy_{\UN}}
\newcommand{\LambdabUN}{\Lambdab_{\UN}}
\newcommand{\lambdabUN}{\lambdab_{\UN}}
\newcommand{\lbUN}{\lb_{\UN}} 
\newcommand{\lambdabotUN}{\lambda^{\bot}_{\UN}}

\newcommand{\HbUN}{\Hb_{\UN}}
\newcommand{\hbUN}{\hb_{\UN}}
\newcommand{\zetabUN}{\zetab_{\UN}}
\newcommand{\omegaUN}{\omega_{\UN}}
\newcommand{\FbUN}{\Fb_{\UN}}

\newcommand{\NN}{\textup{NN}}
\newcommand{\energyNN}{\fluxpotential_{\NN}}
\newcommand{\lambdaNN}{\lambda_{\NN}}
\newcommand{\FbNN}{\Fb_{\NN}}

\newcommand{\NF}{\textup{NF}}
\newcommand{\energyNF}{\fluxpotential_{\NF}}
\newcommand{\LambdabNF}{\Lambdab_{\NF}}
\newcommand{\FbNF}{\Fb_{\NF}}
\newcommand{\HbNF}{\Hb_{\NF}}

\newcommand{\IS}{\textup{IS}}
\newcommand{\energyIS}{\potenergy_{\IS}}
\newcommand{\LambdabIS}{\Lambdab_{\IS}}
\newcommand{\FbIS}{\Fb_{\IS}}
\newcommand{\HbIS}{\Hb_{\IS}}

\newcommand{\TOT}{\textup{Tot}}
\newcommand{\energyTOT}{\energy_{\TOT}}

\newcommand{\hil}{\operatorname{L}^{\!2}}
\newcommand{\hilspace}[1]{\hil(#1)}
\newcommand{\inner}[3][0]{\left\langle #3 \right\rangle_{\hspace{-#1pt}#2}}
\newcommand{\normsq}[3][0]{\left\| #3 \right\|_{\hspace{-#1pt}#2}^2}
\newcommand{\innerH}[3][0]{\inner[#1]{\hilspace{#2}}{#3}}
\newcommand{\normHsq}[3][0]{\normsq[#1]{\hilspace{#2}}{#3}}
\newcommand{\normHQRsq}[2][0]{\normHsq[#1]{\tangentQR}{#2}}

\newcommand{\gb}{\boldsymbol{g}}
\newcommand{\Wb}{\boldsymbol{W}}
\newcommand{\wb}{\boldsymbol{w}}
\newcommand{\wnor}{w_{\bot}}
\newcommand{\qb}{\boldsymbol{q}}
\newcommand{\Qb}{\boldsymbol{Q}}
\newcommand{\rb}{\boldsymbol{r}}
\newcommand{\Rb}{\boldsymbol{R}}
\newcommand{\eb}{\boldsymbol{e}}
\newcommand{\bb}{\boldsymbol{b}}
\newcommand{\Gb}{\boldsymbol{G}}
\newcommand{\Sb}{\boldsymbol{S}}
\newcommand{\Ab}{\boldsymbol{A}}
\newcommand{\Eb}{\boldsymbol{E}}

\newcommand{\Psib}{\boldsymbol{\Psi}}
\newcommand{\Phib}{\boldsymbol{\Phi}}
\newcommand{\Vb}{\boldsymbol{V}}
\newcommand{\vb}{\boldsymbol{v}}
\newcommand{\vnor}{v_{\bot}}
\newcommand{\lambdab}{\boldsymbol{\lambda}}
\newcommand{\Lambdab}{\boldsymbol{\Lambda}}
\newcommand{\ub}{\boldsymbol{u}}
\newcommand{\sigmab}{\boldsymbol{\sigma}}
\newcommand{\etab}{\boldsymbol{\eta}}
\newcommand{\nullb}{\boldsymbol{0}}
\newcommand{\pb}{\boldsymbol{p}}
\newcommand{\Pb}{\boldsymbol{P}}
\newcommand{\Sigmab}{\boldsymbol{\Sigma}}
\newcommand{\Hb}{\boldsymbol{H}}
\newcommand{\varsigmab}{\boldsymbol{\varsigma}}
\newcommand{\Fb}{\boldsymbol{F}}
\newcommand{\fb}{\boldsymbol{f}}
\newcommand{\fnor}[1][]{f^{\bot}_{#1}}

\newcommand{\hb}{\boldsymbol{h}}
\newcommand{\zetab}{\boldsymbol{\zeta}}
\newcommand{\thetab}{\boldsymbol{\theta}}
\newcommand{\xb}{\boldsymbol{x}}
\newcommand{\ab}{\boldsymbol{a}}
\newcommand{\anor}{a_{\bot}}
\newcommand{\Thetab}{\boldsymbol{\Theta}}
\newcommand{\lb}{\boldsymbol{l}}
\newcommand{\Cb}{\boldsymbol{C}}

\newcommand{\Gbcal}{\boldsymbol{\mathcal{G}}}
\newcommand{\Abcal}{\boldsymbol{\mathcal{A}}}

\newcommand{\Acal}{\mathcal{A}} 
\newcommand{\Tcal}{\mathcal{T}} 
\newcommand{\Vcal}{\mathcal{V}} 
\newcommand{\Lfrak}{\mathfrak{L}} 

\newcommand{\Id}{\boldsymbol{Id}}
\newcommand{\IdS}{\Id_{\surf}}
\newcommand{\IdSC}{Id_{\surf}}

\newcommand{\IbxiQ}{\boldsymbol{I}_{\xi}[\Qb]}

\newcommand{\mfrak}{\mathfrak{m}}
\newcommand{\ofrak}{\mathfrak{o}}

\newcommand{\deltafrac}[2]{\frac{\delta #1}{\delta #2}}
\newcommand{\ddfrac}[2][]{\frac{\textup{d} #1}{\textup{d} #2}}
\newcommand{\ddt}[1][]{\ddfrac[#1]{t}}

\newcommand{\vdotsn}[1][n]{\ \vdots_{#1}}
\newcommand{\dbdot}{\operatorname{:}}